\documentclass[%
reprint,
amsmath,amssymb,
aps,
]{revtex4-1}

\usepackage{graphicx}
\usepackage{dcolumn}
\usepackage{bm}

\usepackage{stmaryrd} 
\usepackage{amsthm}  
\usepackage{IEEEtrantools} 
\usepackage{dsfont} 

\usepackage{mathtools} 

\newtheorem{lem}{Lemma}
\newtheorem{theorem}{Theorem}

\newcommand{\Id}[1]{\boldsymbol{\mathbb{I #1}}}

\begin{document}

\preprint{APS/123-QED}

\title{Scattering matrix pole expansions \\ \& \\ invariance with respect to R-matrix parameters}


\author{Pablo Ducru}
\email{p\_ducru@mit.edu ; pablo.ducru@polytechnique.org}
\altaffiliation[Also from ]{\'Ecole Polytechnique, France.}
\affiliation{%
Massachusetts Institute of Technology\\
Department of Nuclear Science \& Engineering\\
77 Massachusetts Avenue, Cambridge, MA, 02139 U.S.A.\\
}%
\author{Vladimir Sobes}%
\email{sobesv@ornl.gov}
\affiliation{%
Oak Ridge National Laboratory\\
Directorate of Nuclear Science \& Engineering\\
1 Bethel Valley Road, Oak Ridge, TN, 37831 U.S.A.
}%

\author{Gerald Hale}
\email{ghale@lanl.gov}
\author{Mark Paris}%
\email{mparis@lanl.gov}
\affiliation{%
Los Alamos National Laboratory\\
Theoretical Division (T-2)\\
MS B283, Los Alamos, NM 87454 U.S.A.\\
}%

\author{Benoit Forget}
\email{bforget@mit.edu}
\affiliation{%
Massachusetts Institute of Technology\\
Department of Nuclear Science \& Engineering\\
77 Massachusetts Avenue, Cambridge, MA, 02139 U.S.A.\\
}%


%


\date{\today}

\begin{abstract}

A plethora of nuclear interactions have been modeled with R-matrix theory, upon which our current nuclear data libraries are based \cite{ENDF_VIII}.
The R-matrix scattering model can parametrize the energy dependence of the scattering matrix in different ways. This article studies and establishes new results on three such sets of parameters --- the Wigner-Eisenbud, the Brune, and the Siegert-Humblet parameters --- showing how the latter two arise from the invariance of the scattering matrix to the Wigner-Eisenbud boundary condition $B_c$.
We exhibit how these parametrizations strike an arbitrage between the complexity of their parameters, and the complexity of the energy dependence they entail for the scattering matrix. In this, our work bridges R-matrix theory with the well-known pole expansions of the scattering matrix developed by Humblet and Rosenfeld \cite{Theory_of_Nuclear_Reactions_I_resonances_Humblet_and_Rosenfeld_1961}.

With respect to invariance to the channel radius $a_c$, we establish for the first time that the scattering matrix's invariance to $a_c$ sets a partial differential equation on the widths of the Kapur-Peierls operator $\boldsymbol{R}_{L}$. This enables us to derive an explicit transformation of the Siegert-Humblet radioactive residue widths $\left\{r_{j,c}\right\}$ under a change of channel radius $a_c$.

Considering the continuation of the scattering matrix to complex energies, several new results are established.
We unveil there exist more Brune parameters than previously thought, depending on the way the scattering matrix is continued to complex energies.
This points to a broader conundrum in the field: how to analytically continue the scattering matrix $\boldsymbol{U}(E)$ while opening and closing channels at the threshold energies.
On this question, we show that the legacy of Lane \& Thomas to force-close channels below threshold -- by defining the shift $S_c$ and penetration $P_c$ functions as the real and imaginary parts of the outgoing wave function reduced logarithmic derivative $L_c$ in order to annul elements of the scattering matrix $U_{cc'}$ below threshold -- presents serious drawbacks when considering complex wavenumbers: this introduces nonphysical, spurious poles to the scattering matrix, while breaking its analytic properties. We show that the roots of the outgoing wave functions $O_c$ should introduce spurious poles to the scattering matrix, but that analytic continuation of the scattering matrix $U_{cc'}$ with constant Wronskian cancels them out.
Moreover, we show that analytic continuation of R-matrix operators enforces the generalized unitarity conditions described by Eden \& Taylor \cite{Eden_and_Taylor}.
Finally, in the case of massive particles, we propose a solution to the conundrum of how to close the channels below thresholds, by invoking both a quantum tunneling argument, whereby the transmission matrix is evanescent bellow threshold, and a physical argument based on the definition of the cross section as the ratio of probability currents.

\end{abstract}

\maketitle


\section{\label{sec:Introduction}Introduction}

Myriad scattering phenomena in nuclear physics are modeled with R-matrix theory, with applications to nuclear simulation, radiation transport, astrophysics and cosmology, and extending to particle physics or atomistic and molecular simulation \cite{Blatt_and_Weisskopf_Theoretical_Nuclear_Physics_1952}\cite{Lane_and_Thomas_1958}\cite{atomistic_R_matrix_2011}\cite{ Computational_methods_in_R-matrix-2007, polyatomic-molecules-using-R-matrix-Royal-Society,  Burke-R-matrix-theory-of-atomic-processes, electron-molecule-scattering-R-matrix-Burke-1979, QB-method_R-matrix_1996}.
In essence, a scattering event takes different incoming particle-waves and lets them interact through a given Hamiltonian to produce different possible outcomes.
R-matrix theory studies the particular two-body-in/two-body-out model of this scattering event, with the fundamental assumption that the total Hamiltonian is the superposition of a short-range, interior Hamiltonian, which is null after a given channel radius $a_c$, and a long-range, exterior Hamiltonian which is well known (free particles or Coulomb potential, for instance).
This partitioning, along with an orthogonality assumption of channels at the channel boundary, is what we could call the \textit{R-matrix scattering model}, described by Kapur and Peierls in their seminal article \cite{Kapur_and_Peierls_1938}, unified by Bloch in \cite{Bloch_1957}, and reviewed by Lane and Thomas in \cite{Lane_and_Thomas_1958}.
It is possible to study the general energy and wavenumber dependence of the scattering matrix emerging from this model, and such expansions were thoroughly studied by Humblet and Rosenfeld in \cite{Theory_of_Nuclear_Reactions_I_resonances_Humblet_and_Rosenfeld_1961,Theory_of_Nuclear_Reactions_II_optical_model_Rosenfeld_1961,Theory_of_Nuclear_Reactions_III_Channel_radii_Humblet_1961_channel_Radii,Theory_of_Nuclear_Reactions_IV_Coulomb_Humblet_1964,Theory_of_Nuclear_Reactions_V_low_energy_penetrations_Jeukenne_1965,Theory_of_Nuclear_Reactions_VI_unitarity_Humblet_1964,Theory_of_Nuclear_Reactions_VII_Photons_Mahaux_1965,Theory_of_Nuclear_Reactions_VIII_evolutions_Rosenfeld_1965,Theory_of_Nuclear_Reactions_IX_few_levels_approx_Mahaux_1965}. Alternatively, this wave-number dependence of the different possible outgoing particle-waves stemming from the R-matrix model can be parametrized, for calculablility reasons, in several ways.

Over the years, different such parametrizations have been proposed for this same R-matrix model. The one that has come to prevail in the nuclear physics community is the Wigner-Eisenbud parametrization \cite{Wigner_and_Eisenbud_1947, Bloch_1957, Lane_and_Thomas_1958}.
There are good reasons for this: the Wigner-Eisenbud parameters are unconstrained real parameters --- i.e. though physically and statistically correlated, any set of real parameters is physically acceptable (though not necessarily likely nor present in nature) ---
that separate and parametrize the interior interaction Hamiltonian (usually an intractable many-body nuclear problem) from the exterior one (usually a well-known free-body or Coulomb Hamiltonian with analytic Harmonic expansions).
Thus, Wigner and Eisenbud constructed a parametrization of the scattering matrix for calculability purposes: introducing simple real parameters that help de-correlate what happens in the inner interaction region from the asymptotic outer region.

But there exists other parametrizations of the R-matrix model, all with their advantages and disadvantages for interpretability or for subsequent treatment.
In particular, we shall here focus on three such parametrizations: the Wigner-Eisenbud (which we will refer to as the $\boldsymbol{R}_B$ for reasons that will become evident)\cite{Wigner_and_Eisenbud_1947}; the Brune (which we will henceforth refer to as $\boldsymbol{R}_S$)\cite{Brune_2002}; and the Siegert-Humblet (henceforth $\boldsymbol{R}_{L}$)\cite{Siegert, Humblet_thesis}.
Since they all parametrize the wave-number dependence of the same R-matrix scattering model, theses three parameteric representations are physically equivalent, and one can go from one set of parameters to another through mathematical transformations we will hereafter explicit. For the same reasons, these parametrizations must also all be physically equivalent to the scattering matrix expansions of Humblet and Rosenfeld.

This article focuses on various properties of the Wigner-Eisenbud, Brune, and Siegert-Humblet respective parametrizations of the R-matrix model.
We bring new insights to the invariance properties of the three different parametrizations and show that they entail a tradeoff between the complexity of the parameter space and the simplicity of the energy (or wave-number) dependence.
Although invariance to the channel boundary condition $B_c$ has been well studied and is at the core of Brune's alternative parametrization, we here unveil new properties of the Brune transform, in theorems \ref{theo::shadow_Brune_poles}, \ref{theo::analytic_Brune_poles}, and \ref{theo::Choice of Brune poles}, as well as of the poles of the scattering matrix, invariant to the representation chosen. In this, we explicitly bridge the R-matrix parameters to the Humblet-Rosenfeld expansions of the scattering matrix, though equations (\ref{eq::U Mittag Leffler}) and (\ref{eq::u_j scattering residue width}).
We also investigate invariance with respect to the channel radius $a_c$, which is a much less thoroughly explored topic in the literature, leading us to new insights with respect to the residues of the Siegert-Humblet $\boldsymbol{R}_{L}$ parametrization, established in theorems \ref{theo::r_j,c uncer change of a_c} and \ref{theo::r_j,c uncer change of a_c explicit}.
Finally, we argue, through theorems \ref{theo::Analytic continuation of scattering matrix cancels spurious poles}, \ref{theo:: poles of U are Siegert-Humblet poles}, \ref{theo::satisfied generalized unitarity condition}, \ref{theo::evanescence of sub-threshold transmission matrix}, and \ref{theo::Analytic continuation annuls sub-threshold cross sections}, that contrary to what Lane \& Thomas prescribed \cite{Lane_and_Thomas_1958}, the R-matrix parametrization should be analytically continued to complex wave-numbers $k_c\in\mathbb{C}$.

\section{\label{sec:R-matrix equations}R-matrix Wigner-Eisenbud parametrization}

We here recall some fundamental definitions and equations of the Wigner-Eisenbud R-matrix parameters \cite{Wigner_and_Eisenbud_1947, Bloch_1957, Lane_and_Thomas_1958}.
As described in Lane \& Thomas, for each channel $c$, the two-body-in/two-body-out R-matrix model allows one to reduce the many-body system into a reduced one-body system.
All the study is then performed in the reduced system and we consider the wave-number of each channel $k_c$, which we can render dimensionless using the channel radius $a_c$ and defining $\boldsymbol{\rho}=\boldsymbol{\mathrm{diag}}\left(\rho_c\right)$ with $\rho_c = k_c a_c$.

\subsection{\label{subsec:Energy dependence and wavenumber mapping} Energy dependence and wavenumber mapping}

All of the channel wavenumbers link back to one unique total system energy $E$, eigenvalue of the total Hamiltonian.
Conservation of energy entails that this energy $E$ must be the total energy of any given channel $c$ (c.f. equation (5.12), p.557 of \cite{Theory_of_Nuclear_Reactions_I_resonances_Humblet_and_Rosenfeld_1961}):
\begin{equation}
\begin{IEEEeqnarraybox}[][c]{rcl}
E  & \ = \ & E_c = E_{c'} = \hdots \; \; , \; \forall \; c
\IEEEstrut\end{IEEEeqnarraybox}
\label{eq:conservation of energy E = E_c = E_c'}
\end{equation}
Each channel's total energy $E_c$ is then linked to the wavenumber $k_c$ of the channel by its corresponding relation (\ref{eq:rho_c(E) mapping}), say (\ref{eq:rho_c EDA}) and (\ref{eq:E_c as a function of s_c}).

In the semi-classical model described in Lane \& Thomas \cite{Lane_and_Thomas_1958}, we can separate on the one hand massive particles, for which the wavenumber $k_c$ is related to the center-of-mass energy $E_c$ of relative motion of channel $c$ particle pair with masses $m_{c,1}$ and $ m_{c,2}$ as
\begin{equation}
\begin{IEEEeqnarraybox}[][c]{rcl}
k_c  & \ = \ & \sqrt{\frac{2m_{c,1} m_{c,2}}{\left(m_{c,1}+m_{c,2}\right) \mathrm{\hbar}^2} \left(E_c - E_{T_c}\right)}
\IEEEstrut\end{IEEEeqnarraybox}
\label{eq:rho_c massive}
\end{equation}
where $E_{T_c}$ denotes a threshold energy beyond which the channel $c$ is closed, as energy conservation cannot be respected ($E_{T_c} = 0 $ for reactions without threshold).
On the other hand, for a photon particle interacting with a massive body of mass $m_{c,1}$ the center-of-mass wavenumber $k_c$ is linked to the total center-of-mass energy $E_c$ of channel $c$ according to:
\begin{equation}
k_c = \frac{\left( E_c - E_{T_c} \right)}{2 \mathrm{\hbar}\mathrm{c}}\left[ 1 + \frac{m_{c,1} \mathrm{c}^2}{\left( E - E_{T_c} \right) + m_{c,1}\mathrm{c}^2}\right]
\label{eq:rho_c photon}
\end{equation}

Alternatively, in a more unified approach, one can perform a relativistic correction and smooth these differences away by means of the special relativity Mandelstam variable $s_c = (p_{c,1} + p_{c,2})$, also known as the square of the center-of-mass energy, where $p_{c,1}$ and $p_{c,2}$ are the Minkowsky metric four-momenta of the two bodies composing channel $c$, with respective masses $m_{c,1}$ and $ m_{c,2}$ (null for photons). The channel wavenumber $k_c$ can then be expressed as:
\begin{equation}
k_c =  \sqrt{\frac{\left[ s_c - (m_{c,1} + m_{c,2})^2\mathrm{c}^2\right]\left[ s_c - (m_{c,1} - m_{c,2})^2\mathrm{c}^2\right]}{4 \mathrm{\hbar}^2s_c}}
\label{eq:rho_c EDA}
\end{equation}
and the Mandelstam variable $s_c$ can be linked to the center-of-mass energy of the channel $E_c$ through
\begin{equation}
E_c = \frac{ s_c - (m_{c,1} + m_{c,2})^2\mathrm{c}^2}{2 (m_{c,1} + m_{c,2})}
\label{eq:E_c as a function of s_c}
\end{equation}
Interestingly, this is identical to the non-relativistic expression for the center-of-mass energy in terms of the lab energy in whichever channel the total mass $(m_{c,1} + m_{c,2})$ is chosen to be the reference for $E$ (but not in any other).
This special relativistic correction to the non-relativistic R-matrix theory is the approach taken by the EDA code in use at the Los Alamos National Laboratory \cite{EDA_2008, EDA_2015}.

Regarless of the approach taken to link the channel energy $E_c$ to the channel wavenumber $k_c$, conservation of energy (\ref{eq:conservation of energy E = E_c = E_c'}) entails there exists a complex mapping linking the total center-of-mass energy $E$ to the wavenumbers $k_c$, or their associated dimensionless variable $\rho_c= k_c r_c$:
\begin{equation}
\rho_c(E)  \quad \longleftrightarrow \quad E
\label{eq:rho_c(E) mapping}
\end{equation}

Critical properties throughout this article will stem from the analytic continuation of R-matrix operators.
As the outgoing $O_c$ and incoming $I_c$ wave functions are defined according to $\rho_c$ (c.f. section \ref{subsec:External_region_waves} below), the natural variable to perform analytic continuation is thus $\rho_c$, which is equivalent to extending the wavenumbers into the complex plane $k_c \in \mathbb{C}$.
We can see that the mapping (\ref{eq:rho_c(E) mapping}) from complex $k_c$ to complex energies is non-trivial, specially since the wavenumbers are themselves all interconnected.
This creates a multi-sheeted Riemann surface, with branchpoints at each threshold $E_{T_c}$, well documented by Eden \& Taylor \cite{Eden_and_Taylor} (also c.f. section 8 of \cite{Theory_of_Nuclear_Reactions_I_resonances_Humblet_and_Rosenfeld_1961}). 
More precisely, when calculating $\rho_c$ from $E$ one has to chose which sign to assign to $\pm \sqrt{E - E_{T_c}}$ in (\ref{eq:rho_c massive}), or more generally to the mapping (\ref{eq:rho_c EDA}). Figure \ref{fig:mapping rho - E} shows this for the semi-classical case of massive particles (\ref{eq:rho_c massive}), with zero threshold $E_{T_c} = 0$.
Each channel $c$ thus introduces two choices, and hence there are $2^{N_c}$ sheets to the Riemann surface mapping (\ref{eq:conservation of energy E = E_c = E_c'}) to (\ref{eq:rho_c(E) mapping}), with the branch points close or equal to the threshold energies $E_{T_c}$. As we will see, the choice of the sheet will have an impact when finding different R-matrix parameters.

\begin{figure}[ht!!] 
  \centering
  \includegraphics[width=0.5\textwidth]{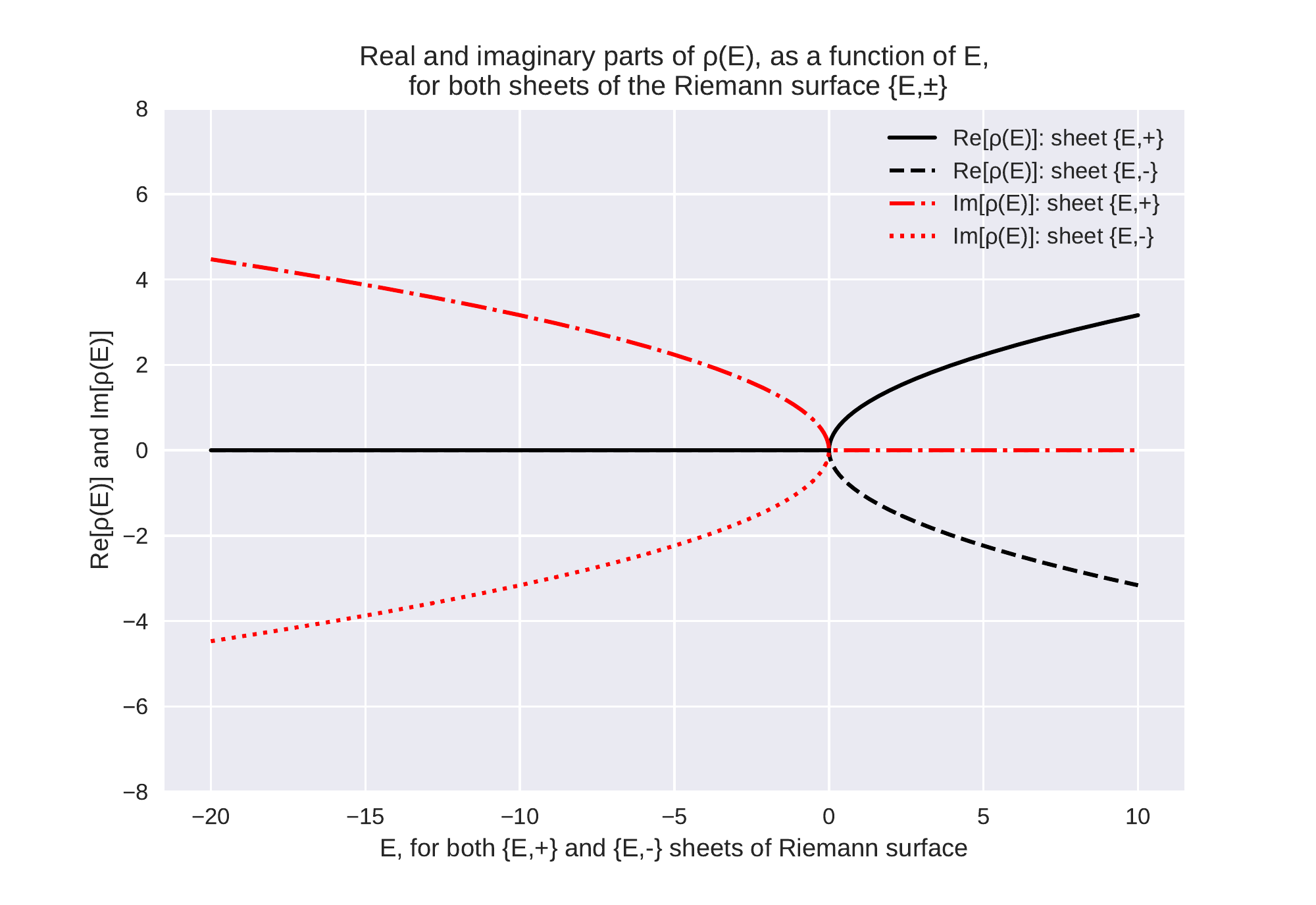}
  \caption{\small{$\rho(E)$ mapping for massive particles in the semi-classical limit (\ref{eq:rho_c massive}). The square root $\rho_c(E) = \pm \rho_0\sqrt{E - E_{T_c}} $ gives rise to two sheets: $\left\{ E,+\right\}$ and $\left\{ E,-\right\}$.}}
  \label{fig:mapping rho - E}
\end{figure}

\subsection{\label{subsec:External_region_waves} External region wave functions}

In the R-matrix model, the external region is subject to either a Coulomb interaction or a free particle movement. In either case, the solutions form a two-dimensional vector space, a basis of which is composed of the incoming and outgoing wave functions: $\boldsymbol{O}(\boldsymbol{k} ) \coloneqq \boldsymbol{\mathrm{diag}}\left(O_c(k_c)\right)$, $\boldsymbol{I}(\boldsymbol{k} ) \coloneqq \boldsymbol{\mathrm{diag}}\left(I_c(k_c)\right)$.
These are Whittaker or confluent hypergeometric function whose analytic continuation is discussed in section II.2.b and the appendix of \cite{Lane_and_Thomas_1958}, and for whose elemental properties and calculation we refer to chapter 14 of \cite{Abramowitz_and_Stegun} and chapter 33 of \cite{NIST_DLMF}, as well as Powell \cite{Powell_1949}, Thompson \cite{Thompson_1986}, and Michel \cite{Michel_2006}.

Note that the incoming and outgoing wave functions are only dependent on the wavenumber of the given channel $k_c$, this is a fundamental hypothesis of the R-matrix model.
For clarity of writing, we will not explicitly write the $k_c$ dependence of these operators unless it is of importance for the argument.

Importantly, the Wronksian of the system is constant: $\forall c, \; \; w_c \ = \ O_c^{(1)}I_c - I_c^{(1)}O_c$, or with identity matrix $\Id{}$:
\begin{equation}
\begin{IEEEeqnarraybox}[][c]{rcl}
\boldsymbol{w} & \ \coloneqq \ & \boldsymbol{O}^{(1)}\boldsymbol{I} - \boldsymbol{I}^{(1)}\boldsymbol{O}\\
& \ = \ & 2\mathrm{i} \Id{}
\IEEEstrut\end{IEEEeqnarraybox}
\label{eq:wronksian expression}
\end{equation}

Of central importance to R-matrix theory is the \textit{Bloch operator}, $\mathcal{L}$, which Claude Bloch introduced as the \textit{op\'erateur de contitions aux limites} in equation (35) of \cite{Bloch_1957}, and that projects the system radially onto the channel boundaries for each channel, at the channel radius $r_c = a_c$. The Bloch operator $\mathcal{L}$ is then added to the Hamiltonian to form a compact Hermitian operator in the internal region (c.f. equation (34) of \cite{Bloch_1957}), from which one can extract a complete discrete generative eigenbasis of the Hilbert space. This is the essence of R-matrix theory, as best described by Claude Bloch in \cite{Bloch_1957}.

This projection on the channel boundaries at $r_c = a_c$, gives rise to the as yet unnamed quantity $\boldsymbol{L}^0$, introduced in equation (1.6a), section VII.1. p.289 of \cite{Lane_and_Thomas_1958}, and which can be recognized in equation (57) of \cite{Bloch_1957}, that is defined for each channel as:
\begin{equation}
\begin{IEEEeqnarraybox}[][c]{rcl}
L_c^0(\rho_c) & \ \coloneqq \ & L_c(\rho_c) - B_c
\IEEEstrut\end{IEEEeqnarraybox}
\label{eq:Bloch L expression}
\end{equation}
where $\rho_c = k_c a_c$ has been projected on the channel surface, $B_c$ is the arbitrary outgoing-wave boundary condition, and $L_c(\rho_c)$ is the dimensionless reduced logarithmic derivative of the outgoing-wave function at the channel surface:
\begin{equation}
\begin{IEEEeqnarraybox}[][c]{rcl}
L_c(\rho_c) & \ \coloneqq \ & \frac{\rho_c}{O_c} \frac{\partial O_c}{\partial \rho_c}
\IEEEstrut\end{IEEEeqnarraybox}
\label{eq:L expression}
\end{equation}
or, equivalently, in matrix notation, and where $\left[ \; \cdot \; \right]^{(1)}$ designates the derivative with respect to $\rho_c$:
\begin{equation}
\begin{IEEEeqnarraybox}[][c]{rcl}
\boldsymbol{L} = \boldsymbol{\mathrm{diag}}\left(L_c\right) = \boldsymbol{\rho}\boldsymbol{O}^{-1}\boldsymbol{O}^{(1)}  \IEEEstrut\end{IEEEeqnarraybox}
\label{eq:L expression matrix}
\end{equation}
so that the $\boldsymbol{L^0}$ matrix function is written: $\boldsymbol{L^0} \coloneqq \boldsymbol{L} - \boldsymbol{B}$.

Using the Powell recurrence formulae \cite{Powell_1949}, R.G. Thomas established the following scheme to calculate the outgoing-wave reduced logarithmic derivatives $L_c$ for different angular momenta $\ell$ values in the Coulomb case (c.f. p.350, appendix of \cite{Lane_and_Thomas_1958}, eqs.(A.12) and (A.13))
\begin{equation}
\begin{IEEEeqnarraybox}[][c]{rcl}
\IEEEstrut
L_\ell & \ = \ & \frac{a_\ell}{b_\ell - L_{\ell-1}} - b_\ell
\end{IEEEeqnarraybox}
\label{eq::L_ell recurrence formula}
\end{equation}
with
\begin{equation}
a_\ell \coloneqq \rho^2 + \left(\frac{\rho \eta}{\ell}\right)^2 \quad , \quad b_\ell \coloneqq \ell + \left(\frac{\rho \eta}{\ell}\right)
\end{equation}
where $\eta = \frac{Z_1 Z_2 e^2 M_\alpha a_c}{\hbar^2 \rho_c}$ is the dimensionless Coulomb field parameter.

In general, both $O_c(\rho)$ and $L_{\ell}(\rho)$ are meromorphic functions of $\rho$ with \textit{a priori} an infinity of poles, and for whose computation we refer to \cite{Powell_1949, Thompson_1986, Michel_2006}.
In lemma \ref{lem::Mittag-Leffler of L_c Lemma}, we here establish the Mittag-Leffler expansion of $L_c(\rho)$.

\begin{lem}\label{lem::Mittag-Leffler of L_c Lemma}
\textsc{Outgoing-wave reduced logarithmic derivative $L_c(\rho)$ Mittag-Leffler Expansion}. \\
The outgoing-wave reduced logarithmic derivative $L_c(\rho)$, defined in (\ref{eq:L expression}), admits the following Mittag-Leffler pole expansion:
\begin{equation}
\begin{IEEEeqnarraybox}[][c]{rcl}
    \frac{L_c(\rho)}{\rho} & = & \frac{- \ell}{\rho} + \mathrm{i} + \sum_{n \geq 1}  \frac{1}{\rho-\omega_n}
\IEEEstrut\end{IEEEeqnarraybox}
\label{eq::Explicit Mittag-Leffler expansion of L_c}
\end{equation}
where $\left\{\omega_n \right\}$ are the roots of the outgoing wavefunctions $O_c(\rho)$.
For neutral particles, there are a finite number of such roots, reported in table \ref{tab::roots of the outgoing wave functions}.
\end{lem}

\begin{proof}
From definition (\ref{eq:L expression}), $L_c$ is the reduced logrithmic derivative of the outgoing wavefunction $L_c(\rho) \coloneqq  \rho \frac{O_c^{(1)}(\rho)}{O_c(\rho)}$. In both the Coulomb and the neutral particle case, the outgoing wavefunction $O_c(\rho)$ is a confluent hypergeometric function with simple roots $\left\{\omega_n \right\}$. Moreover, their logarithmic derivative $ \frac{O_c^{(1)}(\rho)}{O_c(\rho) }$ is bound at infinity. Thus, the following hypotheses stand:
\begin{itemize}
    \item $L_{\ell}(\rho)$ has simple poles $\left\{ \omega_n \right\}$, zeros of the $O_c(\rho)$,
    \item $L_{\ell}(\rho)$ has residues $\left\{ \omega_n \right\}$ at the $\left\{ \omega_n \right\}$ pole,
    \item $\exists M\in \mathbb{R}$ such as $\left| L_{\ell}(\rho) \right| < M |z|  $ on circles $\mathcal{C}_D$ as $D \longrightarrow \infty$ 
\end{itemize}
By removing the pole of $ \frac{O_c^{(1)}(\rho)}{O_c(\rho) }$ at zero, these hypotheses ensure Mittag-Leffler expansion (\ref{eq::Mittag-Leffler expansion of L_c}) to be verified:
\begin{equation}
\begin{IEEEeqnarraybox}[][c]{rcl}
\frac{L_c(\rho)}{\rho} & = & \frac{L_c(0)}{\rho} + L_c^{(1)}(0) + \sum_{n \geq 1} \left[ \frac{1}{\rho-\omega_n} +  \frac{1}{\omega_n} \right]
\IEEEstrut\end{IEEEeqnarraybox}
\label{eq::Mittag-Leffler expansion of L_c}
\end{equation}
R.G. Thomas' recurrence formula (\ref{eq::L_ell recurrence formula}) implies that $L_c(\rho_c)$ satisfies $L_\ell(0) = - \ell$, for both neutral and charged particles.
Moreover, evaluating $ \frac{O_c^{(1)}(\rho)}{O_c(\rho) }$ at the limit of infinity yeilds:
\begin{equation}
     L_c^{(1)}(0) + \sum_{k\geq 1}\frac{1}{\omega_k} = \underset{\rho \to \infty }{\mathrm{Lim}}\left( \frac{L_c(\rho)}{\rho }\right) = \underset{\rho \to \infty }{\mathrm{Lim}}\left( \frac{O_c^{(1)}(\rho)}{O_c(\rho) }\right) = \mathrm{i}
\end{equation}
so that the Mittag-Leffler expansion (\ref{eq::Mittag-Leffler expansion of L_c}) takes the desired form of (\ref{eq::Explicit Mittag-Leffler expansion of L_c}).






\end{proof}

\begin{table*}
\caption{\label{tab::L_values_neutral} Reduced logarithmic derivative $L_\ell(\rho) \coloneqq \frac{\rho}{O_\ell} \frac{\partial O_\ell}{\partial r}(\rho)$ of outgoing wavefunction $O_\ell(\rho)$, and $L_\ell^0(\rho) \coloneqq L_\ell(\rho) - B_\ell $ using $B_\ell = - \ell$, irreducible forms and Mittag-Leffler pole expansions for neutral particles, for angular momenta $0 \leq \ell \leq 4$. }
\begin{ruledtabular}
\begin{tabular}{c|c|c|c|c}
\ \ & $L_\ell(\rho)$ from recurrence (\ref{eq::L_ell recurrence formula}) & $\begin{array}{c}
     L_\ell^0(\rho) \coloneqq L_\ell(\rho) - B_\ell  \\
     \text{using } B_\ell = - \ell \text{ in (\ref{eq::L_ell recurrence formula})} 
\end{array}$ & $\begin{array}{c}
    L_\ell(\rho) \text{ from lemma  \ref{lem::Mittag-Leffler of L_c Lemma},}  \\
    \text{poles } \big\{ \omega_n\big\} \text{ from table \ref{tab::roots of the outgoing wave functions}} 
\end{array} $ & $\begin{array}{c}
     \text{Outgoing wavefunction } \\
     O_{\ell}(\rho) \text{ from  (\ref{eq::Outgoing wavefuntion O_ell expression for neutral particles})}
\end{array} $ \tabularnewline
\hline
$\ell$  &  $ L_\ell(\rho) = \frac{\rho^2 }{\ell - L_{\ell-1}(\rho)} - \ell $ & $ L_\ell^0(\rho) =  \frac{\rho^2}{2\ell -1 - L_{\ell-1}^0(\rho) }$   & $L_\ell(\rho) = -\ell + \mathrm{i} \rho + \sum_{n\geq 1} \frac{\rho}{\rho - \omega_n}$ &$ O_\ell(\rho)  =  \mathrm{e}^{\mathrm{i}\left(\rho + \frac{1}{2}\ell \pi \right)}\frac{\prod_{n \geq 1}\left(\rho-\omega_n\right)}{\rho^\ell}$  \tabularnewline
\hline \hline
0  &  $\mathrm{i}\rho$  & $\mathrm{i}\rho$ & $\left\{ \emptyset \right\}$ &  $\mathrm{e}^{\mathrm{i}\rho}  $\tabularnewline
1  &  $ \frac{-1 + \mathrm{i}\rho + \rho^2}{1-\mathrm{i}\rho}$ & $ \frac{\rho^2}{1-\mathrm{i}\rho}$  &  $\omega_{1,2}^{\ell = 2} = -\mathrm{i} $ & $\mathrm{e}^{\mathrm{i}\rho}\left(\frac{1}{\rho} - \mathrm{i}\right)  $ \tabularnewline
2   &  $ \frac{- 6 + 6\mathrm{i}\rho + 3 \rho^2 - \mathrm{i}\rho^3 }{3 - 3\mathrm{i}\rho - \rho^2} $ & $ \frac{ \rho^2 - \mathrm{i}\rho^3 }{3 - 3\mathrm{i}\rho - \rho^2} $ & $\omega_{1,2}^{\ell = 2} \approx \pm 0.86602 - 1.5\mathrm{i} $ &  $\mathrm{e}^{\mathrm{i}\rho}\left(\frac{3}{\rho^2} -\frac{3\mathrm{i}}{\rho} - 1\right)  $ \tabularnewline
3  &  $\frac{- 45 + 45 \mathrm{i} \rho + 21\rho^2 - 6\mathrm{i}\rho^3 -\rho^4}{15 - 15\mathrm{i}\rho - 6\rho^2 + \mathrm{i}\rho^3}$ & $\frac{3\rho^2 - 3\mathrm{i}\rho^3 -\rho^4}{15 - 15\mathrm{i}\rho - 6\rho^2 + \mathrm{i}\rho^3}$ & $\begin{array}{rl}
     \omega_1^{\ell = 3} & \approx - 2.32219 \mathrm{i}  \\
     \omega_{2,3}^{\ell = 3}  & \approx \pm 1.75438 - 1.83891 \mathrm{i}  
\end{array} $ &  $\mathrm{e}^{\mathrm{i}\rho}\left(\frac{15}{\rho^3} -\frac{15\mathrm{i}}{\rho^2} -\frac{6}{\rho} + \mathrm{i}\right)  $  \tabularnewline
4  &  $\frac{ - 420 + 420\mathrm{i}\rho + 195\rho^2 -55\mathrm{i}\rho^3 - 10\rho^4 + \mathrm{i}\rho^5}{105 - 105\mathrm{i}\rho -45\rho^2 + 10\mathrm{i}\rho^3 + \rho^4}$ & $\frac{ 15\rho^2 -15\mathrm{i}\rho^3 - 6\rho^4 + \mathrm{i}\rho^5}{105 - 105\mathrm{i}\rho -45\rho^2 + 10\mathrm{i}\rho^3 + \rho^4}$  &  $ \begin{array}{rl}
   \omega_{1,2}^{\ell = 4} & \approx \pm 2.65742 - 2.10379\mathrm{i}   \\
    \omega_{3,4}^{\ell = 4} & \approx \pm 0.867234 - 2.89621 \mathrm{i} 
       \end{array}$ & $\mathrm{e}^{\mathrm{i}\rho}\left(\frac{105}{\rho^4} - \frac{105\mathrm{i}}{\rho^3} -\frac{45}{\rho^2} + \frac{10 \mathrm{i}}{\rho} + 1 \right)  $ \tabularnewline
\end{tabular}
\end{ruledtabular}
\end{table*}

Lemma \ref{lem::Mittag-Leffler of L_c Lemma} establishes the Mittag-Leffer expansion of $L^0_c(\rho_c)$ as a function of the roots $\left\{\omega_n \right\}$ of the outgoing wavefunctions $O_c(\rho)$, which are Hankel functions in the neutral particle case, and Whittaker funtions in the more general case of charged particles (c.f. equations (2.14b) and (2.17) section III.2.b. p.269 of \cite{Lane_and_Thomas_1958}).
Extensive literature covers these functions \cite{Abramowitz_and_Stegun, NIST_DLMF}. In the neutral particules case of Hankel functions \cite{Zeros_of_Hankel_Functions_and_Poles_1963, Complex_zeros_of_cylinder_functions_1966, Zeros_of_Hankel_Functions_1982, Reduced_Logarithmic_Dervative_of_Hankel_functions_1983, Calculating_zeros_of_Hankel_functions_USSR_1985, Partial_fraction_expansion_Bessel_2005} the search for their zeros established that the reduced logarithmic derivative of the outgoing wave function is a rational function of $k_c$ of degree $\ell$.
In the general case there are indeed $\ell$ zeros to the Hankel function for $\left| \Re[\rho] \right| < \ell$, but for $\left| \Re[\rho] \right| > \ell$ there exists an infinity of zeros, on or close to the real axis (c.f. FIG.1\&2 of \cite{Complex_zeros_of_cylinder_functions_1966}).
However, in our particular case of physical (i. e. integer) angular momenta $\ell \in \mathbb{Z}$, the order of the Hankel function happens to be a half-integer: $H_{\ell+1/2}$. Cruicially, Hankel functions of half integer order constitute a very special case: they have only a finite number of zeros in the finite complex plane, where all but $\ell$ of them have migrated to infinity.
This behavior is reported in \cite{Zeros_of_Hankel_Functions_1982}, where one can observe how the zeros of $H_{\nu}$ as $\nu$ varies between two consecutive integer values.
Here, we report in table \ref{tab::roots of the outgoing wave functions} all the algebraically solvable cases of up to $\ell=4$, past which Neils Abel and Evariste Galois theorems do not guarantee solvability of $\left\{\omega_n\right\}$ by radicals (c.f. Abel-Ruffini theorem and Galois theory).

Another perspective over this property is that in the neutral particle case, $\eta = 0$ and $L_{\ell=0}(\rho) = \mathrm{i}\rho$, so that recurrence relation (\ref{eq::L_ell recurrence formula}) entails $L_c(\rho_c)$ -- and thus the $\boldsymbol{L^0}$ function -- is a rational fraction in $\rho_c$, whose irreducible expressions are reported in table \ref{tab::L_values_neutral} along with their partial fraction decomposition, established in lemma \ref{lem::Mittag-Leffler of L_c Lemma}, and whose poles are documented in table \ref{tab::roots of the outgoing wave functions}.
Moreover, since definition (\ref{eq:L expression}) entails $\frac{\partial O_c}{\partial \rho}(\rho) = \frac{L_c}{\rho}(\rho) O_c(\rho)$, a direct integration of (\ref{eq::Mittag-Leffler expansion of L_c}) yields (with the correct multiplicative constant):
\begin{equation}
\begin{IEEEeqnarraybox}[][c]{rcl}
O_\ell(\rho) & = & \mathrm{e}^{\mathrm{i}\left(\rho + \frac{1}{2}\ell \pi \right)}\frac{\prod_{n \geq 1}\left(\rho-\omega_n\right)}{\rho^\ell} 
\IEEEstrut\end{IEEEeqnarraybox}
\label{eq::Outgoing wavefuntion O_ell expression for neutral particles}
\end{equation}
This expression converges for neutral particles as the number of poles is finite (c.f. discussion after theorem \ref{theo::r_j,c uncer change of a_c explicit}), and using Vieta's formulas with the denominator of $L_\ell(\rho)$ enables to construct the developed forms reported in table \ref{tab::L_values_neutral}.

Similar results do not hold for the charged particules case of Whittaker functions, where there always exists an infinity of zeros to the outgoing wavefunction \cite{Zeros_of_Whittaker_function_1999, New_Asymptotics_Whittaker_zeros_2001}, and where a Coulomb phase shift would be present for any Weierstrass expansion in infinite product of type (\ref{eq::Outgoing wavefuntion O_ell expression for neutral particles}).

\subsection{\label{subsec:Internal region parameters} Internal region parameters}

Projections upon the orthonormal basis formed by the eigenvectors of the Hamiltonian completed by the Bloch operator $\mathcal{L}$ allow for the parametrization of the interaction Hamiltonian in the internal region by means of the Wigner-Eisenbud \textit{resonance parameters} \cite{Bloch_1957}, composed of both the real \textit{resonance energies} $E_\lambda \in \mathbb{R}$, and the real \textit{resonance widths} $\gamma_{\lambda,c} \in \mathbb{R}$.
From the latter, and using Brune's notation $\boldsymbol{e} \coloneqq \boldsymbol{\mathrm{diag}}\left( E_\lambda \right)$ and $\boldsymbol{\gamma}\coloneqq \boldsymbol{\mathrm{mat}}\left(\gamma_{\lambda,c}\right)_{\lambda,c}$, the \textit{Channel R matrix}, $\boldsymbol{R}$, is defined as
\begin{equation}
\begin{IEEEeqnarraybox}[][c]{rcl}
\boldsymbol{R} & \; \coloneqq \; &  \sum_{\lambda=1}^{N_\lambda}\frac{\gamma_{\lambda,c}\gamma_{\lambda,c'}}{E_\lambda - E} \; = \; \boldsymbol{\gamma}^\mathsf{T} \left(\boldsymbol{e} - E\Id{}\right)^{-1}\boldsymbol{\gamma}
\IEEEstrut\end{IEEEeqnarraybox}
\label{eq:R expression}
\end{equation}
and the \textit{Level A matrix}, $\boldsymbol{A}$, is defined through its inverse:
\begin{equation}
\begin{IEEEeqnarraybox}[][c]{rcl}
\boldsymbol{A^{-1}} & \ \coloneqq \ & \boldsymbol{e} - E\Id{} - \boldsymbol{\gamma}\left( \boldsymbol{L} - \boldsymbol{B} \right)\boldsymbol{\gamma}^\mathsf{T}
\IEEEstrut\end{IEEEeqnarraybox}
\label{eq:inv_A expression}
\end{equation}
where $\boldsymbol{B} = \boldsymbol{\mathrm{diag}}\left( B_c \right)$ is the outgoing-wave boundary condition, which is arbitrary, constant (non-dependent on the wavenumber), and for which Bloch demonstrated that if it is real (i.e. $B_c \in \mathbb{R}$), then the Wigner-Eisenbud resonance parameters are also real \cite{Bloch_1957}.
From this, one can view the Wigner-Eisenbud parameters as the set of channel radii $a_c$, boundary conditions $B_c$, resonance widths $\gamma_{\lambda, c}$, resonance energies $E_\lambda$ and thresholds $E_{T_c}$. This set of parameters fully determine the energy (or wavenumber) dependence of the scattering matrix $\boldsymbol{U}$ through equation (\ref{eq:U expression}).

\subsection{\label{subsec:Scattering matrix} Scattering matrix}

As explained by Claude Bloch, the genius of the R-matrix theory is that it can combine the internal region with the external region to simply express the resulting scattering matrix $\boldsymbol{U}$ (also called \textit{collision matrix}, and often noted $\boldsymbol{S}$, though we here stick to the Lane \& Thomas scripture $\boldsymbol{U}$ for the scattering matrix) as:
\begin{equation}
\begin{IEEEeqnarraybox}[][c]{rcl}
\boldsymbol{U} & \ = \ & \boldsymbol{O}^{-1}\boldsymbol{I} + \boldsymbol{w} \boldsymbol{\rho}^{1/2} \boldsymbol{O}^{-1}\left[ \boldsymbol{R}^{-1} + \boldsymbol{B} -  \boldsymbol{L} \right]^{-1}  \boldsymbol{O}^{-1} \boldsymbol{\rho}^{1/2} \\
& \ = \ & \boldsymbol{O}^{-1}\boldsymbol{I} + 2\mathrm{i} \boldsymbol{\rho}^{1/2} \boldsymbol{O}^{-1} \boldsymbol{\gamma}^\mathsf{T} \boldsymbol{A} \boldsymbol{\gamma} \boldsymbol{O}^{-1} \boldsymbol{\rho}^{1/2} \\
& \ = \ & \boldsymbol{O}^{-1}\boldsymbol{I} + 2\mathrm{i} \boldsymbol{\rho}^{1/2} \boldsymbol{O}^{-1} \boldsymbol{R}_{L} \boldsymbol{O}^{-1} \boldsymbol{\rho}^{1/2}
\IEEEstrut\end{IEEEeqnarraybox}
\label{eq:U expression}
\end{equation}
The equivalence between these channel and level matrix expressions stems from the identity $ \left[\boldsymbol{R}^{-1} + \boldsymbol{B} - \boldsymbol{L}\right]^{-1} = \boldsymbol{\gamma}^\mathsf{T} \boldsymbol{A} \boldsymbol{\gamma}$ which defines the \textit{Kapur-Peirels operator}, $\boldsymbol{R}_{L}$:
\begin{equation}
\begin{IEEEeqnarraybox}[][c]{rcl}
\boldsymbol{R}_{L}  & \ \coloneqq  \ &  \left[ \boldsymbol{R}^{-1} - \boldsymbol{L^0} \right]^{-1}  = \left[ \boldsymbol{R}^{-1} + \boldsymbol{B} -  \boldsymbol{L} \right]^{-1}  \\ & = &  \boldsymbol{\gamma}^\mathsf{T} \boldsymbol{A} \boldsymbol{\gamma}
\IEEEstrut\end{IEEEeqnarraybox}
\label{eq:Kapur-Peirels Operator and Channel-Level equivalence}
\end{equation}
The Kapur-Peirels operator $\boldsymbol{R}_{L}$ will play a central role in this article and is thoroughly discussed in section \ref{sec:R_L Siegert and Humblet}. 
Identity (\ref{eq:Kapur-Peirels Operator and Channel-Level equivalence}) can be proved by means of the \textit{Woodbury identity}:
\begin{equation}
\begin{IEEEeqnarraybox}[][c]{rcl}
\left[ \boldsymbol{A} + \boldsymbol{B}\boldsymbol{D}^{-1}\boldsymbol{C}\right]^{-1}   & \ = \ & \boldsymbol{A}^{-1} - \boldsymbol{A}^{-1}\boldsymbol{B}\left[ \boldsymbol{D} + \boldsymbol{C}\boldsymbol{A}^{-1}\boldsymbol{B} \right]^{-1}\boldsymbol{C}\boldsymbol{A}^{-1}
\IEEEstrut\end{IEEEeqnarraybox}
\label{eq:Woodbury identity}
\end{equation}
Indeed, the application of the Woodbury identity (\ref{eq:Woodbury identity}) to equality (\ref{eq:Kapur-Peirels Operator and Channel-Level equivalence}), with $\boldsymbol{A}_{\mathrm{Wood}} = \boldsymbol{R}^{-1}$, $\boldsymbol{B}_{\mathrm{Wood}} = \boldsymbol{L^0}$, and $\boldsymbol{C}_{\mathrm{Wood}} = \boldsymbol{D}_{\mathrm{Wood}} = \Id{}$ yields
\begin{equation*}
\begin{IEEEeqnarraybox}[][c]{r}
\IEEEstrut
\left[\boldsymbol{R}^{-1} - \boldsymbol{L^0}\right]^{-1}   =  \boldsymbol{R} + \boldsymbol{R} \boldsymbol{L^0} \left[\Id{} - \boldsymbol{R}\boldsymbol{L^0} \right]^{-1}\boldsymbol{R} \\
=  \boldsymbol{\gamma}^\mathsf{T} \left[  \left(\boldsymbol{e} - E\Id{}\right)^{-1} +  \left(\boldsymbol{e} - E\Id{}\right)^{-1}\boldsymbol{\gamma} \boldsymbol{L^0} \times \right. \\
\left. \left[\Id{} - \boldsymbol{\gamma}^\mathsf{T} \left(\boldsymbol{e} - E\Id{}\right)^{-1}\boldsymbol{\gamma} \boldsymbol{L^0}  \right]^{-1}\boldsymbol{\gamma}^\mathsf{T} \left(\boldsymbol{e} - E\Id{}\right)^{-1} \right] \boldsymbol{\gamma}
\end{IEEEeqnarraybox}
\label{eq:R_L Woodbury channel-level equivalence 1}
\end{equation*}
and then reversely applying the Woodbury identity with $\boldsymbol{A}_{\mathrm{Wood}} = \left(\boldsymbol{e} - E\Id{}\right)$, $\boldsymbol{B}_{\mathrm{Wood}} = -\boldsymbol{\gamma} \boldsymbol{L^0}$, $\boldsymbol{C}_{\mathrm{Wood}} = \boldsymbol{\gamma}^\mathsf{T}$, and $ \boldsymbol{D}_{\mathrm{Wood}} = \Id{}$ one now recognizes
\begin{equation*}
\begin{IEEEeqnarraybox}[][c]{rCl}
\IEEEstrut
\left[\boldsymbol{R}^{-1} - \boldsymbol{L^0}\right]^{-1} 
& = & \boldsymbol{\gamma}^\mathsf{T}   \left[ \left(\boldsymbol{e} - E\Id{}\right) - \boldsymbol{\gamma} \boldsymbol{L^0} \boldsymbol{\gamma}^\mathsf{T}\right]^{-1} \boldsymbol{\gamma}  \\
& = & \boldsymbol{\gamma}^\mathsf{T} \boldsymbol{A} \boldsymbol{\gamma}
\end{IEEEeqnarraybox}
\label{eq:R_L Woodbury channel-level equivalence 2}
\end{equation*}

Considering the multi-sheeted Riemann surface stemming from the analytic continuation of mapping (\ref{eq:rho_c(E) mapping}),
a truly remarkable and seldom noted property of the Wigner-Eisenbud formalism is that it completely de-entagles the branch points and the multi-sheeted structure --- entirely present in the outgoing $\boldsymbol{O}$ and incoming $\boldsymbol{I}$ wave functions in the scattering matrix expression (\ref{eq:U expression}) --- from the resonance parameters --- which are the poles and residues of the channel matrix $\boldsymbol{R}$ as of equation (\ref{eq:R expression}), and these poles and residues live on a simple complex energy $E$ sheet, with no branch points, and furthermore are all real.
This de-entanglement of the branch-point structure gives the $\boldsymbol{R}$ matrix all its uniqueness in R-matrix theory.
For instance, it does not translate to the level matrix $\boldsymbol{A}$, whose analytic continuation entails a multi-sheeted Riemann surface due to the introduction of the $\boldsymbol{L^0}(\boldsymbol{\rho}(E)))$ matrix function in its definition (\ref{eq:inv_A expression}). The same is true for the Brune or the Siegert-Humblet parameters, as will be discussed throughout this article.

\section{\label{sec:Invariance to B_c, Brune and Siegert-Humblet}Consequences of invariance with respect to boundary condition parameters.}

Having recalled essential results from R-matrix theory and the Wigner-Eisenbud parameters $\left\{a_c , B_c , \gamma_{\lambda,c}, E_\lambda, E_{T_c} \right\}$, we here focus on the fact that the fundamental physical operator describing the scattering event is the scattering matrix $\boldsymbol{U}$, and while the threshold energies $E_{T_c}$ are intrinsic physical properties of the system, all the other Wigner-Eisenbud parameters $a_c$, $B_c$, $\gamma_{\lambda,c}$, and $E_\lambda$ are interrelated and depend on arbitrary values of the channel radius $a_c$, or the boundary condition $B_c$.
Though the channel radius $a_c$ can have some physical interpretation, this is not the case of the boundary condition $B_c$. This section explains how the search to remove the explicit dependence of the parameters on the arbitrary boundary condition $B_c$ has lead to both the Brune $\boldsymbol{R}_S$ parameters and the Siegert-Humblet $\boldsymbol{R}_{L}$ parameters.

\subsection{\label{sec:Invariance to B_c}Invariance to $B_c$}

The dependence of the Wigner-Eisenbug parameters to the boundary condition $B_c$ can be made explicit by fixing the channel radius $a_c$ and performing a change of boundary condition $\boldsymbol{B} \to \boldsymbol{B'}$. This must entail a change in resonance parameters $E_\lambda \rightarrow E_\lambda'$ and $\gamma_{\lambda,c} \rightarrow \gamma_{\lambda,c}'$ which leaves the scattering matrix $\boldsymbol{U}$ unchanged.

As described by Barker in \cite{Boundary_condition_Barker_1972}, such change of variables can be performed by noticing that $ \boldsymbol{e} - \boldsymbol{\gamma}\left( \boldsymbol{B}' - \boldsymbol{B} \right)\boldsymbol{\gamma}^\mathsf{T}$ is a real symmetric matrix when both $\boldsymbol{B}$ and $\boldsymbol{B'}$ are real.
The spectral theorem thus assures there exists a real orthogonal matrix $\boldsymbol{K}$ and a real diagonal matrix $\boldsymbol{D}$ such that
\begin{equation}
\boldsymbol{e} - \boldsymbol{\gamma}\left( \boldsymbol{B}' - \boldsymbol{B} \right)\boldsymbol{\gamma}^\mathsf{T} = \boldsymbol{K}^\mathsf{T}\boldsymbol{D}\boldsymbol{K}
\end{equation}
The new parameters are then defined as
\begin{equation}
\boldsymbol{e}' \coloneqq  \boldsymbol{D} \quad \quad , \quad \quad \boldsymbol{\gamma}' \coloneqq \boldsymbol{K}\boldsymbol{\gamma}
\label{eq: Wigner-Eisenbud parameters transformations under change of Bc}
\end{equation}
This change of variables satisfies:
\begin{equation}
{\boldsymbol{\gamma}'}^\mathsf{T} \boldsymbol{A}_{B'} \boldsymbol{\gamma}' =  \boldsymbol{\gamma}^\mathsf{T} \boldsymbol{A}_{B} \boldsymbol{\gamma}
\label{eq:: gAg invariance for B'}
\end{equation}
and thus leaves the scattering matrix unaltered through equation (\ref{eq:U expression}). Here $\boldsymbol{A}_{B'}$ designates the level matrix from parameters $\boldsymbol{e}'$, $\boldsymbol{\gamma}'$ and $\boldsymbol{B}'$.
Equivalently, using the Woodbury identity (\ref{eq:Woodbury identity}) shows that this change of variables verifies (c.f. eq.(4) of \cite{Boundary_condition_Barker_1972} or eq. (3.27) of \cite{Descouvemont_2010}):
\begin{equation}
\boldsymbol{R}_B^{-1} + \boldsymbol{B} \ = \ \boldsymbol{R}_{B'}^{-1} + \boldsymbol{B'}
\label{eq:: R_B invariance for B'}
\end{equation}
If the change of variable is infinitesimal, this invariance property translates into the following equivalent differential equations on the Wigner-Eisenbud $\boldsymbol{R}_B$ matrix,
\begin{equation}
\frac{\partial \boldsymbol{R}_B^{-1}}{\partial \boldsymbol{B}} + \Id{} \ = \ \boldsymbol{0} \quad \text{i.e.} \quad  \frac{\partial \boldsymbol{R}_B}{\partial \boldsymbol{B}} - \boldsymbol{R}_B^2 \ = \ \boldsymbol{0}
\label{eq:: R_B invariance for infinitesimal B'}
\end{equation}
(c.f. eq (2.5b) section IV.2. p.274 of \cite{Lane_and_Thomas_1958}) where we made use of the following property to prove the equivalence:
\begin{equation}
\frac{\partial \boldsymbol{M^{-1}}}{\partial z}(z) = - \boldsymbol{M^{-1}}(z)  \left( \frac{\partial \boldsymbol{M}}{\partial z}(z) \right) \boldsymbol{M^{-1}}(z)
\label{eq::inv M derivatie property}
\end{equation}

\subsection{\label{sec:R_S Brune transform}Real invariant $\boldsymbol{R}_S$ parameters: Brune's alternative parametrization}

Since the physics of the system are invariant with the choice of the arbitrary $B_c$ boundary condition, Brune built on Barker's work \cite{Boundary_condition_Barker_1972} to propose an alternative parametrization of R-matrix theory in which the alternative parameters, $\boldsymbol{\widetilde{e}}$ and $\boldsymbol{\widetilde{\gamma}}$, are boundary-condition independent \cite{Brune_2002}.

\subsubsection{\label{sec:R_S def}Definition of Brune's $\boldsymbol{R}_S$ parametrization}

Key to Brune's alternative parametrization is the splitting of the outgoing-wave reduced logarithmic derivative -- and thus the $\boldsymbol{L^0}$ matrix function -- into real and imaginary parts, respectively the shift $\boldsymbol{S}$ and penetration $\boldsymbol{P}$ factors:
\begin{equation}
\boldsymbol{L} \ = \ \boldsymbol{S} + i \boldsymbol{P}
\label{eq:: L = S + iP}
\end{equation}

From there, and with slight changes from the notation in \cite{Brune_2002}, the \textit{physical level matrix} $\boldsymbol{\widetilde{A}}$ is defined as:
\begin{equation}
\boldsymbol{\widetilde{A}^{-1}}(E) \ = \ \boldsymbol{\widetilde{G}} + \boldsymbol{\widetilde{e}} - E\left[\Id{} + \boldsymbol{\widetilde{H}} \right] - \boldsymbol{\widetilde{\gamma}}\boldsymbol{L}(E)\boldsymbol{\widetilde{\gamma}}^\mathsf{T}
\label{eq::Brune physical level matrix}
\end{equation}
with
\begin{equation}
\widetilde{G}_{\lambda \mu} \ = \ \frac{\widetilde{\gamma_\mu}\left( S_\mu \widetilde{E_\lambda} -  S_\lambda \widetilde{E_\mu}  \right)\widetilde{\gamma_\lambda}}{ \widetilde{E_\lambda} - \widetilde{E_\mu} }
\end{equation}
and
\begin{equation}
\widetilde{H}_{\lambda \mu} \ = \ \frac{\widetilde{\gamma_\mu}\left( S_\mu -  S_\lambda \right)\widetilde{\gamma_\lambda}}{ \widetilde{E_\lambda} - \widetilde{E_\mu} }
\end{equation}
such that with the new \textit{alternative resonance parameters}, $\widetilde{E_i}$ and $\widetilde{\gamma_{i,c}}$, the following equality stands,
\begin{equation}
\boldsymbol{\gamma^\mathsf{T} A \gamma} = \boldsymbol{\widetilde{\gamma}^\mathsf{T} \widetilde{A} \widetilde{\gamma}}
\label{eq:R_L unchanged by Brune}
\end{equation}
and thus the scattering matrix $\boldsymbol{U}$ is left unchanged.

These alternative Brune parameters $\widetilde{e}$ and $\widetilde{\gamma}$ are no longer $\boldsymbol{B}$ dependent since the arbitrary boundary condition does not appear in the definition of the physical level matrix, and from there in the parametrization of the scattering matrix.

Brune explains how to compute his parameters from the Wigner-Eisenbud ones by finding the $\left\{\widetilde{E_i}\right\}$ scalars and $\left\{\boldsymbol{g_i}\right\}$ vectors (noted $\left\{\boldsymbol{a_i}\right\}$ in \cite{Brune_2002}) that solve the generalized eigenproblem \cite{Brune_2002}:
\begin{equation}
\begin{IEEEeqnarraybox}[][c]{rCl}
\left[\boldsymbol{e} - \boldsymbol{\gamma} \left( \boldsymbol{S}(\widetilde{E_i}) - \boldsymbol{B} \right)\boldsymbol{ \gamma}^\mathsf{T} \right]\boldsymbol{g_i} = \widetilde{E_i}\boldsymbol{g_i}
\IEEEstrut\end{IEEEeqnarraybox}
\label{eq:Brune eigenproblem}
\end{equation}
and defining the Brune parameters as:
\begin{equation}
\boldsymbol{\widetilde{e}} \coloneqq  \boldsymbol{\mathrm{diag}}(\widetilde{E_i}) \quad \quad , \quad \quad \boldsymbol{\widetilde{\gamma}} \coloneqq \boldsymbol{g}^\mathsf{T} \boldsymbol{\gamma}
\label{eq:Brune parameters}
\end{equation}
where $\boldsymbol{g}$ is the matrix composed of the column eigenvectors: $\boldsymbol{g} \coloneqq \left[\boldsymbol{g_1}, \hdots, \boldsymbol{g_i} , \hdots \right]$.
The physical level matrix is then defined as (c.f. equation (30), \cite{Brune_2002}):
\begin{equation}
\boldsymbol{\widetilde{A}}^{-1} \ \coloneqq \ \boldsymbol{g}^\mathsf{T}\boldsymbol{A}^{-1}\boldsymbol{g}
\label{eq:: Brune invA}
\end{equation}
which guarantees
\begin{equation}
\boldsymbol{A} \ \coloneqq \ \boldsymbol{g}\boldsymbol{\widetilde{A}}\boldsymbol{g}^\mathsf{T}
\label{eq:: Brune A = aAa}
\end{equation}
and thus (\ref{eq:R_L unchanged by Brune}), and whose explicit expression is (\ref{eq::Brune physical level matrix}).

Note that searching for the general eigenvalues in (\ref{eq:Brune eigenproblem}) is equivalent to solving (apply the Sylvester determinant identity theorem, or c.f. eq. (49)-(50) in \cite{Brune_2002}):
\begin{equation}
\left.\mathrm{det}\left(\boldsymbol{R}_S^{-1}(E)\right)\right|_{E = \widetilde{E_i}} = 0
\label{eq:R_S by Brune det search}
\end{equation}
i.e. solving for the poles of the $\boldsymbol{R}_S$ operator defined as
\begin{equation}
\boldsymbol{R}_S^{-1} \coloneqq \boldsymbol{R}^{-1} + \boldsymbol{B} - \boldsymbol{S}
\label{eq:R_S by Brune}
\end{equation}
hence our dubbing of the Brune parametrization as the $\boldsymbol{R}_S$ parametrization. 
We here provide additional insights into this alternative $\boldsymbol{R}_S$ parametrization of the R-matrix scattering model, which will be foretelling of the structure of the $\boldsymbol{R}_L$ parametrization discussed in section \ref{sec:R_L Siegert and Humblet}.

The key insight is that in equation (22) of \cite{Brune_2002}, Brune builds a square matrix $\boldsymbol{g}\coloneqq \left[\boldsymbol{g_1}, \hdots, \boldsymbol{g_i} , \hdots, \boldsymbol{g_{N_\lambda}} \right]$, from which he is able to built the inverse physical level matrix in his equation (30) of \cite{Brune_2002}. Brune justifies that this matrix is indeed square in the paragraphs between equations (46) and (47) by a three-step monotony argument depicted in FIG. 1 of \cite{Brune_2002}: 1) he assumes $S_c(E)$ is continuous (i.e. has no real poles); 2) he assumes $\frac{ \partial S_c}{\partial E} \geq 0$, which is always true for negative energies and was just proved to be true for positive energies in the case of repulsive Coulomb interactions \cite{Brune_Mark_monotonic_properties_of_shift_2018} (a general proof is lacking for positive energy attractive Coulomb channels but has always been verified in practice); 3) he invokes the eigenvalue repulsion behavior (no-crossing rule) of which we hereafter provide a discussion in section \ref{subsubsec:: Semi-simple poles in R-matrix theory}.
If these three assumption are true, since the left-hand-side of (\ref{eq:Brune eigenproblem}) is a real symmetric matrix for any real energy value, then the spectral theorem guarantees there exists $N_\lambda$ real eigenvalues to it, and Brune's three assumptions above elegantly guarantee that there exists exactly $N_\lambda$ real solutions to the generalized eiganvalue problem (\ref{eq:Brune eigenproblem}).

\subsubsection{\label{subsubsec::Ambiguity in shift and penetration}Ambiguity in shift and penetration factors definition for complex wavenumbers}

There is a subtlety, however. A careful analysis reveals that the assumption that $S_c(E)$ is continuous or monotonously increasing is not unequivocal, and points to an open discussion in the field of nuclear cross section evaluations: the way of continuing the scattering matrix $\boldsymbol{U}$ to complex wavenumbers $k_c \in \mathbb{C}$.
Indeed, there is an ambiguity in the definition of the shift $S_c(E)$ and penetration $P_c(E)$ functions: two approaches are possible, and the community is not clear on which is correct.

The first, Lane \& Thomas approach is to define the shift and penetration functions as the real and imaginary parts of the the outgoing-wave reduced logarithmic derivative:
\begin{equation}
\forall E \in \mathbb{C}\;, \left\{ \quad \begin{array}{cc}
     \boldsymbol{S}(E) \coloneqq & \Re\left[\boldsymbol{L}(E)\right]  \; \; \in \; \mathbb{R} \\
     \boldsymbol{P}(E) \coloneqq & \Im\left[\boldsymbol{L}(E)\right]  \; \; \in \; \mathbb{R}
\end{array}   \right.
\label{eq:: Def S = Re[L], P = Im[L]}
\end{equation}
This definition, introduced in \cite{Lane_and_Thomas_1958} III.4.a. from equations (4.4) to (4.7c), finds its justification in the discussion between equations (2.1) and (2.2) of \cite{Lane_and_Thomas_1958} VII.2, as it presents the advantage of automatically closing the sub-threshold channels since:
\begin{equation}
\forall E<E_{T_c}\;, \quad  \Im\left[L_c(E)\right]= 0
\label{eq:: Im[L] = 0 for E<0}
\end{equation}
This elegant closure of channels comes at two costs: 1) the scattering matrix $\boldsymbol{U}$ is no longer analytic for complex wavenumbers $k_c \in \mathbb{C}$; 2) artificial poles are introduced to the scattering matrix, as we will show in theorem \ref{theo::Analytic continuation of scattering matrix cancels spurious poles} that analytic continuation of $\boldsymbol{O}$ and $\boldsymbol{I}$ --- and thus the operators $\boldsymbol{L}$, $\boldsymbol{S}$, and $\boldsymbol{P}$ --- is necessary to cancel the poles of $\boldsymbol{O}$ out of the scattering matrix $\boldsymbol{U}$.
In this Lane \& Thomas approach (\ref{eq:: Def S = Re[L], P = Im[L]}), the function calculated for $\boldsymbol{S}$ changes from $S(E) \coloneqq S_c(E)$ above threshold ($E\geq E_{T_c}$), to $S(E) \coloneqq L_c(E)$ below threshold ($E <  E_{T_c}$), because of (\ref{eq:: Im[L] = 0 for E<0}). 
Moreover, definition (\ref{eq:: Def S = Re[L], P = Im[L]}) induces ramifications for both the shift and the penetration factors, as we show in lemma \ref{lem:: Lane and Thomas S_c ramification properties}.

\begin{lem}\label{lem:: Lane and Thomas S_c ramification properties}
\textsc{Branch-point definition of shift $S_c(E)$ and penetration $P_c(E)$ functions}.\\
Definition (\ref{eq:: Def S = Re[L], P = Im[L]}) of the shift $S_c(E)$ and penetration $P_c(E)$ functions, legacy of Lane \& Thomas, entails:
\begin{itemize}
    \item branch-points for both $S_c(E)$ and $P_c(E)$, induced by the multi-sheeted nature of mapping (\ref{eq:rho_c(E) mapping}),
    \item on the $\big\{ E, - \big\}$ sheet below threshold $E < E_{T_c}$, the shift function $S_c(E)$ can present discontinuities and areas where $\frac{ \partial S_c }{\partial E}(E) < 0$,
    \item in particular, for neutral particles of odd angular momenta $\ell_c \equiv 1 \; (\mathrm{mod} \; 2)$, there is exactly one real sub-threshold pole to $S_c(E)$ on the $\big\{ E, - \big\}$ sheet,
    \item everywhere other than sub-threshold $\big\{ E, - \big\}$ sheet, and in particular on the $\big\{ E, + \big\}$ sheet, the shift function $S_c(E)$ is continuous and monotonously increasing:  $\frac{ \partial S_c }{\partial E}(E) \geq 0$.
\end{itemize}
\end{lem}

\begin{proof}
The proof simply introduces the branch-structure of the $\rho_c(E)$ mapping (\ref{eq:rho_c(E) mapping}), observable in figure \ref{fig:mapping rho - E}, into the Lane \& Thomas definition (\ref{eq:: Def S = Re[L], P = Im[L]}).
Historically, the study of the properties emanating from this definition have neglected the $\big\{ E, - \big\}$ sheet. 
Importantly, it was recently proved that $\frac{\partial S_c}{\partial E}(E)  \geq 0$ was true for most cases \cite{Brune_Mark_monotonic_properties_of_shift_2018}. This proof did not consider the $\big\{ E, - \big\}$ sheet of mapping (\ref{eq:rho_c(E) mapping}).
However, their proof of $\frac{\partial S_c}{\partial E} (E) \geq 0$ should still stand on the $\big\{ E, + \big\}$ sheet.
Moreover, the proof of lemma \ref{lem:: analytic S_c and P_c lemma} establishes that all the discontinuity points, i.e. the real-energy poles, happen at sub-threshold energies, and in particular that neutral particles with odd angular moment introduce exactly one such sub-threshold discontinuity. 
This means that above threshold, both the shift $S_c(E)$ and penetration $P_c(E)$ functions are continuous. 
These behaviors are depicted in figure \ref{fig:Lane_and_Thomas_shift_and_penetration_factors_with_branch_points}.
Finally, one will notice that the $\big\{ E, + \big\}$ and $\big\{ E, - \big\}$ sheets coincide above threshold for the shift function $S_c(E)$, and below threshold for the penetration function $P_c(E)$.
For $P_c(E)$, this is because of property (\ref{eq:: Im[L] = 0 for E<0}).
For $S_c(E)$, this is because for real energies above threshold, both definitions (\ref{eq:: Def S = Re[L], P = Im[L]}) and (\ref{eq:: Def S and P analytic continuation from L}) coincide, and lemma \ref{lem:: analytic S_c and P_c lemma} shows the analytic continuation definition of $S_c(E)$ is function of $\rho_c^2(E)$, which unfolds the sheets of the Rieman mapping (\ref{eq:rho_c(E) mapping}). Hence, for above-threshold energies, this property still stands for the Lane \& Thomas definition of the shift factor $S_c(E)$.
\end{proof}

\begin{figure}[ht!!] 
  \centering
  \includegraphics[width=0.50\textwidth]{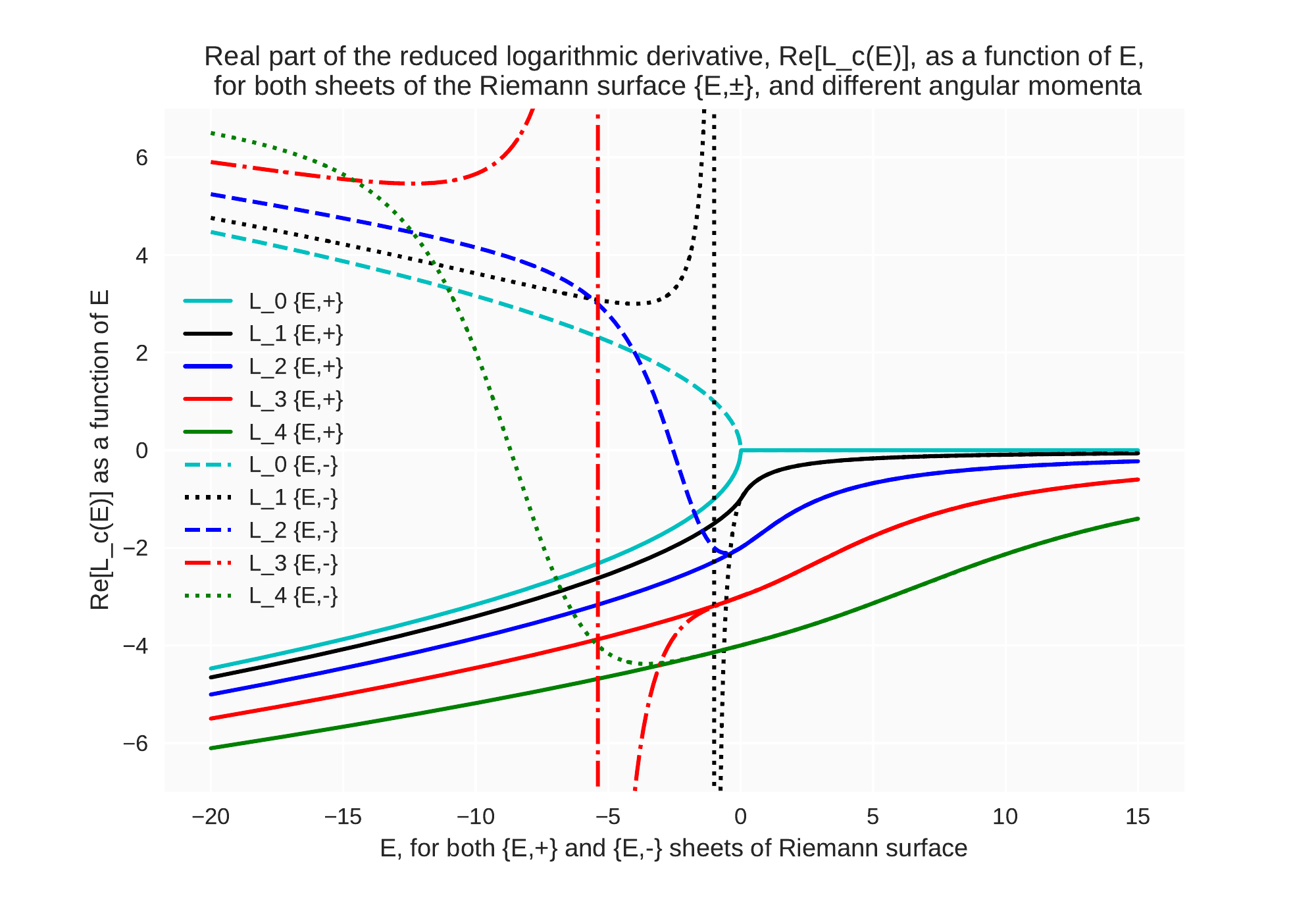}
  \includegraphics[width=0.50\textwidth]{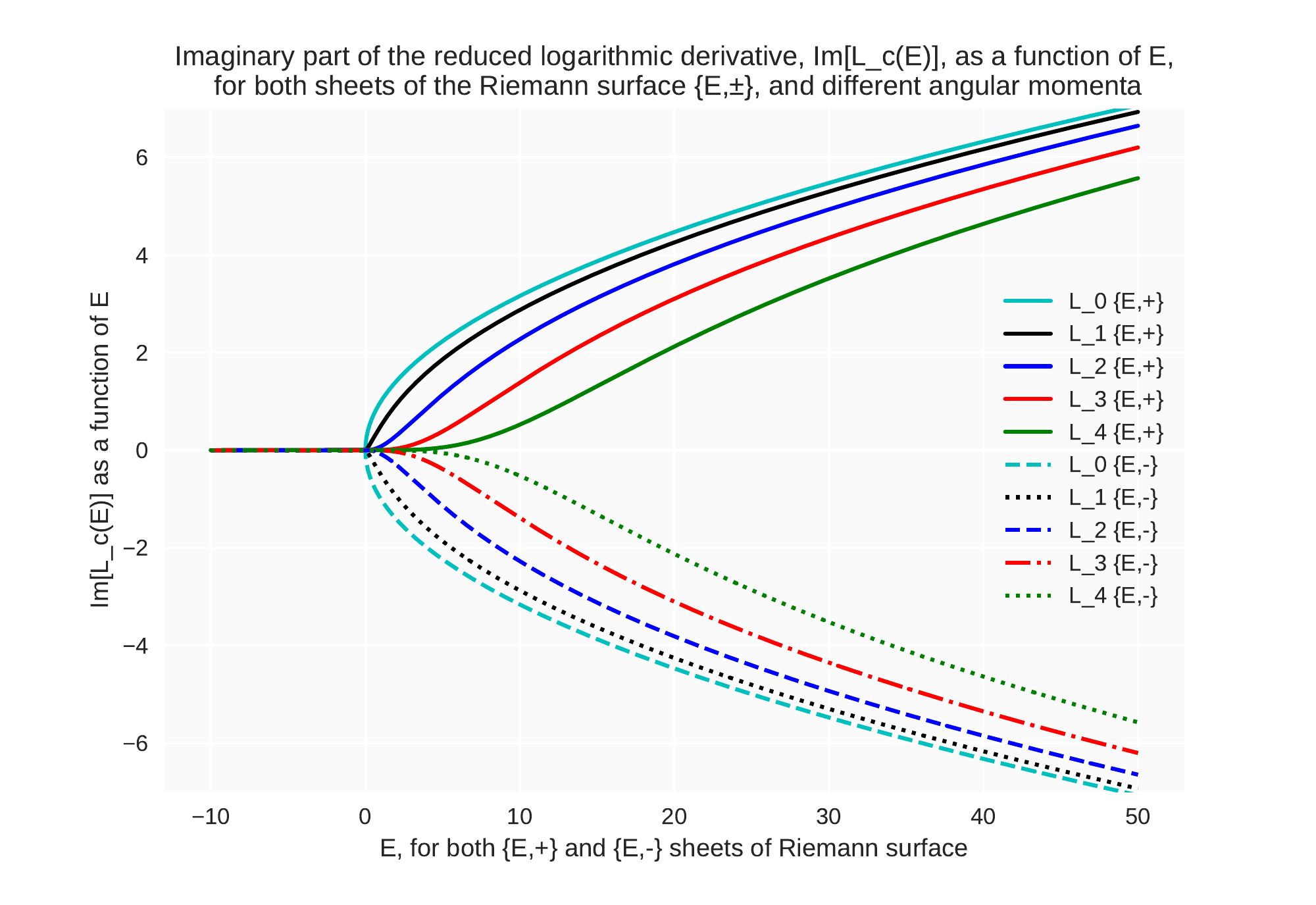}
  \caption{\small{Real and imaginary parts of the massive neutral particles reduced logarithmic derivative of the outgoing wavefunction, $\Re\big[L_\ell(E)\big]$ and $\Im\big[L_\ell(E)\big]$, for different angular momenta $\ell \in \llbracket 1, 4 \rrbracket$. This quantity was used by Lane \& Thomas to define the shift and penetration functions, $S_\ell(E) $ and $P_\ell(E)$, as (\ref{eq:: Def S = Re[L], P = Im[L]}). This definition commands branch points from mapping (\ref{eq:rho_c massive}) (c.f. figure \ref{fig:mapping rho - E}). $\Re\big[L_\ell(E)\big]$ presents sub-threshold discontinuities (for odd $\ell$) and non-monotonic behavior (for even $\ell$) below threshold on the $\big\{ E, - \big\}$ sheet.}}
  \label{fig:Lane_and_Thomas_shift_and_penetration_factors_with_branch_points}
\end{figure}

The second approach to defining the shift and penetration functions, $\boldsymbol{S}$ and $\boldsymbol{P}$, consists of performing analytic continuation of the scattering matrix $\boldsymbol{U}$ to complex energies $E\in\mathbb{C}$. This is implicit in the Kapur-Peirls or Siegert-Humblet expansions (c.f. section \ref{sec:R_L Siegert and Humblet}), and an abundant literature revolves around the analytic properties of the scattering matrix in the complex plane, including the vast Theory of Nuclear Reaction of Humblet and Rosenfeld \cite{Theory_of_Nuclear_Reactions_I_resonances_Humblet_and_Rosenfeld_1961,Theory_of_Nuclear_Reactions_II_optical_model_Rosenfeld_1961,Theory_of_Nuclear_Reactions_III_Channel_radii_Humblet_1961_channel_Radii,Theory_of_Nuclear_Reactions_IV_Coulomb_Humblet_1964,Theory_of_Nuclear_Reactions_V_low_energy_penetrations_Jeukenne_1965,Theory_of_Nuclear_Reactions_VI_unitarity_Humblet_1964,Theory_of_Nuclear_Reactions_VII_Photons_Mahaux_1965,Theory_of_Nuclear_Reactions_VIII_evolutions_Rosenfeld_1965,Theory_of_Nuclear_Reactions_IX_few_levels_approx_Mahaux_1965}, or the general unitarity condition on the multi-sheeted Riemann surface introduced by Eden and Taylor in \cite{Eden_and_Taylor}.
In this approach, energy dependence of the shift and penetration factors for positive energies are analytically continued into the complex plane, i.e.
\begin{equation}
\boldsymbol{S} : \left\{ \begin{array}{rcl}
\mathbb{C} & \mapsto & \mathbb{C} \\
E & \to & S_c(E)
\end{array}\right. \mathrm{s.t.} \; \; S(E) = S_c(E) , \; \forall (E-E_{T_c}) \in \mathbb{R}_+
\label{eq:: Def S analytical}
\end{equation}
so that they can be computed from the outgoing wavefunction reduced logarithmic derivative $\boldsymbol{L}$ by analytic continuation in wavenumber space $k_c \in \mathbb{C}$: 
\begin{equation}
\forall \rho_c \in \mathbb{C}\;, \left\{ \quad \begin{array}{cc}
     S_c(\rho_c) \coloneqq & \frac{L_c(\rho_c) +\left[ L_c(\rho_c^*)\right]^* }{2}  \; \; \in \; \mathbb{C} \\
     P_c(\rho_c) \coloneqq &  \frac{L_c(\rho_c) -\left[ L_c(\rho_c^*)\right]^* }{2\mathrm{i}} \; \; \in \; \mathbb{C}
\end{array}   \right.
\label{eq:: Def S and P analytic continuation from L}
\end{equation}
From this definition (\ref{eq:: Def S and P analytic continuation from L}), and using the recurrence relation (\ref{eq::L_ell recurrence formula}), one readily finds the expressions for the neutral particles shift and penetration factors documented in table \ref{tab::S_and_P_expressions_neutral}.
Critically, both definitions (\ref{eq:: Def S = Re[L], P = Im[L]}) and (\ref{eq:: Def S and P analytic continuation from L}) will yield the same shift $S_c(E)$ and penetration $P_c(E)$ functions for real energies above threshold $E \geq E_{T_c}$. 
Moreover, definition (\ref{eq:: Def S and P analytic continuation from L}) bestows interesting analytic properties onto the shift and penetration functions, here established in lemma \ref{lem:: analytic S_c and P_c lemma}.

\begin{table*}
\caption{\label{tab::S_and_P_expressions_neutral} Shift $S_\ell(\rho)$, $S_\ell^0(\rho) \coloneqq S_\ell(\rho) - B_\ell$ using $B_\ell = - \ell$, and $P_\ell(\rho)$ irreducible forms for neutral particles, for angular momenta $0 \leq \ell \leq 4$, all defined from analytic continuation (\ref{eq:: Def S and P analytic continuation from L}).}
\begin{ruledtabular}
\begin{tabular}{c|c|c|c}
\ \ & $S_\ell(\rho)$ & $S_\ell^0(\rho) \coloneqq S_\ell(\rho) - B_\ell$ (recurrence for $B_\ell = - \ell$ ) & $P_\ell(\rho)$  \tabularnewline
\hline
$\ell$  &  $ S_\ell(\rho) = \frac{\rho^2 \left(\ell - S_{\ell-1}(\rho)  \right)}{\left(\ell - S_{\ell-1}(\rho)  \right)^2 + P_{\ell-1}(\rho)^2} - \ell $  & $ S_\ell^0(\rho) \coloneqq S_\ell(\rho) + \ell =  \frac{\rho^2 \left(2\ell - 1 - S_{\ell-1}^0(\rho)  \right)}{\left(2\ell -1 - S_{\ell-1}^0(\rho)  \right)^2 + P_{\ell-1}(\rho)^2}$   & $ P_\ell(\rho) = \frac{\rho P_{\ell-1}(\rho) }{\left(\ell - S_{\ell-1}(\rho)  \right)^2 + P_{\ell-1}(\rho)^2}$  \tabularnewline
\hline \hline
0  &  $0$  & $0$ & $\rho$ \tabularnewline
1  &  $ - \frac{1}{1+\rho^2}$ & $\frac{\rho^2}{1+\rho^2}$ & $\frac{\rho^3}{1+\rho^2}$ \tabularnewline
2   &  $ - \frac{18 + 3 \rho^2}{9 + 3\rho^2 + \rho^4}$ & $\frac{ 3 \rho^2 + 2 \rho^4}{9 + 3\rho^2 + \rho^4}$ & $\frac{\rho^5}{9 + 3\rho^2 + \rho^4}$ \tabularnewline
3  &  $ - \frac{675 + 90 \rho^2 + 6\rho^4}{225 + 45\rho^2 + 6\rho^4 + \rho^6}$ & $\frac{45 \rho^2 + 12\rho^4 + 3 \rho^6}{225 + 45\rho^2 + 6\rho^4 + \rho^6}$ &  $\frac{\rho^7}{225 + 45\rho^2 + 6\rho^4 + \rho^6}$ \tabularnewline
4  &  $ - \frac{44100 + 4725 \rho^2 + 270\rho^4 + 10 \rho^6}{11025 + 1575\rho^2 + 135\rho^4 + 10\rho^6 + \rho^8}$ & $ \frac{1575 \rho^2 + 270\rho^4 + 30 \rho^6 + 4\rho^8}{11025 + 1575\rho^2 + 135\rho^4 + 10\rho^6 + \rho^8}$ &  $\frac{\rho^9}{11025 + 1575\rho^2 + 135\rho^4 + 10\rho^6 + \rho^8}$\tabularnewline
\end{tabular}
\end{ruledtabular}
\end{table*}

\begin{lem}\label{lem:: analytic S_c and P_c lemma}
\textsc{Analytic continuation definition of shift $S_c(E)$ and penetration $P_c(E)$ functions}. \\
When defined by analytic continuation (\ref{eq:: Def S and P analytic continuation from L}), the shift function, $S_c(\rho)$, satisfies the Mittag-Leffler expansion:
\begin{equation}
\begin{IEEEeqnarraybox}[][c]{rcl}
    S_c(\rho) & \;  = \; & - \ell + \underset{\mathrm{arg}(\omega_n)\in \left[-\frac{\pi}{2},0\right]}{\sum_{n \geq 1}} \frac{\rho^2}{\rho^2 - \omega_n^2} + \frac{\rho^2}{\rho^2 - {\omega_n^*}^2}
\IEEEstrut\end{IEEEeqnarraybox}
\label{eq::S_c expansion in rho^2}
\end{equation}
where the poles $\big\{ \omega_n \big\}$ are only the lower-right-quadrant roots -- i.e. such that $\mathrm{arg}(\omega_n) \in [-\frac{\pi}{2},0]$ -- of the outgoing wave function $O_c(\rho_c)$.
In the neutral particles cases, these are reported in table \ref{tab::roots of the outgoing wave functions}.
Given $\rho_c(E)$ mapping (\ref{eq:rho_c(E) mapping}), this entails $S_c(E)$:
\begin{itemize}
    \item unfolds the sheets of $\rho_c(E)$ mapping (\ref{eq:rho_c(E) mapping}),
    \item is purely real for real energies: $\forall E\in\mathbb{R}, \; S_c(E) \in \mathbb{R}$.
\end{itemize}
The penetration function, $P_c(\rho)$, satisfies the Mittag-Leffler expansion:
\begin{equation}
\begin{IEEEeqnarraybox}[][c]{rcl}
    P_c(\rho) & \;  = \; & \rho \Bigg[ 1 - \mathrm{i} \underset{\mathrm{arg}(\omega_n)\in \left[1 \frac{\pi}{2},0\right]}{\sum_{n \geq 1}} \frac{\omega_n}{\rho^2 - \omega_n^2} - \frac{{\omega_n^*}}{\rho^2 - {\omega_n^*}^2} \Bigg]
\IEEEstrut\end{IEEEeqnarraybox}
\label{eq::P_c expansion in rho^2}
\end{equation}
which in turn entails that $P_c(E)$:
\begin{itemize}
    \item is purely real for above threshold energies: $\forall E > E_{T_c}, \;  P_c(E) \in \mathbb{R}$,
    \item is purely imaginary for sub-threshold energies: $\forall E < E_{T_c}, \;  P_c(E) \in \mathrm{i}\mathbb{R}$,
\end{itemize}
In the neutral particles case, Mittag-Leffler expansions (\ref{eq::S_c expansion in rho^2}) and (\ref{eq::P_c expansion in rho^2}) are the partial fraction decompositions of the rational fractions reported in table \ref{tab::S_and_P_expressions_neutral}, and for all odd angular momenta $\ell_c \equiv 1 \; (\mathrm{mod} \; 2)$, both have one, shared, real sub-threshold pole.
\end{lem}

\begin{proof}
The proof uses lemma \ref{lem::Mittag-Leffler of L_c Lemma}, where we establish the Mittag-Leffler expansion (\ref{eq::Mittag-Leffler expansion of L_c}) of the reduced logarithmic derivative $L_c(\rho_c)$.
Using the conjugacy properties (\ref{eq:Conjugacy for O and I}) on the poles $\big\{ \omega_n \big\}$ means each pole $\omega_n$ on the lower right quadrant of the complex plane -- i.e. such that $\mathrm{arg}(\omega_n) \in [-\frac{\pi}{2},0]$ -- induces a specular pole $-\omega_n^*$.
Dividing the poles in specular pairs, we can re-write the Mittag-Leffler expansion (\ref{eq::Mittag-Leffler expansion of L_c}) as:
\begin{equation}
\begin{IEEEeqnarraybox}[][c]{rcl}
    L_c(\rho) & \;  = \; & - \ell + \mathrm{i}\rho + \underset{\mathrm{arg}(\omega_n)\in \left[-\frac{\pi}{2},0\right]}{\sum_{n \geq 1}} \frac{\rho}{\rho - \omega_n} + \frac{\rho}{\rho + {\omega_n^*}}
\IEEEstrut\end{IEEEeqnarraybox}
\label{eq::L_c Mittag-Lellfer expansion in pairs}
\end{equation}
Plugging-in expression (\ref{eq::L_c Mittag-Lellfer expansion in pairs}) into the shift function definition (\ref{eq:: Def S and P analytic continuation from L}) readily yields (\ref{eq::S_c expansion in rho^2}) and (\ref{eq::P_c expansion in rho^2}).

Note that (\ref{eq::S_c expansion in rho^2}) unfolds the Riemann surface of mapping (\ref{eq:rho_c(E) mapping}), whereas (\ref{eq::P_c expansion in rho^2}) factors-out the branch points so that all its branches are symmetric.
In (\ref{eq::P_c expansion in rho^2}) we recognize the odd powers of $\rho$ in the neutral particles case of table \ref{tab::S_and_P_expressions_neutral}, which do not unfold the Riemann sheets of mapping (\ref{eq:rho_c(E) mapping}).
These behaviors are illustrated in figure \ref{fig:Analytic_continuation_S_and_P}.

In the neutral particles case, $L_c$ is a rational fraction in $\rho_c$, and its denominator is of degree $\ell_c$, as can be observed in table \ref{tab::L_values_neutral}, thus inducing $\ell_c $ poles, reported in table \ref{tab::roots of the outgoing wave functions}.
Since these poles $\big\{ \omega_n \big\}$ must respect the specular symmetry: $ \omega \longleftrightarrow -\omega_n^*$; it thus entails that these poles come in symmetric pairs. 
For neutral particles, odd angular momenta mean there is an odd number of poles $\big\{ \omega_n \big\}$. For them to come in pairs thus imposes one is exactly imaginary $\omega_n = - \mathrm{i} x_n$, with $x_n \in \mathbb{R}_+$.
When squared, this purely imaginary pole will introduce a purely sub-threshold pole in both (\ref{eq::S_c expansion in rho^2}) and (\ref{eq::P_c expansion in rho^2}), though: $ \frac{1}{\rho^2 + x_n^2}$.
\end{proof}

\begin{figure}[ht!!] 
  \centering
  \includegraphics[width=0.50\textwidth]{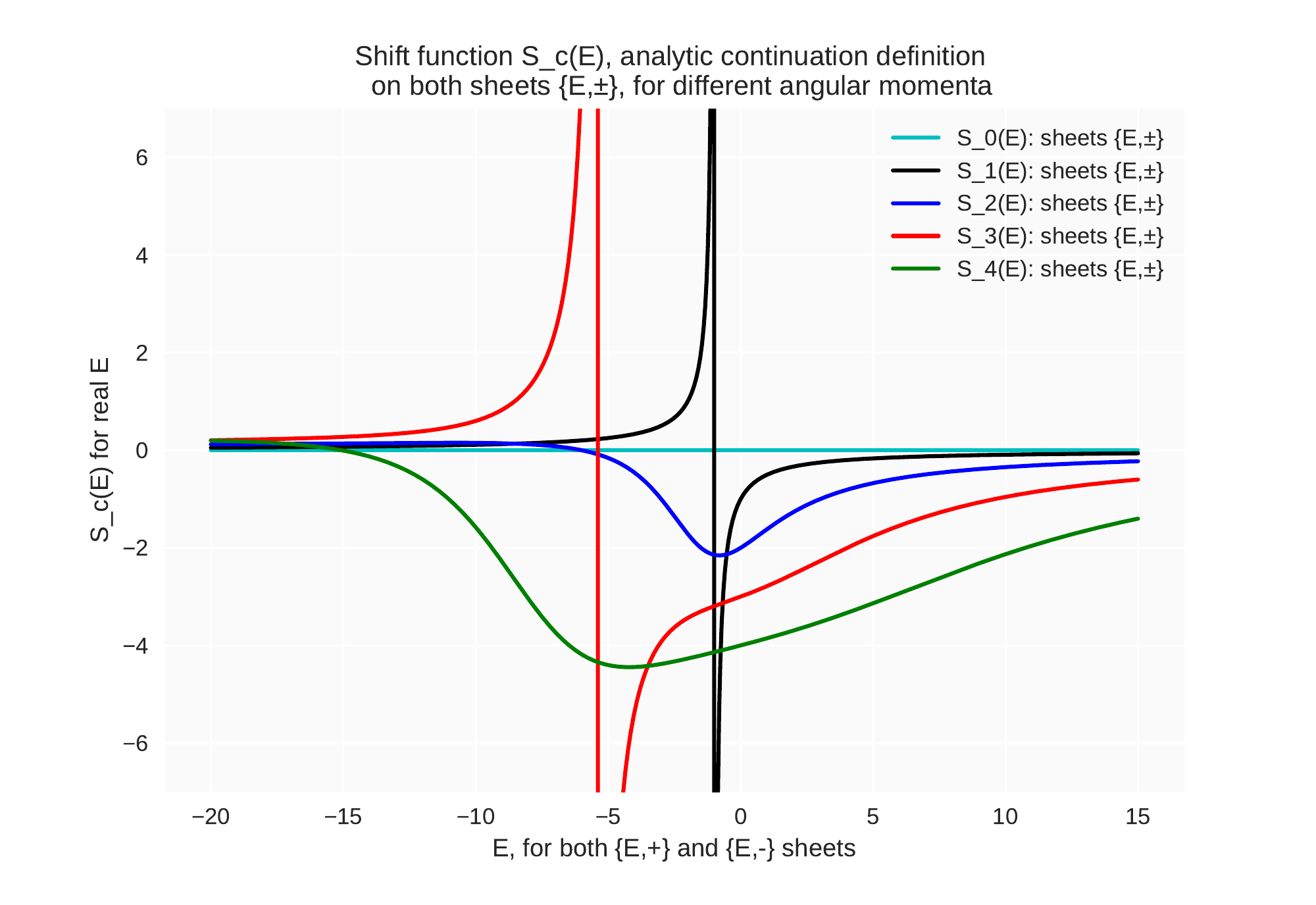}
  \includegraphics[width=0.50\textwidth]{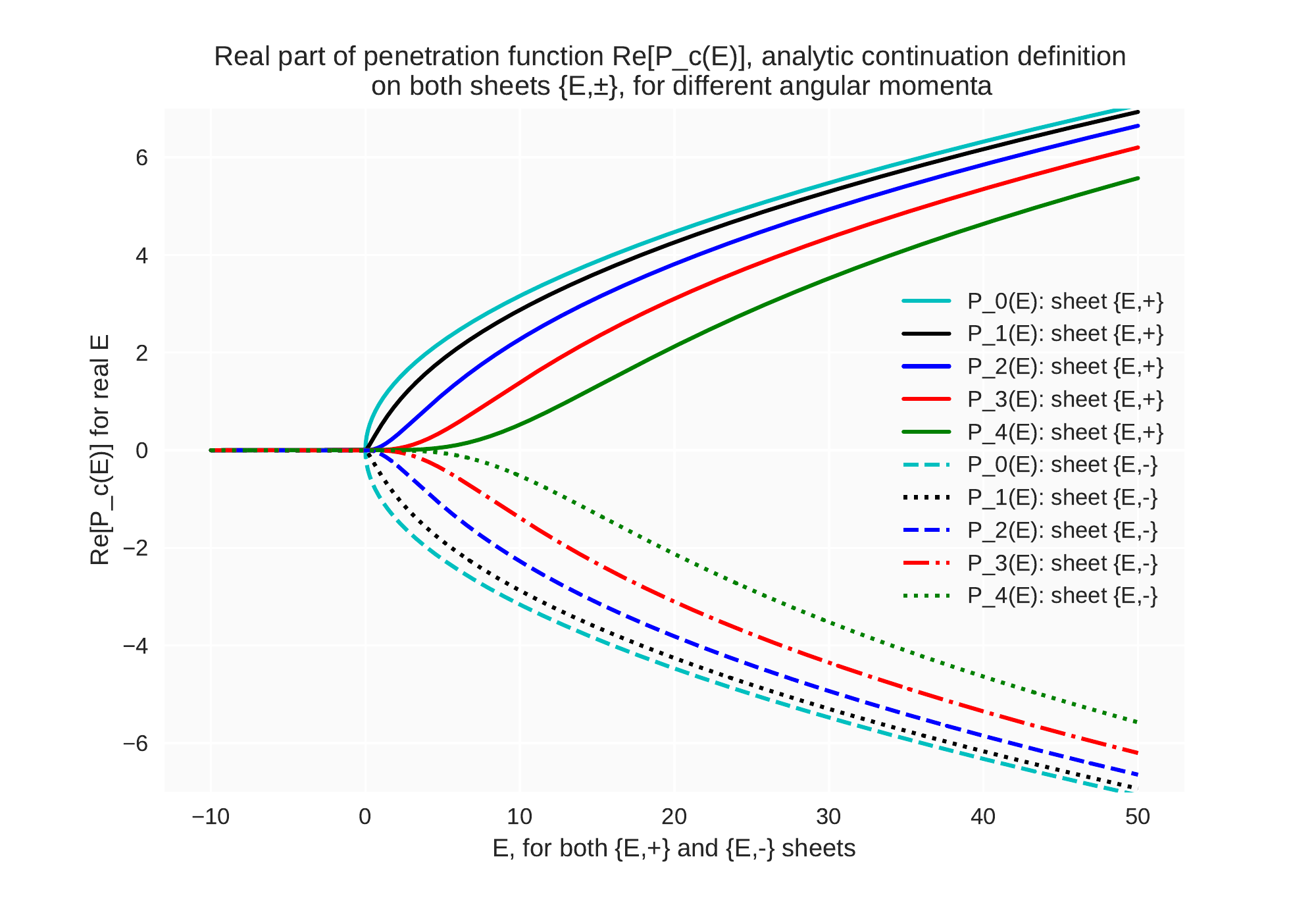}
  \includegraphics[width=0.50\textwidth]{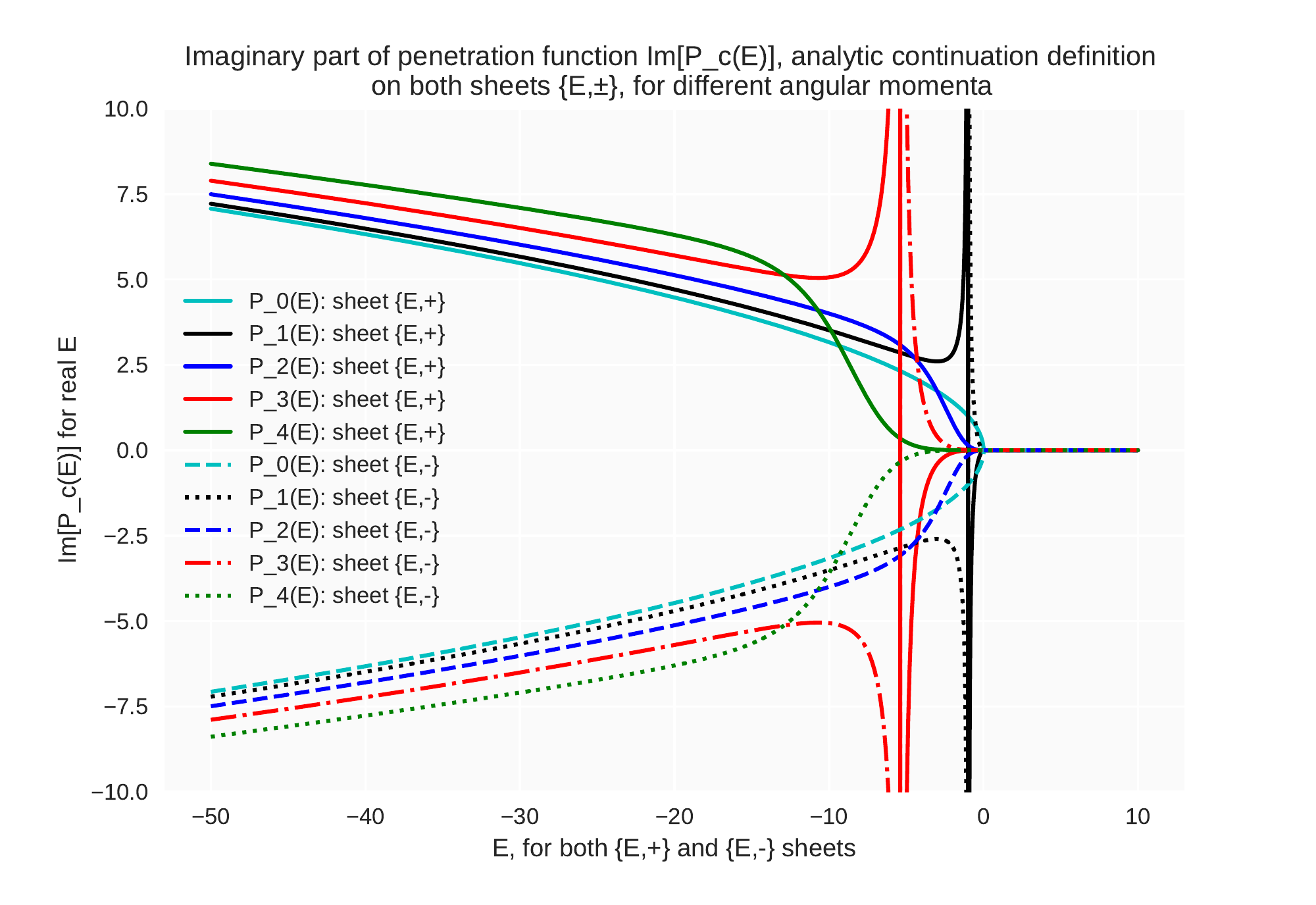}
  \caption{\small{Shift $S_c(E)$ and penetration $P_c(E)$ functions for massive neutral particles, as defined by analytic continuation (\ref{eq:: Def S and P analytic continuation from L}), for different angular momenta $\ell_c \in \llbracket 0, 4 \rrbracket$. This definition induces no branch points for the shift function $S_c(E)$, as it unfolds the sheets of mapping (\ref{eq:rho_c(E) mapping}), in this non-relativistic massive particles case (\ref{eq:rho_c massive}), as shown in lemma \ref{lem:: analytic S_c and P_c lemma}. One can observe discontinuities (for odd angular momenta) and non-monotonic behavior (for even angular momenta) for sub-threshold energies. $P_c(E)$ is purely real, with branches, above threshold; and purely imaginary, with branches, below threshold.}}
  \label{fig:Analytic_continuation_S_and_P}
\end{figure}

An example to illustrate the difference between definitions (\ref{eq:: Def S = Re[L], P = Im[L]}) and (\ref{eq:: Def S and P analytic continuation from L}) is depicted in figures \ref{fig:Lane_and_Thomas_shift_and_penetration_factors_with_branch_points} and \ref{fig:Analytic_continuation_S_and_P}.
Consider the elemental case of a neutron channel with angular momentum $\ell_c = 1$, and let $\rho_0$ be the proportionality constant so that (\ref{eq:rho_c massive}) is written $\rho(E) = \pm \rho_0 \sqrt{E - E_{T_c}}$. Let us also set a zero threshold $E_{T_c} = 0$, for simplicity.

In this case, the legacy Lane \& Thomas definition (\ref{eq:: Def S = Re[L], P = Im[L]}) corresponds to taking $S(E) \coloneqq S_c(\rho_c(E)) = -\frac{1}{1+\rho_c^2}$ for above-threshold energies $E\geq E_{T_c}$, and switch to $S(E) \coloneqq L_c(\rho_c(E)) = \frac{- 1 + \mathrm{i}\rho_c + \rho_c^2 }{1- \mathrm{i}\rho_c}$ for sub-threshold energies $E < E_{T_c}$.
Since the (\ref{eq:rho_c massive}) mapping $\rho(E) = \pm \rho_0 \sqrt{E - E_{T_c}}$ has two sheets, this means definition (\ref{eq:: Def S = Re[L], P = Im[L]}) entails: $S(E) \coloneqq S_c(E) = -\frac{1}{1+\rho_0^2 E}$ for $E\geq E_{T_c}$, and $S(E) \coloneqq L_c(E) = \frac{- 1 \pm \mathrm{i}\rho_0 \sqrt{E} + \rho_0^2 E }{1  \mp \mathrm{i}\rho_0 \sqrt{E}}$ for $E < E_{T_c}$, which is a real quantity. Definition (\ref{eq:: Def S = Re[L], P = Im[L]}) thus introduces the ramifications reported in figure \ref{fig:Lane_and_Thomas_shift_and_penetration_factors_with_branch_points}. In particular, the full cyan line of our $\Re\big[L_c(E)\big]$ plot corresponds to the uncharged case for angular momentum $\ell = 0$ reported as a black curve in FIG.1, p.6 of \cite{Brune_Mark_monotonic_properties_of_shift_2018}. 
Notice that all the $\big\{ E,+ \big\}$ curves are continuous and monotonically increasing ($\frac{\partial S_c}{\partial E} \geq 0$), which is in accordance to the monotonic properties established in  \cite{Brune_Mark_monotonic_properties_of_shift_2018}. 
However, on the $\big\{E,-\big\}$ sheet below threshold, $\Re\big[L_c(E)\big]$ is no longer monotonic for even angular momenta ($\frac{\partial \Re\big[L_c(E)\big]}{\partial E} \geq 0$ does not hold), and is discontinuous in the case of odd angular momenta.

In contrast, for our same elemental case, the analytic continuation definition (\ref{eq:: Def S and P analytic continuation from L}) simply defines $S(E) \coloneqq S_c(\rho_c(E)) = -\frac{1}{1+\rho_c^2}$ for all real or complex energies $E \in \mathbb{C}$, that is $S(E) \coloneqq -\frac{1}{1+\rho_0^2 E}$. The later happens to have a real pole, which introduces a discontinuity, at $E_{\mathrm{dis.}} = -\frac{1}{\rho_0^2}$, as can be seen in figure \ref{fig:Analytic_continuation_S_and_P}.
One can observe that all odd angular momenta are monotonous but have a real sub-threshold pole.
For even angular momenta, $S_\ell(E)$ is continuous, monotonically increasing above-threshold, but $\frac{\partial S}{\partial E}(E) \geq 0$ does not hold below-threshold.
For the penetration function $P_c(E)$, each ramification is monotonous, but in opposite, mirror direction. 
In figure \ref{fig:Analytic_continuation_S_and_P}, the shift function $S_c(E)$ does not present branch points, as proved in lemma \ref{lem:: analytic S_c and P_c lemma}: it is a function of $\rho^2$ so no $\pm \sqrt{\cdot}$ choice is necessary in $\rho_c(E)$ mapping (\ref{eq:rho_c EDA}).

\subsubsection{\label{subsubsec::Number of Brune poles}Number of Brune poles}

Definitions (\ref{eq:: Def S = Re[L], P = Im[L]}) and (\ref{eq:: Def S analytical}) have a major impact on the Brune parameters (\ref{eq:Brune parameters}): they command that the number $N_S$ of Brune poles $\left\{\widetilde{E_i}\right\}$, solutions to Brune's generalized eigenproblem (\ref{eq:Brune eigenproblem}), is greater than the $N_\lambda$ previously found in \cite{Brune_2002}: i.e. $N_S \geq N_\lambda$.
And this is regardless of whether definition (\ref{eq:: Def S = Re[L], P = Im[L]}) or (\ref{eq:: Def S analytical}) is chosen for the shift factor $S_c(E)$ when searching for these solutions.

The fundamental reason for this is that Brune's three-step monotony argument, which elegantly proved in \cite{Brune_2002} that there are exactly $N_\lambda$ solutions to (\ref{eq:Brune eigenproblem}) and which we here recall in the last paragraph of section \ref{sec:R_S def}, rests on two hypotheses on the shift function $S_c(E)$: 1) it is continuous (i.e. has no real poles), and; 2) it is monotonously increasing, i.e. $\frac{ \partial S_c}{\partial E} \geq 0$.
In \cite{Brune_Mark_monotonic_properties_of_shift_2018}, these two hypotheses have just been proved to hold true for energies above threshold $E \geq E_{T_c}$, i.e. for real wavenumbers $k_c \in \mathbb{R}$.
Yet, we just established in lemmas \ref{lem:: Lane and Thomas S_c ramification properties} and \ref{lem:: analytic S_c and P_c lemma} that proper accounting of the multi-sheeted nature of the Riemann surface created by mapping (\ref{eq:rho_c(E) mapping}) shows these two hypotheses do not hold for sub-threshold energies $E < E_{T_c}$, where the wavenumber is purely imaginary from mapping (\ref{eq:rho_c massive}). 
This engenders additional solutions to Brune's generalized eigenproblem (\ref{eq:Brune eigenproblem}), so that the number $N_S$ of Brune poles $\left\{\widetilde{E_i}\right\}$ is in fact greater than the number of channels: $N_S \geq N_\lambda$.
So how many $N_S$ solutions are there? 
This depends on the R-matrix parameters and on the definition chosen for the shift function $S_c(E)$, as we now show in theorems \ref{theo::shadow_Brune_poles} and \ref{theo::analytic_Brune_poles}, for definitions (\ref{eq:: Def S = Re[L], P = Im[L]}) and (\ref{eq:: Def S analytical}), respectively.

\begin{theorem}\label{theo::shadow_Brune_poles} \textsc{Shadow Brune Poles}. \\
Let the branch Brune poles $\left\{\widetilde{E_i}\right\}$ be the solutions of the Brune generalized eigenproblem (\ref{eq:Brune eigenproblem}), using the legacy Lane \& Thomas definition (\ref{eq:: Def S = Re[L], P = Im[L]}) for the shift $S_c(E)$, and let $N_S$ be the number of such solutions, then:
\begin{itemize}
    \item all the branch Brune poles are real, and live on the $2^{N_c}$ sheets of the Riemann surface from (\ref{eq:rho_c(E) mapping}) mapping: $\left\{\widetilde{E_i}, \underbrace{\pm , \hdots , \pm  }_{N_c}\right\} \in \mathbb{R}^{N_S}$,
    \item exactly $N_\lambda$ branch Brune poles are present on the $\big\{ E, \underbrace{+ , \hdots , + }_{N_c} \big\}$ sheet of mapping (\ref{eq:rho_c(E) mapping}),
    \item additional shadow Brune poles can be found below threshold, $E<E_{T_c}$, on the $\big\{ E, - \big\}$ sheets of mapping (\ref{eq:rho_c(E) mapping}), depending on the values of the resonance parameters $\big\{E_\lambda, \gamma_{\lambda c}, B_c, E_{T_c}, a_c \big\}$ -- though in a way that is invariant under change of boundary-condition $B_c$,
    \item each neutral particle, odd angular momentum $\ell_c \equiv 1 \; (\mathrm{mod}\; 2)$, channel adds at least one shadow Brune pole below threshold on its $\big\{E, - \big\}$ sheet,
\end{itemize}
so that the total number $N_S^{\pm}$ of branch Brune poles on all sheets of mapping (\ref{eq:rho_c(E) mapping}) is greater than the number $N_\lambda$ of levels: $N_S^{\pm} \geq N_\lambda$.
\end{theorem}

\begin{proof}
Let us go about solving the Brune generalized eigenproblem (\ref{eq:Brune eigenproblem}), following the three-step argument of Brune (c.f. last paragraph of section \ref{sec:R_S def}).
We consider the left-hand side of (\ref{eq:Brune eigenproblem}).
According to definition (\ref{eq:: Def S = Re[L], P = Im[L]}), the shift function is always real, even for complex wavenumbers $k_c \in \mathbb{C}$. 
Since by construction the Wigner-Eisenbud R-matrix parameters $\big\{E_\lambda, \gamma_{\lambda c}, B_c, E_{T_c}, a_c \big\}$ are also all real, this implies the right-hand side must be real to solve (\ref{eq:Brune eigenproblem}).
Thus, all branch Brune poles from definition (\ref{eq:: Def S = Re[L], P = Im[L]}) are real.
To find them, we follow Brune's approach: for any energy $E$, on any of the $2^{N_c}$ sheets of mapping (\ref{eq:rho_c(E) mapping}), the left-hand side is a real symmetric matrix, and its eigenvalue decomposition will thus yield $N_\lambda$ real eigenvalues: $\big\{\widetilde{E_i}(E)\big\} \in \mathbb{R}$. 
We then have to vary the $E$ value until these real eigenvalues cross the $E=E$ identity line in the right-hand side. 
In general, the full accounting of all the Riemann sheets from mapping (\ref{eq:rho_c(E) mapping}) will entail solutions of the generalized Brune eigenproblem (\ref{eq:Brune eigenproblem}) on all sheets.
These branch Brune poles should thus be reported with the choice of sheet from the mapping (\ref{eq:rho_c(E) mapping}) for each channel, as in (\ref{eq:: pole E_j sheet reporting}) for the poles of the Kapur-Peierls operator of section \ref{sec:R_L parameters}: $\left\{\widetilde{E_i}, +, -, \hdots, + \right\} $.

We state in lemma \ref{lem:: Lane and Thomas S_c ramification properties} than on the $\big\{ E, + \big\}$ sheet, $S_c(E)$ is indeed continuous and monotonously increasing. 
We can thus apply Brune's three-step argument: the $N_\lambda$ eigenvalues of the left-hand side of (\ref{eq:Brune eigenproblem}) will satisfy $\frac{ \partial \widetilde{E_i}}{\partial E} (E) \leq 0$, and thus each and every one of them will eventually cross the $E=E$ identity line exactly once as $E$ varies continuously. 
On the $\big\{ E, + \big\}$ sheet for all channels, there are thus exactly $N_\lambda $ Brune poles: $\left\{\widetilde{E_i}, \underbrace{+ , \hdots , +  }_{N_c}\right\} \in \mathbb{R}^{N_\lambda}$

However, we showed in lemma \ref{lem:: Lane and Thomas S_c ramification properties} that $S_c(E)$ is not monotonous and can be discontinuous for sub-threshold energies $E < E_{T_c}$ on the $\big\{ E, - \big\}$ sheet.
So how many Brune poles are there on all sheets? 
Unfortunately, the number of solutions to Brune's generalized eigenproblem (\ref{eq:Brune eigenproblem}) will depend on the values of the resonance parameters $\big\{E_\lambda, \gamma_{\lambda c}, B_c, E_{T_c}, a_c \big\}$ -- though in a way that is invariant under change of boundary-condition $B_c$, as made evident in (\ref{eq:R_S by Brune det search}) when considering invariance (\ref{eq:: R_B invariance for B'}).
That the number of solutions to (\ref{eq:Brune eigenproblem}) depends on the parameters can be observed in figure \ref{fig:Brune_argument_even_angular_momenta}. 

For neutral particles odd momenta $\ell_c \equiv 1 \; (\mathrm{mod} \; 2)$ channels, lemma \ref{lem:: Lane and Thomas S_c ramification properties} also showed there exist exactly one sub-threshold pole to $S_c(E)$ on the $\big\{ E, - \big\}$ sheet of mapping (\ref{eq:rho_c(E) mapping}).
This pole will automatically cross the $E=E$ line of Brune's three-step argument twice, once below and once above threshold, adding an additional shadow Bune pole to the $N_\lambda$ Brune found in \cite{Brune_2002}.
This proves that there exists shadow Brune poles, just as shadow poles in the Siegert-Humblet parameters were revealed by G.Hale in \cite{Hale_1987, anti-Hale_1987}. 
This behavior is illustrated in figure \ref{fig:Brune_argument_odd_angular_momenta}.
\end{proof}

\begin{figure}[ht!!] 
  \centering
  \includegraphics[width=0.50\textwidth]{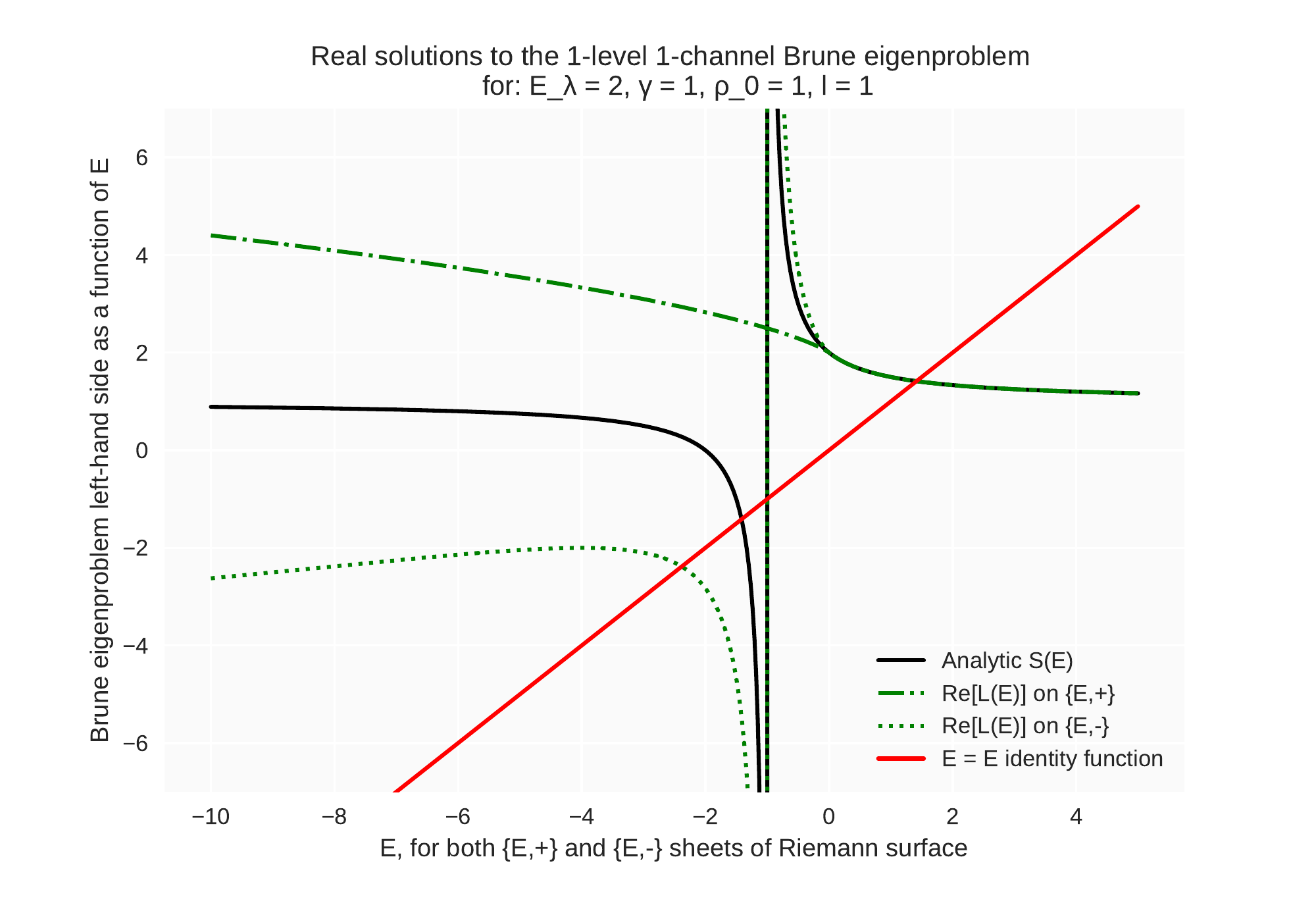}
  \includegraphics[width=0.50\textwidth]{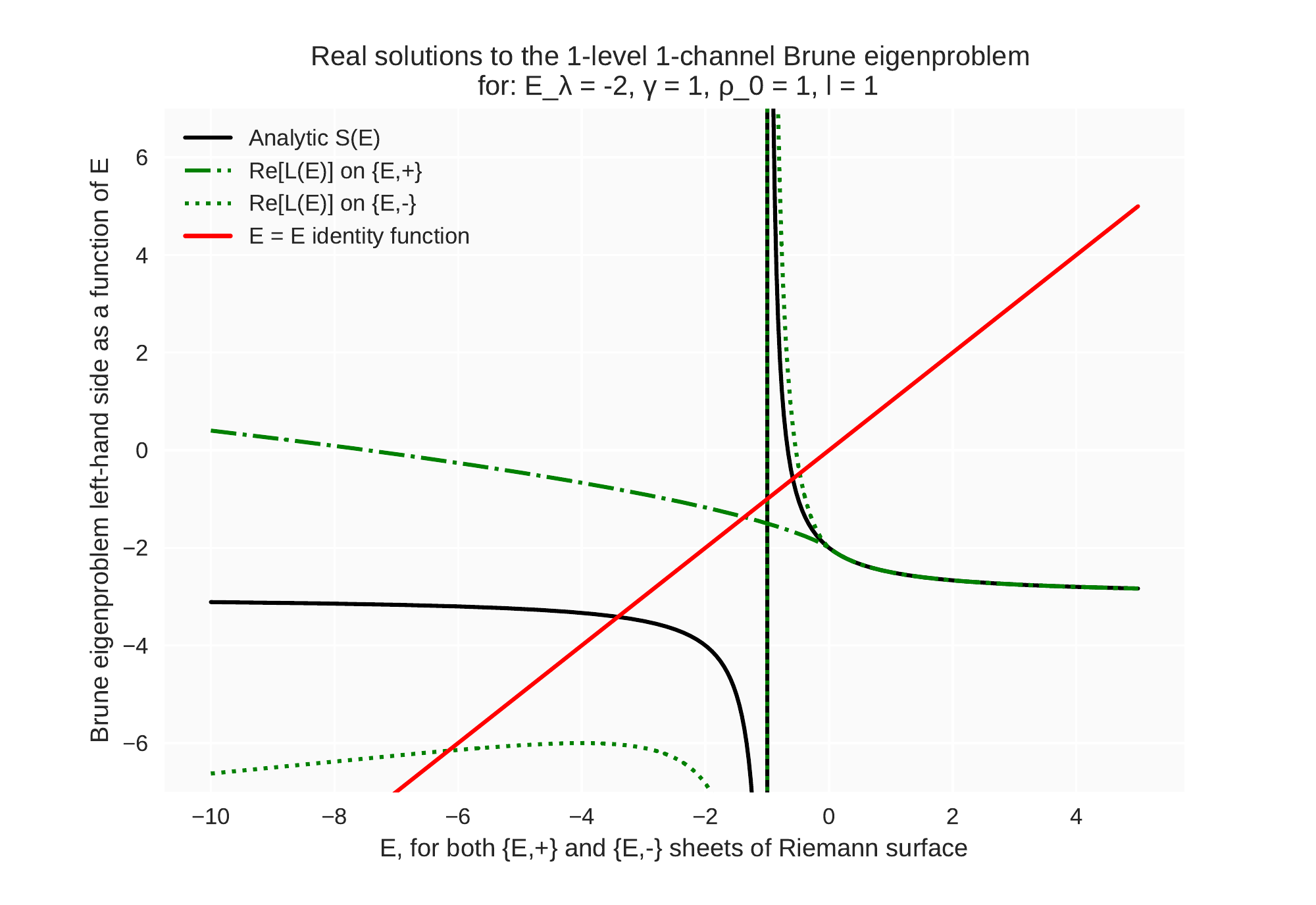}
  \caption{\small{Elemental Brune eigenproblem (\ref{eq:: elemental Brune problem}): comparison of solutions from definitions (\ref{eq:: Def S = Re[L], P = Im[L]}) versus (\ref{eq:: Def S and P analytic continuation from L}), for angular momentum $\ell_c = 1$, neutral particles, using $B_c = - \ell_c$ convention and zero threshold $E_{T_c}$.  Since both have a real sub-threshold poles, both will yield two solutions (crossing the E=E diagonal), one above and one below the discontinuity. If at threshold energy $E_{T_c}$ the left hand side of (\ref{eq:: elemental Brune problem}) is above the E=E diagonal, then the above-threshold solutions from both definitions coincide. In any case, the sub-threshold solutions differ. Behavior is analogous for all odd angular momenta $\ell_c \equiv 1 (\mathrm{mod} 2)$.}}
  \label{fig:Brune_argument_odd_angular_momenta}
\end{figure}

Theorem \ref{theo::shadow_Brune_poles} establishes the existence of sub-threshold shadow Brune poles when the legacy Lane \& Thomas definition (\ref{eq:: Def S = Re[L], P = Im[L]}) is chosen for the shift function $S_c(E)$. 
If instead the analytic continuation definition (\ref{eq:: Def S analytical}) is chosen, we now show in theorem \ref{theo::analytic_Brune_poles} that this unfolds the Riemann surface for the shift function $S_c(E)$ so that no branch points are required to define the Brune parameters.
In contrast, a drawback of definition (\ref{eq:: Def S and P analytic continuation from L}) is that it does not automatically close the channels of the scattering matrix for negative energies, though we propose a solution to this problem in theorems \ref{theo::evanescence of sub-threshold transmission matrix} and \ref{theo::Analytic continuation annuls sub-threshold cross sections}.
We argue in this article that the analytic continuation approach (\ref{eq:: Def S and P analytic continuation from L}) is the physically correct one.
This is because of the commendable argument that analytic continuation cancels out of the scattering matrix $\boldsymbol{U}(E)$ non-physical poles otherwise introduced by the Lane \& Thomas approach (\ref{eq:: Def S = Re[L], P = Im[L]}), as we will demonstrate in theorem \ref{theo::Analytic continuation of scattering matrix cancels spurious poles}.
Though there is no absolute consensus yet amongst the community as to which approach ought to be valid, both yield identical results for real energies above threshold (real wavenumbers $k_c \in \mathbb{R}$).

\begin{theorem}\label{theo::analytic_Brune_poles} \textsc{Analytic Brune Poles}. \\
Let the analytic Brune poles $\left\{\widetilde{E_i}\right\}$ be the solutions of the Brune generalized eigenproblem (\ref{eq:Brune eigenproblem}), using the analytic continuation definition (\ref{eq:: Def S and P analytic continuation from L}) for the shift $S_c(E)$, and let $N_S$ be the number of such solutions, then:
\begin{itemize}
    \item the analytic Brune poles are in general complex, and live on the single sheet of the unfolded Riemann surface from (\ref{eq:rho_c(E) mapping}) mapping: $\left\{\widetilde{E_i} \right\} \in \mathbb{C}^{N_S}$, 
    \item in the neutral particle case, there are exactly $N_S$ complex analytic Brune poles with:
\begin{equation}
N_S = N_\lambda + \sum_{c=1}^{N_c}\ell_c
\label{eq:N_S Brune poles}
\end{equation}
    \item in the charged particles case, there is a countable infinity of complex analytic Brune poles: $N_S = \infty$,
    \item the number $N_S^{\mathbb{R}}$ of real analytic Brune poles, $\left\{\widetilde{E_i} \right\} \in \mathbb{R}^{N_S^{\mathbb{R}}}$, is greater than the number of levels, $ N_S^{\mathbb{R}} \geq N_\lambda$, and depends on the values of the resonance parameters $\big\{E_\lambda, \gamma_{\lambda c}, B_c, E_{T_c}, a_c \big\}$ -- though in a way that is invariant under change of boundary-condition $B_c$,
    \item each neutral particle, odd angular momentum $\ell_c \equiv 1 \; (\mathrm{mod}\; 2)$, channel adds at least one real analytic Brune pole below threshold,
\end{itemize}
so that the number $N_S$ of complex and $N_S^{\mathbb{R}}$ of real analytic Brune poles is greater than the number $N_\lambda$ of levels: $ N_S \geq N_S^{\mathbb{R}} \geq N_\lambda$.
\end{theorem}

\begin{proof}

The proof follows the one of theorem \ref{theo::shadow_Brune_poles}. 
However, when considering the left-hand side of (\ref{eq:Brune eigenproblem}), the shift function is now defined from analytic continuation definition (\ref{eq:: Def S and P analytic continuation from L}), which in general entails $S_c(E)$ is a complex number. 
This entails the left-hand side of (\ref{eq:Brune eigenproblem}) is now a complex symmetric matrix. 
In general, a complex symmetric matrix is not diagonalizable, has no special properties on its spectrum, and we refer to reference literature on its Jordan canonical form and other properties 
\cite{Craven_complex_symmetric_1969, Nondefective_complex_symmetric_matrices_1985, complex_symmetric_matrix_SVD_1988, Scott_complex_symmetric_1993, fast_diag_of_complex_symmetric_matrices_for_quantum_applications_1997, Complex_symmetric_operators_2005, Complex_symmetric_operators_II_2007}.
Nonetheless, we know the left-hand side of (\ref{eq:Brune eigenproblem}) will be real-symmetric, thus diagonalizable, for real energies above threshold, which hints (but does not prove) it is probably a good assumption to assume the complex symmetric matrix to be non-defective in general. 
Regardless of the eigenvectors, we can search for the Brune poles $\left\{\widetilde{E_i} \right\} $ by solving problem (\ref{eq:R_S by Brune det search}) directly (c.f. discussion around equation (51) in \cite{Brune_2002}).
Here, the analytic properties of definition (\ref{eq:: Def S and P analytic continuation from L}), established in lemma \ref{lem:: analytic S_c and P_c lemma}, entail the determinant in (\ref{eq:R_S by Brune det search}) is a meromorphic operator of $\rho^2$, which unfolds mapping (\ref{eq:rho_c(E) mapping}) so that all the solutions of (\ref{eq:R_S by Brune det search}) live on one single sheet. 

In the case of $N_c$ massive neutral channels, the shift factor $S_c(\rho)$ is a rational fraction in $\rho^2$ with a degree of $\ell_c$ (in $E$ space) in the denominator, where $\ell_c$ is the angular momentum of the channel (c.f. table \ref{tab::S_and_P_expressions_neutral} and lemma \ref{lem:: analytic S_c and P_c lemma} with table \ref{tab::roots of the outgoing wave functions}).
The search for the poles of the $\boldsymbol{R}_S$ operator (\ref{eq:R_S by Brune det search}) will then yield $N_S$ complex Brune poles $\left\{\widetilde{E_i}\right\} \in \mathbb{C}$ with $N_S = N_\lambda + \sum_{c=1}^{N_c}\ell_c$, as stated in (\ref{eq:N_S Brune poles}).
The intuition behind this number $N_S$ is that both the R-matrix (\ref{eq:R expression}) and the diagonal matrix of shift functions, $\boldsymbol{S}(E) \coloneqq \mathrm{\boldsymbol{diag}}\left( S_c(E)\right)$, will each contribute their number of poles, $N_\lambda$ and $\sum_c \ell_c$ respectively, adding them up to yield $N_S = N_\lambda + \sum_{c=1}^{N_c}\ell_c$ solutions (\ref{eq:N_S Brune poles}) to the determinant problem (\ref{eq:R_S by Brune det search}).
We achieved a formal proof of result (\ref{eq:N_S Brune poles}), though it is somewhat technical. 
It rests on the diagonal divisibility and capped multiplicities lemma \ref{lem::diagonal divisibility and capped multiplicities}, which we apply to the developed rational fraction $\mathrm{det}\left( \boldsymbol{R}_S^{-1}(E)\right)$ in (\ref{eq:R_S by Brune det search}), or directly to (\ref{eq:Brune eigenproblem}), depending on whether $N_\lambda \geq N_c$ or $N_c \geq N_\lambda$. 
In the (most common) case of $N_\lambda \geq N_c$, we develop $\mathrm{det}\left( \boldsymbol{R}_S^{-1}\right)(E) = \mathrm{det}\left( \boldsymbol{R}^{-1} - \boldsymbol{S^0}\right)(E)$ by n-linearity:
  $ \mathrm{det}\left( \boldsymbol{R}^{-1} - \boldsymbol{S^0}\right)= \mathrm{det}\left( \boldsymbol{R}^{-1}\right) \mathrm{det}\left( \Id{} - \boldsymbol{R}\boldsymbol{S^0}\right)$ with $\mathrm{det}\left( \Id{} - \boldsymbol{R}\boldsymbol{S^0}\right) = 1 - \mathrm{Tr}\left( \boldsymbol{R}\boldsymbol{S^0}\right) + \hdots + \mathrm{Tr}\left( \boldsymbol{\mathrm{Adj}}\left(-\boldsymbol{R}\boldsymbol{S^0}\right)\right) + \mathrm{det}\left(-\boldsymbol{R}\boldsymbol{S^0}\right) $, so that: 
  $\mathrm{det}\left( \boldsymbol{R}_S^{-1}\right) = \mathrm{det}\left( \boldsymbol{R}^{-1}\right)  - \mathrm{Tr}\left( \boldsymbol{\mathrm{Adj}}\left(\boldsymbol{R}^{-1}\right)\boldsymbol{S^0}\right) + \hdots - \mathrm{Tr}\left(\boldsymbol{R}^{-1}\boldsymbol{\mathrm{Adj}}\left(\boldsymbol{S^0}\right)\right) +  (-1)^{N_c} \mathrm{det}\left(\boldsymbol{S^0}\right)  $.
  In the latter expression, $\boldsymbol{R}^{-1}(E) = \boldsymbol{\gamma}^+ \left(\boldsymbol{e} - E \Id{} \right) {\boldsymbol{\gamma}^\mathsf{T}}^+$ has no poles, so its determinant is a polynomial $\mathrm{det}\left( \boldsymbol{R}^{-1}\right)(E) \in \mathbb{C}[X]$. 
 The rational fraction with greatest degree in the denominator is $\mathrm{det}\left(\boldsymbol{S^0}\right) (E) \in \mathbb{C}(X)$.
 For neutral particles $S_c^0(E) = \frac{s^0_c(E)}{d_c(E)} $, where the denominator is of degree $\ell_c = \mathrm{deg}\left(d_c(E) \right)$ in $E$ space (c.f. table \ref{tab::S_and_P_expressions_neutral}), so that to rationalize the rational fraction $\mathrm{det}\left( \boldsymbol{R}_S^{-1}\right)(E) \in \mathbb{C}(X)$, we must multiply it by the denominator of $\mathrm{det}\left(\boldsymbol{S^0}\right) (E)$, which is $\prod_{c=1}^{N_c} d_c(E)$, a polynomial of degree $\sum_c \ell_c$. 
 That is $\left(\prod_{c=1}^{N_c} d_c(E) \right) \times \mathrm{det}\left( \boldsymbol{R}_S^{-1}\right)(E) = \left(\prod_{c=1}^{N_c} d_c(E) \right) \times \mathrm{det}\left( \boldsymbol{R}^{-1}\right)(E)  + \hdots + (-1)^{N_c}\prod_{c=1}^{N_c} s^0_c(E)  \in \mathbb{C}[X]$.
 The dominant degree polynomial in this expression is $\left(\prod_{c=1}^{N_c} d_c(E) \right) \times \mathrm{det}\left( \boldsymbol{R}^{-1}\right)(E) $. In this expression, the total degree of the polynomial is the sum of the degrees of the product terms.
 We readily have $\mathrm{deg}\left(\prod_{c=1}^{N_c} d_c(E) \right) = \sum_c \ell_c$. 
For the degree of the determinant term $\mathrm{det}\left( \boldsymbol{R}^{-1}\right)(E)$, the application of diagonal divisibility and capped multiplicities lemma \ref{lem::diagonal divisibility and capped multiplicities} stipulates that if $E_{\lambda_1} = E_{\lambda_2} = \hdots = E_{\lambda_{m_\lambda}} $, this multiplicity ${m_\lambda}$ of the resonance energy value $E_\lambda $ will be capped by $N_c$. In practice, this does not happen because the Wigner-Eisenbud resonance parameters $E_\lambda$ are defined as different from each other $E_\lambda \neq E_{\mu\neq \lambda}$.
 This is no longer true in the case $N_c \geq N_\lambda$, where developing the determinant of (\ref{eq:Brune eigenproblem}) directly will similarly yield by n-linearity, and denoting $\boldsymbol{\Delta} \coloneqq \boldsymbol{e} - E \Id{}$ for clarity of scripture: 
 $\mathrm{det}\left( \boldsymbol{\Delta} - \boldsymbol{\gamma} \boldsymbol{S^0} \boldsymbol{ \gamma}^\mathsf{T}\right) = \mathrm{det}\left( \boldsymbol{\Delta}\right)  - \mathrm{Tr}\left( \boldsymbol{\mathrm{Adj}}\left(\boldsymbol{\Delta}\right)\boldsymbol{\gamma} \boldsymbol{S^0} \boldsymbol{ \gamma}^\mathsf{T}\right) + \hdots - \mathrm{Tr}\left(\boldsymbol{\Delta} \; \boldsymbol{\mathrm{Adj}}\left(\boldsymbol{\gamma} \boldsymbol{S^0} \boldsymbol{ \gamma}^\mathsf{T}\right)\right) + (-1)^{N_\lambda} \mathrm{det}\left(\boldsymbol{\gamma} \boldsymbol{S^0} \boldsymbol{ \gamma}^\mathsf{T}\right)  $.
Again, in the latter expression the rational fraction with the highest-degree denominator is $\mathrm{det}\left(\boldsymbol{\gamma} \boldsymbol{S^0} \boldsymbol{ \gamma}^\mathsf{T}\right)(E) \in \mathbb{C}(X)$. 
Applying the diagonal divisibility and capped multiplicities lemma \ref{lem::diagonal divisibility and capped multiplicities} to it commands that if there are various channels with the same $S_c(E)$, for instance with the same $\ell_c$ and ${\rho_0}_c$, their multiplicity of occurrence is capped by $N_\lambda$ when rationalizing the fraction  $\mathrm{det}\left(\boldsymbol{\gamma} \boldsymbol{S^0} \boldsymbol{ \gamma}^\mathsf{T}\right)(E) \in \mathbb{C}(X)$, so that $Q(E) \times \mathrm{det}\left(\boldsymbol{\gamma} \boldsymbol{S^0} \boldsymbol{ \gamma}^\mathsf{T}\right)(E) \in \mathbb{C}[X]$ is a polynomial, with $Q(E) \coloneqq \left(\prod_{\begin{array}{c}
     c = 1  \\
     d_c \neq d_{c \neq c'}
\end{array}}^{N_c} d_c(E)  \right) \times  \left( \prod_{\begin{array}{c}
     c' = 1  \\
     d_c = d_{c \neq c'}
\end{array}}^{\mathrm{mim}\left\{N_c,N_\lambda\right\}} d_c(E) \right)$.
In the developed expression of the polynomial $ Q(E) \times \mathrm{det}\left( \boldsymbol{\Delta} - \boldsymbol{\gamma} \boldsymbol{S^0} \boldsymbol{ \gamma}^\mathsf{T}\right)$, the dominant degree term is now: $Q(E) \times \mathrm{det}\left( \boldsymbol{\Delta}\right) $, the degree of which is the sum of the degree of each term. The degree of $ \mathrm{det}\left( \boldsymbol{\Delta}\right)$ is $N_\lambda$, whereas the degree of $Q(E)$ is $\mathrm{deg}\left(Q(E)\right) = \sum_{c=1 | \ell_c \neq \ell_{c'}}^{N_c} \ell_c + \sum_{c=1 | \ell_c = \ell_{c'}}^{\mathrm{min}\left\{N_\lambda, N_c\right\}} \ell_c $.
Hence, we find back the expression (\ref{eq:N_S Brune poles}) to be proved: $N_S = N_\lambda + \sum_{c=1}^{N_c}\ell_c$, but with the additional subtlety that the multiplicities (repeating occurrences) are capped, both for $\sum_{\begin{array}{c}
     \tiny{E_\lambda \mathrm{\, multiplicity }}  \\
    \tiny{\mathrm{ capped \, at} \, N_c}
 \end{array}} \mathrm{deg}\left(E_\lambda - \rho^2(E)\right)$ and for $\sum_{\begin{array}{c}
   \tiny{S_c \mathrm{\, multiplicity }}  \\
    \tiny{\mathrm{ capped \, at} \, N_\lambda}
 \end{array}} \mathrm{deg}\left(d_c(\rho(E))\right)$, so that the final, exact number of complex eigenvalues to Brune's generalized eigenproblem (\ref{eq:Brune eigenproblem}) in the neutral channels case is: 
\begin{equation}
\begin{IEEEeqnarraybox}[][c]{rcl}
 N_S =  N_\lambda + \sum_{\begin{array}{c}
   \tiny{S_c \mathrm{\, multiplicity }}  \\
    \tiny{\mathrm{ capped \, at} \, N_\lambda}
 \end{array}} \ell_c 
\IEEEstrut\end{IEEEeqnarraybox}
\label{eq: capped multiplicities N_S}
\end{equation}
This means that if many channels, say $m_c$, have the same shift function $S_c = S_{c'}$, the resulting $\ell_c = \ell_{c'} $ will only be added $\mathrm{min}\left\{m_c, N_\lambda\right\}$ times in the sum (\ref{eq: capped multiplicities N_S}). \\
A final technical note to state that this number $N_S$ of poles (\ref{eq: capped multiplicities N_S}) is true in $E$ space, as we have showed in lemma \ref{lem:: analytic S_c and P_c lemma} that definition (\ref{eq:: Def S and P analytic continuation from L}) unfolds the Riemann sheet of (\ref{eq:rho_c(E) mapping}). 
If we were performing this in $\rho$ space, we would thus simply multiply the degrees by 2.
This is not true if we were searching for the poles of the Kapur-Peierls operator $\boldsymbol{R}_L$, as the mapping of $\rho(E)$ is not one-to-one anymore. 
From table \ref{tab::L_values_neutral}, we would be able to perform the same analysis that yielded (\ref{eq: capped multiplicities N_S}), but it would have to be in $\rho$ space, as we did to establish (\ref{eq::NL number of poles}).

In the charged particles case, $S_c(E)$ has an infinity of poles (c.f. our discussion in section \ref{subsec::Elemental solutions under pole expansion}). Extending our proof of (\ref{eq: capped multiplicities N_S}) from the neutral particles to the charged particles ones would thus yield a countable infinity of complex Brune poles. 

The key question is: how many of the $N_S$ complex Brune poles are real?
To address it, we come back to the three-step Brune argument and look for real eigenvalues from the left-hand-side of (\ref{eq:Brune eigenproblem}) that will cross the right-hand side identity line $E=E$ for real values. 
Here again, Brune's three-step argument will guarantee at least $N_\lambda$ real solutions.
There are in general more solutions however, and as for the shadow Brune poles of theorem \ref{theo::shadow_Brune_poles}, the number of real analytic Brune poles, solutions to (\ref{eq:Brune eigenproblem}), will depend on the R-matrix parameters $\big\{E_\lambda, \gamma_{\lambda c}, B_c, E_{T_c}, a_c \big\}$, in a way that is invariant under change of boundary-condition $B_c$ (plugg-in invariance (\ref{eq:: R_B invariance for B'}) into (\ref{eq:R_S by Brune det search})).
We illustrate various such cases in figure \ref{fig:Brune_argument_even_angular_momenta}. 
However, each neutral particle channel with odd angular momentum $\ell_c \equiv 1 \; (\mathrm{mod} \; 2 )$ will add at least one real sub-threshold solution to the $N_\lambda$ ones, due to the real sub-threshold pole of $S_c(E)$ unveiled in lemma \ref{lem:: analytic S_c and P_c lemma}.
This behavior is depicted in figure \ref{fig:Brune_argument_odd_angular_momenta}.
\end{proof}

\begin{lem}\label{lem::diagonal divisibility and capped multiplicities}
\textsc{Diagonal divisibility and capped multiplicities}.\\
Let $\boldsymbol{M} \in \mathbb{C}^{m\times n}$ be a complex matrix and $\boldsymbol{D}(z) \in \boldsymbol{\mathrm{Diag}}_{n}\left(\mathbb{C}\left(X\right)\right)$ be a diagonal matrix of complex rational functions with simple poles, that is $D_{ij}(z) = \delta_{ij} \frac{R_i(z) \in \mathbb{C}\left[ X \right] }{P_i(z) \in \mathbb{C}\left[ X \right]}$, with $\mathbb{C}\left[ X \right]$ designating the set of polynomials and $\mathbb{C}\left(X\right)$ the set of rational expressions, and we assume $P_i(z)$ has simple roots.\\
Let $Q(z) \in \mathbb{C}\left[ X \right]$ be the denominator of $\mathrm{det}\left( \boldsymbol{D} \right)(z)$, but with all multiplicities capped by $m$, i.e.
\begin{equation}
Q(z) \coloneqq \prod_{\begin{array}{c}
     j = 1  \\
     P_j \neq P_{i \neq j}
\end{array}}^n P_j(z)\prod_{\begin{array}{c}
     i = 1  \\
     P_i = P_{i \neq j}
\end{array}}^{\mathrm{mim}\left\{n,m\right\}} P_i(z)
\end{equation}
then $Q(z)$ is the denominator of $\mathrm{det}\left( \boldsymbol{M} \boldsymbol{D}(z) \boldsymbol{M}^\mathsf{T} \right) $, so that:
\begin{equation}
Q(z)\cdot \mathrm{det}\left( \boldsymbol{M} \boldsymbol{D}(z) \boldsymbol{M}^\mathsf{T} \right)  \in \mathbb{C}\left[ X \right]
\end{equation}
\end{lem}

\begin{proof}
Leibniz's determinant formula yields:
\begin{equation*}
\begin{IEEEeqnarraybox}[][c]{rcl}
\mathrm{det}\left( \boldsymbol{M} \boldsymbol{D}(z) \boldsymbol{M}^\mathsf{T} \right)  & = & \sum_{\sigma \in S_m} \epsilon(\sigma) \prod_{i=1}^{m}\sum_{j=1}^{n}M_{ij}M_{\sigma{i} j} \frac{R_j(z)}{P_j(z)}  \\
\IEEEstrut\end{IEEEeqnarraybox}
\label{eq:: Leibniz formula}
\end{equation*}
Let us now develop the product using the formula:
\begin{equation*}
\begin{IEEEeqnarraybox}[][c]{rcl}
\prod_{i=1}^{m}\sum_{j=1}^{n}x_{i,j} = \sum_{j_1, \hdots , j_m \in \llbracket 1, n \rrbracket ^m } \prod_{i=1}^{m}x_{i,j_i}
\IEEEstrut \end{IEEEeqnarraybox}
\label{eq:: Yoann's touch formula}
\end{equation*}
which leads to:
\begin{equation}
\begin{IEEEeqnarraybox}[][c]{rcl}
\mathrm{det}\left( \boldsymbol{M} \boldsymbol{D} \boldsymbol{M}^\mathsf{T} \right)  & = & \sum_{\sigma \in S_m} \epsilon(\sigma)\sum_{\tiny{\begin{array}{c}
     j_1, \hdots , j_m  \\
      \in \llbracket 1, n \rrbracket ^m 
\end{array}}}  \prod_{i=1}^{m} M_{ij_i}M_{\sigma{i} j_i} \frac{R_{j_i}(z)}{P_{j_i}(z)}  \\
\IEEEstrut\end{IEEEeqnarraybox}
\label{eq:: Developed determinant for divisibility Lemma}
\end{equation}
We here have a sum of products of $m$ terms; thus, the $ \frac{R_{j}(z)}{P_{j}(z)}$ never appear more than $m$ times in each product -- nor more than their multiplicity in $\mathrm{det}\left( \boldsymbol{D} \right)(z)$.
It thus suffices to account for each $P_j(z)$ a number of times that is the maximum between its multiplicity and $m$ in order to rationalize the $\mathrm{det}\left( \boldsymbol{M} \boldsymbol{D}(z) \boldsymbol{M}^\mathsf{T} \right) \in \mathbb{C}(X)$ fraction.
\end{proof}

\begin{figure}[ht!!] 
  \centering
  \includegraphics[width=0.50\textwidth]{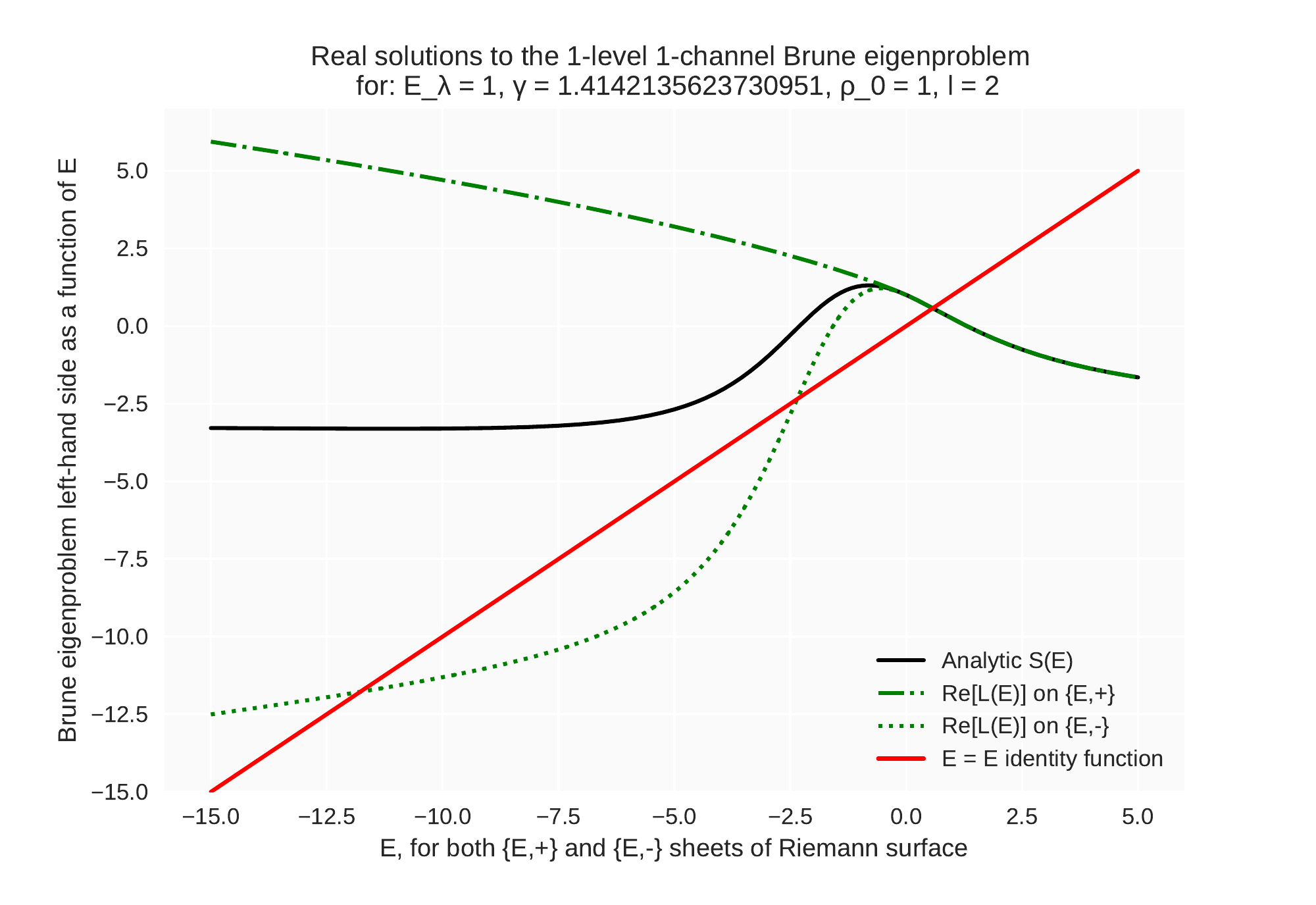}
  \includegraphics[width=0.50\textwidth]{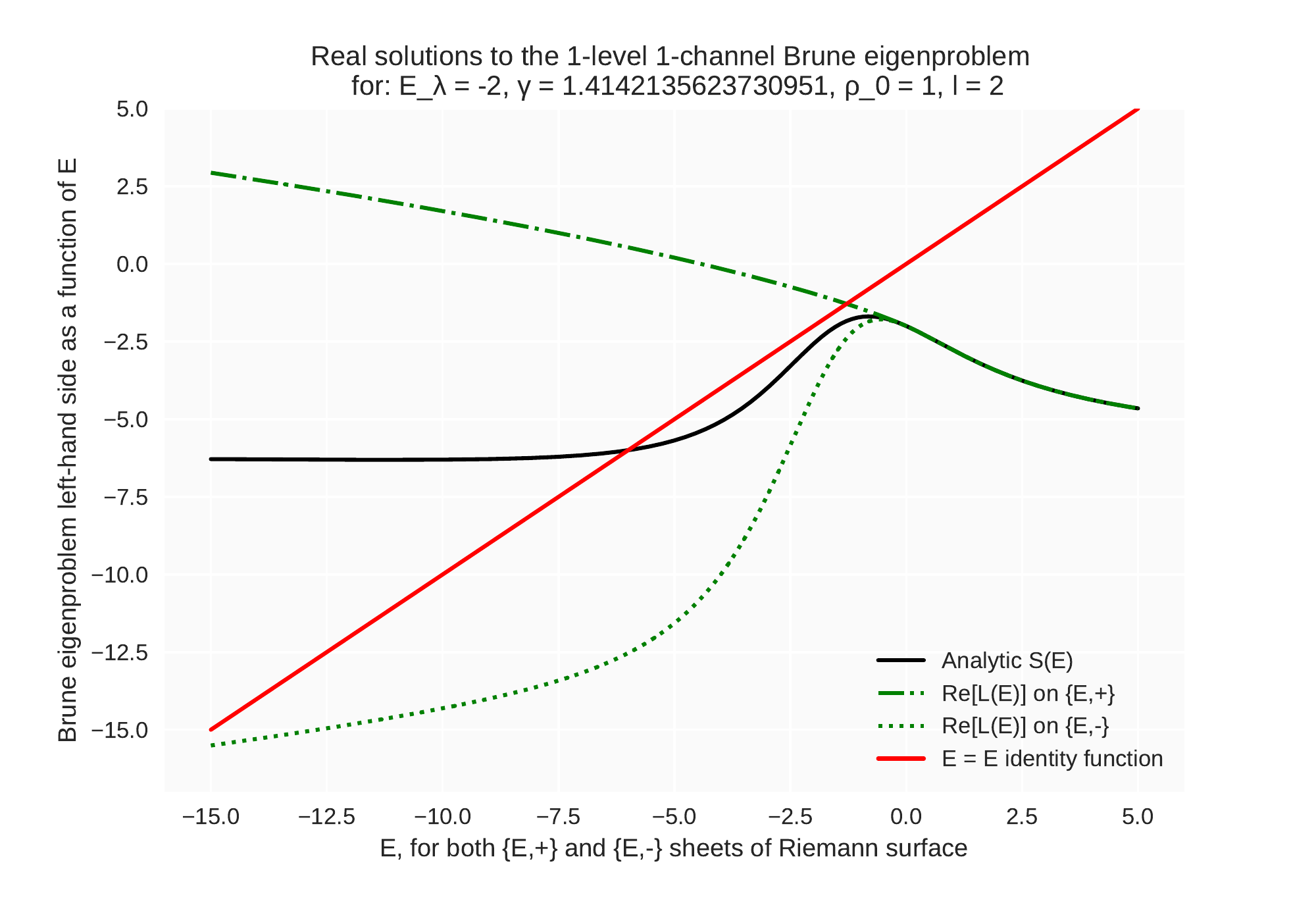}
  \caption{\small{Elemental Brune eigenproblem (\ref{eq:: elemental Brune problem}): comparison of solutions from definitions (\ref{eq:: Def S = Re[L], P = Im[L]}) versus (\ref{eq:: Def S and P analytic continuation from L}), for angular momentum $\ell_c = 2$, neutral particles, using $B_c = - \ell_c$ convention and zero threshold $E_{T_c} = 0$.  Since there are no real sub-threshold poles, both can yield one, two, or three solutions (crossing the E=E diagonal), depending on the values of the resonance parameters. If at threshold energy $E_{T_c}$ the left hand side of (\ref{eq:: elemental Brune problem}) is above the E=E diagonal, then the above-threshold solutions from both definitions coincide. In any case, the sub-threshold solutions differ. Behavior is analogous for all even angular momenta $\ell_c \equiv 0 (\mathrm{mod} 2)$.}}
  \label{fig:Brune_argument_even_angular_momenta}
\end{figure}

Importantly, since both shift function $S_c(E)$ definitions (\ref{eq:: Def S = Re[L], P = Im[L]}) and (\ref{eq:: Def S analytical}) coincide above threshold, the solutions to (\ref{eq:Brune eigenproblem}) will be the same above thresholds. 
The discrepancy in the values of the Brune parameters, solutions to (\ref{eq:Brune eigenproblem}), will only differ when certain channels have to be considered below threshold: $S_c(E)$ with $E<E_{T_c}$.

To illustrate these differences, let us consider the simple example of a one-level, one-channel neutral particle interaction, with a zero-threshold $E_{T_c} = 0$, and set about solving the Brune generalized eigenproblem (\ref{eq:Brune eigenproblem}), which here takes the simple scalar form:
\begin{equation}
    E_\lambda - \gamma_{\lambda,c}\big( S_c(E) - B_c \big) \gamma_{\lambda,c} = E
    \label{eq:: elemental Brune problem}
\end{equation}
In figures \ref{fig:Brune_argument_odd_angular_momenta} and \ref{fig:Brune_argument_even_angular_momenta}, we plotted the left and right hand side of this elemental Brune eigenproblem (\ref{eq:: elemental Brune problem}), for both definitions (\ref{eq:: Def S = Re[L], P = Im[L]}) and (\ref{eq:: Def S and P analytic continuation from L}) of the shift function $S_c(E)$, for various values of resonance parameters $\left\{E_\lambda, \gamma_{\lambda,c}\right\}$ and the convention $B_c = -\ell_c$, for different angular momenta $\ell_c$.

In the case of $\ell_c = 1$, depicted in figure \ref{fig:Brune_argument_odd_angular_momenta}, one can observe that the real sub-threshold pole engendered by odd angular momenta (c.f. section \ref{subsubsec::Ambiguity in shift and penetration}) introduces a sub-threshold Brune parameter, where the left-hand side of (\ref{eq:: elemental Brune problem}) crosses the $E=E$ identity line.
In the case of the Lane \& Thomas legacy definition (\ref{eq:: Def S = Re[L], P = Im[L]}), this sub-threshold shadow Brune pole is on the $\big\{ E, - \big\}$ sheet of mapping (\ref{eq:rho_c massive}),  whereas for analytic continuation definition (\ref{eq:: Def S and P analytic continuation from L}) it is on the same, unique sheet. The same behavior will be observable for all odd angular momenta $\ell_c \equiv 1 \; (\mathrm{mod} \; 2)$.

In the case $\ell_c=2$, depicted in figure \ref{fig:Brune_argument_even_angular_momenta}, the non-purely-imaginary poles $\left\{ \omega_n , \omega_n^* \right\} \not\in \mathrm{i}\mathbb{R}$ (c.f. lemma \ref{lem:: analytic S_c and P_c lemma} and table \ref{tab::roots of the outgoing wave functions}) will impact the shift function $S_c(\rho_c)$ in ways that may or may not produce additional real solutions $\left\{\widetilde{E_i}\right\} \in \mathbb{R}$ to the generalized eigenproblem (\ref{eq:Brune eigenproblem}).
This behavior is reported in figure \ref{fig:Brune_argument_even_angular_momenta}, where one can observe that, depending on the R-matrix parameter values $\big\{ E_\lambda, \gamma_{\lambda,c}, B_c \big\}$, there are either one, two (tangential for the analytic continuation definition), or three solutions to the Brune generalized eigenproblem (\ref{eq:: elemental Brune problem}).
For instance, one can see that definition (\ref{eq:: Def S = Re[L], P = Im[L]}) can yield situations with two sub-threshold branch Brune poles -- one on the $\big\{E,+\big\}$ branch and one shadow pole (i.e. on the $\big\{E,-\big\}$ branch) -- or with two sub-threshold shadow Brune poles -- both sub-threshold on the $\big\{E,-\big\}$ branch -- or situations where only one, above-threshold solution is produced. 
On the other hand, analytic continuation definition (\ref{eq:: Def S and P analytic continuation from L}) can also yield one, two (tangentially) or three solutions, depending on the sub-threshold behavior and the resonant parameters eigenvalues $\big\{ E_\lambda, \gamma_{\lambda,c}, B_c \big\}$.
The number of real solutions $\left\{\widetilde{E_i}\right\} \in \mathbb{R}$ to the Brune generalized eigenproblem (\ref{eq:Brune eigenproblem}) will thus depend on the R-matrix parameters, and is in general comprised between $N_\lambda$ and $N_S$.

To verify the number of complex analytic Brune poles (\ref{eq:N_S Brune poles}), a trivial example is considering (\ref{eq:: elemental Brune problem}) in the $\ell_c = 1$ case, where the analytic shift function takes the wavenumber dependence, $S(\rho) = - \frac{1}{1+\rho^2}$, and thus the poles of the $\boldsymbol{R}_S$ operator are nothing but the solutions to $\frac{E_\lambda - E}{\gamma_{\lambda,c}^2} + B + \frac{1}{1+\rho_0^2(E-E_{T_c})} = 0$.
The fundamental theorem of algebra then guarantees this problem has $N_S = 2$ complex solutions, not $N_\lambda = 1$.
The surprising part is that both are real poles: one above and one below threshold, which again stems from the fact the number of roots $\left\{ \omega_n \right\}$ is odd and that their symmetries thus require one pole to be exaclty imaginary (in wavenumber space), as explained in section \ref{subsubsec::Ambiguity in shift and penetration}. 
For $\ell_c = 2$, we would have $S_2^0(E) = \frac{3E + 2 E^2}{\frac{9}{\rho_0^2} + 3E + E^2}$, so that the fundamental theorem of algebra commands (\ref{eq:: elemental Brune problem}) will have $N_S = 3$ solutions, verifying the $N_S = N_\lambda + \sum_{c=1}^{N_c}\ell_c$ complex poles we establish in (\ref{eq:N_S Brune poles}).
In the general charged-particles case, the shift factor $S_c(\rho)$ is no longer a rational fraction in $\rho^2$ but is a meromorphic operator in $\rho^2$ with an infinity of poles (c.f. lemma \ref{lem:: analytic S_c and P_c lemma}). This means that, in general, there exist $N_\lambda \leq N_S \leq \infty$ complex poles of the $\boldsymbol{R}_S$ operator, and that at least $N_\lambda$ of them are real.

When the left-hand side of (\ref{eq:: elemental Brune problem}) crosses the $E=E$ identity line above threshold, the branch Brune poles coincide with the analytic Brune poles, as can be observed in figures \ref{fig:Brune_argument_odd_angular_momenta} and \ref{fig:Brune_argument_even_angular_momenta}. 
Since the shift function $S_c(E)$ is continuous and monotonically increasing above threshold, the question is whether the eigenvalues of the left-hand side of (\ref{eq:Brune eigenproblem}) are above the $E=E$ line at the threshold value: $E = E_{T_c}$.
If yes, then it would mean that past the last threshold there will be exactly $N_\lambda$ solutions to (\ref{eq:Brune eigenproblem}).
However, nothing guarantees \textit{a priori} that all the eigenvalues of the left hand side of (\ref{eq:Brune eigenproblem}) are above the $E=E$ at the last threshold.
From solving the elemental Brune problem (\ref{eq:: elemental Brune problem}), we observed that it seems to require negative resonance levels $E_\lambda < 0$ to induce the left-hand side of (\ref{eq:Brune eigenproblem}) to be below the $E=E$ line at the threshold value, as illustrated in figures \ref{fig:Brune_argument_odd_angular_momenta} and \ref{fig:Brune_argument_even_angular_momenta}.
When this happens, the Brune poles will be sub-threshold, and thus depend on the (\ref{eq:: Def S = Re[L], P = Im[L]}) or (\ref{eq:: Def S and P analytic continuation from L}) definition for the shift function $S_c(E)$.
However, the fact that different channels will have different threshold levels $E_{T_c} \neq E_{T_{c'}}$, and that nothing stops R-matrix parameters from displaying negative resonance levels $E_\lambda < 0$, mean no definitive conclusion can be reached as to the number of real Brune parameters.

\subsubsection{\label{subsubsec::choice of Brune poles}Choice of Brune poles}

Brune defined his alternative Brune parameters in (\ref{eq:Brune parameters}) and (\ref{eq:: Brune invA}) by building the square matrix $\boldsymbol{g}$, and then inverting it to guarantee (\ref{eq:: Brune A = aAa}) (c.f. section \ref{sec:R_S def}).
We just demonstrated in theorems \ref{theo::shadow_Brune_poles} and \ref{theo::analytic_Brune_poles} that there are in general more Brune poles $N_S$ -- either branch Brune poles or analytic Brune poles -- than the number $N_\lambda$ of resonance levels: $N_S \geq N_\lambda$.
Yet the fact that there are more than $N_\lambda$ solutions to (\ref{eq:Brune eigenproblem}) implies the $\boldsymbol{g} \coloneqq \left[\boldsymbol{g_1}, \hdots, \boldsymbol{g_i} , \hdots , \boldsymbol{g_{N_S}} \right]$ matrix, composed of the $N_S$ solutions to Brune's eigenproblem (\ref{eq:Brune eigenproblem}), is in general not square, and could even be infinite if $N_S = \infty$ (Coulomb channels).
This brings two critical questions: 1) do these additional Brune poles impede us from well defining the Brune parameters? 2) can we still uniquely define the Brune poles?

We here demonstrate in theorem \ref{theo::Choice of Brune poles} the striking property that choosing any finite set of at least $N_\lambda$ different solutions from the $N_\lambda \leq N_S \leq \infty$ solutions of Brune's eigenproblem (\ref{eq:Brune eigenproblem}), suffices, under our new extended definition (\ref{eq:: pseudo invA Brune}), to properly describe the R-matrix scattering model.

\begin{theorem}\label{theo::Choice of Brune poles} \textsc{Choice of Brune poles} \\
If we generalize Brune's definition (\ref{eq:: Brune invA}) of the physical level matrix to its pseudo-inverse $\boldsymbol{\widetilde{A}}^{+} $, setting:
\begin{equation}
\boldsymbol{\widetilde{A}}^{+} \ \coloneqq \ \boldsymbol{g}^\mathsf{T}\boldsymbol{A}^{-1}\boldsymbol{g}
\label{eq:: pseudo invA Brune}
\end{equation}
then the choice of any number $N_S$ of Brune poles, solutions to the Brune generalized eigenproblem (\ref{eq:Brune eigenproblem}), will reconstruct the scattering matrix $\boldsymbol{U}(E)$, as long as we choose at least $N_\lambda$ solutions: $N_S \geq N_\lambda$.
\end{theorem}

\begin{proof}
The proof rests on the pseudo-inverse property for independent columns and rows, and applies it to the $\boldsymbol{g} \coloneqq \left[\boldsymbol{g_1}, \hdots, \boldsymbol{g_i} , \hdots , \boldsymbol{g_{N_S}} \right]$ matrix, constructed by choosing $N_S$ solutions of the generalized eigenproblem (\ref{eq:Brune eigenproblem}).
If $N_S \geq N_\lambda$, then $\boldsymbol{g}$ has independent rows so that its pseudo-inverse will yield: $ \boldsymbol{\widetilde{A}}  = \boldsymbol{g}^{+}\boldsymbol{A}{\boldsymbol{g}^\mathsf{T}}^{+}$.
This property in turn entails (\ref{eq:: Brune A = aAa}) is satisfied, and thus (\ref{eq:R_L unchanged by Brune}) stands, leaving unchanged the Kapur-Peierls operator $\boldsymbol{R}_L$, and hence fully representing the scattering matrix $\boldsymbol{U}(E)$.
\end{proof}

Critically, $N_\lambda$ real solutions to (\ref{eq:Brune eigenproblem}) can always be found -- as shown in theorems \ref{theo::shadow_Brune_poles} and \ref{theo::analytic_Brune_poles} -- meaning the Brune parametrization is always capable of fully reconstructing the scattering matrix energy behavior with real parameters through generalized pseudo-inverse definition (\ref{eq:: pseudo invA Brune}). It is well defined.

Yet, if any choice of $N_\lambda$ Brune poles will yield the same scattering matrix $\boldsymbol{U}(E)$ through definition (\ref{eq:: pseudo invA Brune}), this choice is \textit{a priori} not unique. 
Can we define some conventions on the choice of Brune parameters to make them unique?
Under the legacy Lane \& Thomas definition (\ref{eq:: Def S = Re[L], P = Im[L]}), this can readily be achieved by neglecting the shadow poles and restraining the search to the principal sheet $\big\{ E, + , \hdots , +\big\}$, for all $N_c$ channels, where we have shown in theorem \ref{theo::shadow_Brune_poles} that one will find exactly $N_\lambda$ poles. 
Under the analytic continuation definition (\ref{eq:: Def S and P analytic continuation from L}), one can still uniquely define the $N_\lambda$ ``first" solutions in the following algorithmic way: one starts the search by diagonalizing, at the last threshold energy (greatest $E_{T_c}$ value), the left-hand side of (\ref{eq:Brune eigenproblem}). If all the eigenvalues are above the $E=E$ line, then increase the energy until the eigenvalues cross the $E=E$ diagonal, and we will have $N_\lambda $ uniquely defined real analytic Brune poles. If at the first threshold some eigenvalues are below the $E=E$ line (as we saw could happen if some resonance energies are negative $E_\lambda < 0$), then we can decrease the energy values until those cross the $E=E$ line for the first time, and stop the search there, thus again uniquely defining $N_\lambda$ analytic Brune poles. 
This foray into the algorithmic procedure for solving (\ref{eq:: Def S and P analytic continuation from L}) gives us the occasion to point to the vast literature on methods to solve non-linear eigenvalue problems, in particular \cite{Handbook_of_linear_algebra}. 

In the end, though we argue that the physically correct definition for the shift function $S_c(E)$ ought to be through analytic continuation (\ref{eq:: Def S and P analytic continuation from L}), both approaches enable to set conventions that will uniquely determine $N_\lambda$ real Brune poles.

\subsection{\label{sec:R_L Siegert and Humblet} Complex invariant $\boldsymbol{R}_{L}$ parameters: Siegert-Humblet expansion in radioactive states}

The previous section describes how Brune transformed the real Wigner-Eisenbud resonance parameters $\left\{E_\lambda , \gamma_{\lambda,c}, B_c \right\}$ into a set of real boundary-condition-independent parameters $\left\{\widetilde{E_i} , \widetilde{\gamma_{i,c}}\right\}$. As we saw, this comes at the cost of having to set a convention to uniquely choose $N_\lambda$ Brune parameters (c.f. theorem \ref{theo::Choice of Brune poles}), but at the gain of producing a set of real parameters $\left\{a_c , \widetilde{E_i} , \widetilde{\gamma_{i,c}} , E_{T_c}\right\}$ that entirely characterizes the scattering matrix $\boldsymbol{U}$ and thus the reaction.

In this section we provide new insights into another way of parametrizing the R-matrix model that leads to complex invariant parameters through $\boldsymbol{R}_{L}$: the Siegert-Humblet expansion into radioactive states (c.f. sections IX.2.c-d-e p.297-298 in \cite{Lane_and_Thomas_1958}).
As we will see, these parameters have the advantage of being unique and invariant to boundary condition $B_c$, as well as easily transformed under a change of channel radius $a_c$ (theorem \ref{theo::r_j,c uncer change of a_c} section \ref{sec:Invariance to chanel radius}), and of locally untangling the energy dependence of the scattering matrix as a simple sum of poles and residues. This comes at the cost of greater parameter complexity: the parameters are complex and strongly intertwined; they live on a sub-manifold of the multi-sheeted Riemann surface that the wavenumber-energy mapping (\ref{eq:rho_c(E) mapping}) introduces; and they only allow for a local characterization of the scattering matrix $\boldsymbol{U}$, not a global one $\forall E \in \mathbb{C}$.

\subsubsection{\label{sec:R_L def}Definition of Siegert \& Humblet $\boldsymbol{R}_{L}$ parametrization}

At the heart of the Siegert-Humblet parametrization stands the Kapur-Peierls operator, $\boldsymbol{R}_{L}$, defined in (\ref{eq:Kapur-Peirels Operator and Channel-Level equivalence}):
\begin{equation*}
\boldsymbol{R}_{L}^{-1} \, \coloneqq \, \boldsymbol{R}^{-1} - \boldsymbol{L^0}  \, = \, \boldsymbol{R}^{-1} + \boldsymbol{B} - \boldsymbol{L} = \boldsymbol{\gamma^\mathsf{T} A \gamma}
\end{equation*}
This definition is analogous to Brune's $\boldsymbol{R}_S$ in (\ref{eq:R_S by Brune}).
By analytically continuing the $\boldsymbol{R}$ (\ref{eq:R expression}) and $\boldsymbol{L}$ (\ref{eq:L expression matrix}) operators to complex energies $E\in\mathbb{C}$, the Kapur-Peierls matrix $\boldsymbol{R}_{L}$ becomes a locally meromorphic operator.
The poles of this meroporphic operator can be assumed to have a Laurent expansion of order one, as we will discuss in section \ref{subsubsec:: Semi-simple poles in R-matrix theory}.
Since the Kapur-Peierls $\boldsymbol{R}_{L}$ operator is complex-symmetric, this entails its residues at any given pole value $\mathcal{E}_j \in \mathbb{C}$ are also complex symmetric. For non-degenerate eigenvalues $\mathcal{E}_j \in \mathbb{C}$, the corresponding residues are rank-one and expressed as $\boldsymbol{r_j}\boldsymbol{r_j}^\mathsf{T}$, while for degenerate eigenvalues $\mathcal{E}_j \in \mathbb{C}$ of multiplicity $M_j$, the corresponding residues are rank-$M_j$ and expressed as $ \sum_{m=1}^{M_j}\boldsymbol{r_{j}^m}{\boldsymbol{r_{j}^m}}^\mathsf{T}$.
On a given domain, the Mittag-Leffler theorem \cite{Mittag-Leffler, Pacific_Journal_Mittag_Leffler_and_spectral_theory_1960} then states that $\boldsymbol{R}_{L}$ locally takes the form, in the vicinity $\mathcal{V}(E)$ of any complex energy $E\in\mathbb{C}$ away from the branch points (threshold energies $E_{T_c}$) of mapping (\ref{eq:rho_c(E) mapping}), of a sum of poles and residues and a holomorphic entire part $\boldsymbol{\mathrm{Hol}}_{\boldsymbol{R}_{L}}(E)$:
\begin{equation}
\boldsymbol{R}_{L}(E) \underset{\mathcal{V}(E)}{=} \sum_{j\geq 1} \frac{\sum_{m=1}^{M_j}\boldsymbol{r_{j}^{m}}{\boldsymbol{{r_{j}^{m}}}}^\mathsf{T}}{E - \mathcal{E}_j} + \boldsymbol{\mathrm{Hol}}_{\boldsymbol{R}_{L}}(E)
\label{eq::RL Mittag Leffler degenerate state}
\end{equation}
or, in the particular (but usual) case where $\mathcal{E}_j$ is a non-degenerate eigenvalue (with multiplicity $M_j=1)$,
\begin{equation}
\boldsymbol{R}_{L}(E) \underset{\mathcal{V}(E)}{=} \sum_{j\geq 1} \frac{\boldsymbol{r_j}\boldsymbol{r_j}^\mathsf{T}}{E - \mathcal{E}_j} + \boldsymbol{\mathrm{Hol}}_{\boldsymbol{R}_{L}}(E)
\label{eq::RL Mittag Leffler}
\end{equation}

This is the Siegert-Humblet expansion into so-called \textit{radioactive states} \cite{Siegert, Breit_radioactive_1940, Radioactive_Mahaux_1969, Radioactive_states_Roumania} --- equivalent to equation (2.16) of section IX.2.c. in \cite{Lane_and_Thomas_1958} where we have modified the notation for greater consistency ($\mathcal{E}_j$ corresponds to $H_\lambda$ of \cite{Lane_and_Thomas_1958} and $\boldsymbol{r_j}$ corresponds to $\boldsymbol{\omega_\lambda}$) since there are more complex poles $\mathcal{E}_j$ than real energy levels $E_\lambda$. The Siegert-Humblet parameters are then the poles $\left\{\mathcal{E}_j\right\}$ and residue widths $\left\{\boldsymbol{r_j}\right\}$ of this complex resonance expansion of the Kapur-Peierls operator $\boldsymbol{R}_{L}$.

The Gohberg-Sigal theory provides a method for calculating these poles and residues by solving the generalized eigenvalue problem \cite{Gohberg_Sigal_1971}:
\begin{equation}
\left.\boldsymbol{R}_{L}^{-1}(E)\right|_{E = \mathcal{E}_j} \boldsymbol{q_j} = \boldsymbol{0}
\label{eq:R_L eigenproblem}
\end{equation}
i.e. solving for the poles $\big\{\mathcal{E}_j\big\}$ of the Kapur-Peierls matrix $\boldsymbol{R}_{L}$ operator and their associated eigenvectors $(\boldsymbol{q_j})$. The poles are complex and usually decomposed as:
\begin{equation}
\mathcal{E}_j \coloneqq E_j - \mathrm{i}\frac{\Gamma_j}{2}
\label{eq:E_j pole def}
\end{equation}
It can be shown (c.f. discussion section IX.2.d pp.297--298 in \cite{Lane_and_Thomas_1958}, or section 9.2 eq. (9.11) in \cite{Theory_of_Nuclear_Reactions_I_resonances_Humblet_and_Rosenfeld_1961}) that fundamental physical properties (conservation of probability, causality and time reversal) ensure that the poles reside either on the positive semi-axis of purely-imaginary $k_c \in \mathrm{i}\mathbb{R}_+$ -- corresponding to bound states for real sub-threshold energies, i.e. $E_j < E_{T_c} $ and $\Gamma_j = 0$ -- or that all the other poles are on the lower-half $k_c$ plane, with $\Gamma_j > 0$, corresponding to ``resonance'' or ``radioactively decaying'' states. All poles enjoy the specular symmetry property: if $k_c \in \mathbb{C}$ is a pole of the Kapur-Peierls operator, then $-k_c^*$ is too.

Let $M_j = \mathrm{dim}\left( \mathrm{Ker} \left( \boldsymbol{R}_{L}^{-1}(\mathcal{E}_j) \right) \right) $ be the dimension of the nullspace of the Kapur-Peierls operator at pole value $\mathcal{E}_j$ -- that is $M_j$ is the geometric multiplicity.
We can thus write $\mathrm{Ker} \left( \boldsymbol{R}_{L}^{-1}(\mathcal{E}_j) \right) = \mathrm{vect}\left( \boldsymbol{q_{j}^{1}}, \hdots, \boldsymbol{q_{j}^{m}},\hdots, \boldsymbol{q_{j}^{M_j}} \right)$.
As we discuss in section \ref{subsubsec:: Semi-simple poles in R-matrix theory}, it is physically reasonable to assume that the geometric and algebraic multiplicities are equal (semi-simplicity condition), which entails a Laurent development of order one for the poles -- i.e. no higher powers of $\frac{1}{E - \mathcal{E}_j}$ in expansion (\ref{eq::RL Mittag Leffler degenerate state}).
Since $\boldsymbol{R}_{L}$ is complex symmetric, if we assume we can find non-quasi-null eigenvectors solutions to (\ref{eq:R_L eigenproblem}) -- that is $\forall \; (j,m) \; , \; \; {\boldsymbol{q_{j}^{m}}}^\mathsf{T} \boldsymbol{q_{j}^{m}} \neq 0 $ so it is non-defective \cite{Craven_complex_symmetric_1969, Nondefective_complex_symmetric_matrices_1985, complex_symmetric_matrix_SVD_1988, Scott_complex_symmetric_1993, fast_diag_of_complex_symmetric_matrices_for_quantum_applications_1997, Complex_symmetric_operators_2005, Complex_symmetric_operators_II_2007} -- then the Gohberg-Sigal theory can be adapted to the case of complex symmetric matrices to normalize the rank-$M_j$ residues of $\boldsymbol{R}_{L}$ matrix as:
\begin{equation}
\sum_{m=1}^{M_j} \boldsymbol{r_{j}^{m}} {\boldsymbol{r_{j}^{m}}}^\mathsf{T}  =  \sum_{m=1}^{M_j} \frac{\boldsymbol{q_{j}^{m}}{\boldsymbol{q_{j}^{m}}}^\mathsf{T}}{{\boldsymbol{q_{j}^{m}}}^\mathsf{T} \left( \left. { \frac{\partial \boldsymbol{R}_{L}^{-1}}{\partial E} }\right|_{E=\mathcal{E}_j} \right) {\boldsymbol{q_{j}^{m}}}}
\label{eq:R_L residues for degenerate case}
\end{equation}

In practice, we are most often presented with non-degenerate states where $M_j = 1$, meaning the kernel is an eigenline $\mathrm{Ker} \left( \boldsymbol{R}_{L}^{-1}(\mathcal{E}_j) \right) = \mathrm{vect}\left( \boldsymbol{q_{j}} \right)$, which entails rank-one residues normalized as:
\begin{equation}
\boldsymbol{r_j} \boldsymbol{r_j}^\mathsf{T}  = \frac{\boldsymbol{q_j}\boldsymbol{q_j}^\mathsf{T}}{\boldsymbol{q_j}^\mathsf{T} \left( \left. { \frac{\partial \boldsymbol{R}_{L}^{-1}}{\partial E} }\right|_{E=\mathcal{E}_j} \right) \boldsymbol{q_j}}
\label{eq:R_L residues }
\end{equation}

The residue widths $\left\{\boldsymbol{r_{j}^{m}}\right\}$, here called \textit{radioactive widths}, can thus directly be expressed as:
\begin{equation}
\boldsymbol{r_{j}^{m}}  = \frac{\boldsymbol{q_{j}^{m}}}{\sqrt{{\boldsymbol{q_{j}^{m}}}^\mathsf{T} \left( \left. { \frac{\partial \boldsymbol{R}_{L}^{-1}}{\partial E} }\right|_{E=\mathcal{E}_j} \right) \boldsymbol{q_{j}^{m}}}}
\label{eq:R_L residues widths}
\end{equation}
where $\left. { \frac{\partial \boldsymbol{R}_{L}^{-1}}{\partial E} }\right|_{E=\mathcal{E}_j} $ can readily be calculated by means of property (\ref{eq::inv M derivatie property}) to yield
\begin{equation}
\left. { \frac{\partial \boldsymbol{R}_{L}^{-1}}{\partial E} }\right|_{E=\mathcal{E}_j} = \frac{\partial \boldsymbol{R}^{-1} }{\partial E} (\mathcal{E}_j) - \frac{\partial \boldsymbol{L} }{\partial E} (\mathcal{E}_j)
\end{equation}
with
\begin{equation}
\frac{\partial \boldsymbol{R}^{-1} }{\partial E} (E) = - \boldsymbol{R}^{-1} \boldsymbol{\gamma}^\mathsf{T} \left(\boldsymbol{e} - E\Id{}\right)^{-2} \boldsymbol{\gamma} \boldsymbol{R}^{-1}
\end{equation}
The \textit{radioactive poles}, $\left\{\mathcal{E}_j\right\}$, and residue widths, $\left\{\boldsymbol{r_{j}^{m}} = \left[ r_{{j,c_1}}^{m}, \hdots, r_{{j,c}}^m , \hdots , r_{{j,c_{N_c}}}^{m} \right]^\mathsf{T} \right\}$, are the Siegert-Humblet parameters. They are complex and locally untangle the energy dependence into the sum of poles and residues (\ref{eq::RL Mittag Leffler degenerate state}).
Additional discussion on these poles and residues can be found in \cite{Lane_and_Thomas_1958}, sections IX.2.c-d-e p.297-298, or in \cite{Siegert, Breit_radioactive_1940, Radioactive_Mahaux_1969, Radioactive_states_Roumania}.
Focusing on invariance, since the Kapur-Peierls matrix $\boldsymbol{R}_{L}$ is invariant to a change in boundary conditions $B_c$ --- c.f. equations (\ref{eq:: gAg invariance for B'}) and (\ref{eq:: R_B invariance for B'}) --- this entails the radioactive poles $\left\{\mathcal{E}_j\right\}$ and widths $\left\{\boldsymbol{r_j}\right\}$ are boundary condition $B_c$ independent.
We will also prove (c.f. theorem \ref{theo:: poles of U are Siegert-Humblet poles}) that the poles $\left\{\mathcal{E}_j\right\}$ are exactly the poles of the scattering matrix $\boldsymbol{U}(E)$, which also makes them invariant to channel radii $\left\{a_c\right\}$. From (\ref{eq:U expression}), the radioactive widths $\left\{\boldsymbol{r_j}\right\}$ are not themselves invariant in change of channel radius $a_c$, but we will also show in theorem \ref{theo::r_j,c uncer change of a_c} section \ref{sec:Invariance to chanel radius} how to transform them under a change of channel radius $a_c$.

\subsubsection{\label{sec: Level matrix approach to Siegert-Humblet expansion}Level matrix $\boldsymbol{A}(E)$ approach to Siegert \& Humblet expansion}

An alternative approach to calculating the Siegert-Humblet parameters $\left\{a_c,\mathcal{E}_j,r_{j,c}^m,E_{T_c}\right\}$ from the Wigner-Eisenbud ones $\left\{a_c, B_c, \gamma_{\lambda,c}, E_\lambda, E_{T_c} \right\}$ is through the level matrix $\boldsymbol{A}$.
This strongly mirrors Brune's generalized eigenvalue problem (\ref{eq:Brune eigenproblem}) in that we search for the poles and eigenvectors of the level matrix operator $\boldsymbol{A}$:
\begin{equation}
\left.\boldsymbol{A}^{-1}(E)\right|_{E=\mathcal{E}_j}\boldsymbol{b_j} = \boldsymbol{0}
\label{eq::invA det roots}
\end{equation}
i.e. solve for the eigenvalues $\left\{\mathcal{E}_j\right\}$ and associated eigenvectors $\left\{\boldsymbol{b_j}\right\}$ that satisfy:
\begin{equation}
\begin{IEEEeqnarraybox}[][c]{rCl}
\left[\boldsymbol{e} - \boldsymbol{\gamma} \left( \boldsymbol{L}(\mathcal{E}_j) - \boldsymbol{B} \right)\boldsymbol{ \gamma}^\mathsf{T} \right]\boldsymbol{b_j} = \mathcal{E}_j\boldsymbol{b_j}
\IEEEstrut\end{IEEEeqnarraybox}
\label{eq: invA eigenproblem for RL poles}
\end{equation}
This problem is analogous to (\ref{eq:Brune eigenproblem}), replacing the shift factor $\boldsymbol{S}$ with the outgoing-wave reduced logarithmic derivative $\boldsymbol{L}$.

Again, the same hypotheses as for $\boldsymbol{R}_{L}$ in section \ref{sec:R_L def} allow us to adapt the Gohberg-Sigal theory to the case of complex symmetric operators to locally yield the following formula for the normalized residues in the Mittag-Leffler expansion of the level matrix:
\begin{equation}
\boldsymbol{A}(E) \underset{\mathcal{V}(E)}{=} \sum_{j\geq 1} \frac{ \sum_{m=1}^{M_j}\boldsymbol{a_{j}^{m}}{\boldsymbol{a_{j}^{m}}}^\mathsf{T}}{E - \mathcal{E}_j} + \boldsymbol{\mathrm{Hol}}_{\boldsymbol{A}}(E)
\label{eq::A Mittag Leffler}
\end{equation}
In the most frequent case of non-degenerate eigenvalues to (\ref{eq::invA det roots}), this yields rank-one residues as:
\begin{equation}
\boldsymbol{A}(E) \underset{\mathcal{V}(E)}{=} \sum_{j\geq 1} \frac{ \boldsymbol{a_j}\boldsymbol{a_j}^\mathsf{T}}{E - \mathcal{E}_j} + \boldsymbol{\mathrm{Hol}}_{\boldsymbol{A}}(E)
\label{eq::A Mittag Leffler non-degenerate}
\end{equation}
Again, under non-quasi-null eigenvectors assumption $ {\boldsymbol{b_{j}^{m}}}^\mathsf{T}\boldsymbol{b_{j}^{m}}\neq 0$, Gohberg-Sigal theory ensures the residues are normalized as:
\begin{equation}
\boldsymbol{a_j^m}{\boldsymbol{a_j^m}}^\mathsf{T}  = \frac{\boldsymbol{b_j^m}{\boldsymbol{b_j^m}}^\mathsf{T}}{{\boldsymbol{b_j^m}}^\mathsf{T} \left( \left. { \frac{\partial \boldsymbol{A}^{-1}}{\partial E} }\right|_{E=\mathcal{E}_j} \right) \boldsymbol{b_j^m}}
\label{eq:A residues }
\end{equation}
which is readily calculable from
\begin{equation}
\frac{\partial \boldsymbol{A}^{-1}}{\partial E} (\mathcal{E}_j) = - \Id{} - \boldsymbol{\gamma} \frac{\partial \boldsymbol{L} }{\partial E} (\mathcal{E}_j) \boldsymbol{\gamma}^\mathsf{T}
\end{equation}
Plugging (\ref{eq::A Mittag Leffler}) into (\ref{eq:Kapur-Peirels Operator and Channel-Level equivalence}), and invoking the unicity of the complex residues, implies the radioactive widths (\ref{eq:R_L residues widths}) can be obtained as
\begin{equation}
\boldsymbol{r_j^m} = \boldsymbol{\gamma}^\mathsf{T}\boldsymbol{a_j^m}
\label{eq::radioactive widths rj from aj}
\end{equation}

This is an interesting and novel way to define the Siegert-Humblet parameters, which is similar to the parameter definition of Brune (\ref{eq:Brune parameters}).
From this perspective, the Brune parameters appear as a special case that leave the Siegert-Humblet level-matrix parameters boundary condition $B_c$ invariant.
Indeed, the defining property of Brune's parameters (\ref{eq:R_L unchanged by Brune}) means we can search for the Siegert-Humblet expansion of the Brune parameters, simply by proceeding as in (\ref{eq::invA det roots}) with Brune's alternative physical level matrix $\boldsymbol{\widetilde{A}}$ from (\ref{eq::Brune physical level matrix}) or (\ref{eq:: pseudo invA Brune}):
\begin{equation}
\left.\boldsymbol{\widetilde{A}}^{-1}(E)\right|_{E=\mathcal{E}_j}\boldsymbol{\widetilde{b_j}} = \boldsymbol{0}
\label{eq::invA det roots on Brune}
\end{equation}
The exact same Gohberg-Sigal procedure can then be applied to the Mittag-Leffler expansion of Brune's $\boldsymbol{\widetilde{A}}$ physical level matrix, in the vicinity $\mathcal{V}(E)$ of $E\in\mathbb{C}$ away from branch points $\left\{E_{T_c}\right\}$,
\begin{equation}
\boldsymbol{\widetilde{A}}(E) \underset{\mathcal{V}(E)}{=} \sum_{j\geq 1} \frac{\sum_{m=1}^{M_j}\boldsymbol{\widetilde{a_j^m}}{\boldsymbol{\widetilde{a_j^m}}}^\mathsf{T}}{E - \mathcal{E}_j} + \boldsymbol{\mathrm{Hol}}_{\boldsymbol{\widetilde{A}}}(E)
\label{eq:Brune A Mittag-Leffler}
\end{equation}
yielding the normalized residue widths:
\begin{equation}
\boldsymbol{\widetilde{a_j^m}} {\boldsymbol{\widetilde{a_j^m}}}^\mathsf{T}  = \frac{\boldsymbol{\widetilde{b_j^m}}{\boldsymbol{\widetilde{b_j^m}}}^\mathsf{T}}{{\boldsymbol{\widetilde{b_j^m}}}^\mathsf{T} \left( \left. { \frac{\partial \boldsymbol{\widetilde{A}}^{-1}}{\partial E} }\right|_{E=\mathcal{E}_j} \right) \boldsymbol{\widetilde{b_j^m}}}
\label{eq:A residues Brune}
\end{equation}
where (\ref{eq::Brune physical level matrix}) can be used to calculate the energy derivative.
Then, when plugging (\ref{eq:A residues Brune}) into (\ref{eq:R_L unchanged by Brune}), we obtain the relation between the Brune and the Siegert-Humblet R-matrix parameters:
\begin{equation}
\boldsymbol{r_j^m} = \boldsymbol{\widetilde{\gamma}}^\mathsf{T}\boldsymbol{\widetilde{a_j^m}}
\label{eq::radioactive widths rj from aj Brune}
\end{equation}

This relation (\ref{eq::radioactive widths rj from aj Brune}) is especially enlightening when compared to (\ref{eq::radioactive widths rj from aj}) from the viewpoint of invariance to boundary condition $B_c$. Indeed, we explained that the Siegert-Humblet parameters $\left\{\mathcal{E}_j, \boldsymbol{r_j^m}\right\}$ are invariant with a change of boundary condition $B_c \to B_c'$.
This is however not true of the level matrix residue widths $\left\{\boldsymbol{a_j^m}\right\}$ from (\ref{eq:A residues }). Thus, we can formally write this invariance by differentiating (\ref{eq::radioactive widths rj from aj}) with respect to $B_c$ and noting that $\frac{\partial \boldsymbol{r_j^m}}{\partial \boldsymbol{B}} = \boldsymbol{0}$, yielding:
\begin{equation}
\boldsymbol{0} = \frac{\partial \boldsymbol{\gamma}^\mathsf{T}}{\partial \boldsymbol{B}} \boldsymbol{a_j^m} + \boldsymbol{\gamma}^\mathsf{T}\frac{\partial \boldsymbol{a_j^m}}{\partial \boldsymbol{B}}
\label{eq:: invariance from Bc radioactive widths rj and aj}
\end{equation}
This new relation links the variation of the Wigner-Eisenbud widths $\gamma_{\lambda,c}$ at level values $E_\lambda$ under a change of boundary conditions $B_{c'}$, to the variation of the level matrix residue widths $a_{j,c}^m$ at pole values $\mathcal{E}_j$ under change of boundary condition $B_{c'}$. Since transformations (\ref{eq: Wigner-Eisenbud parameters transformations under change of Bc}) detail how to perform $\frac{\partial \boldsymbol{\gamma}^\mathsf{T}}{\partial \boldsymbol{B}}$, equation (\ref{eq:: invariance from Bc radioactive widths rj and aj}) could be used to update $\boldsymbol{a_j^m}$ under a change $B_c \to B_{c'}$.

Another telling insight from relation (\ref{eq:: invariance from Bc radioactive widths rj and aj}) is when we apply it to the relation between the Brune parameters and the Siegert-Humblet residue widths (\ref{eq::radioactive widths rj from aj Brune}). There, since the Brune parameters $\boldsymbol{\widetilde{\gamma}}$ are invariant to $B_c$, the same differentiation as in (\ref{eq:: invariance from Bc radioactive widths rj and aj}) now yields zero derivatives,
\begin{equation}
\boldsymbol{0} = \boldsymbol{\widetilde{\gamma}}^\mathsf{T}\frac{\partial \boldsymbol{\widetilde{a_j^m}}}{\partial \boldsymbol{B}}
\label{eq:: invariance from Bc radioactive widths rj and aj Brune}
\end{equation}
Though this is obvious from the fact that Brune's physical level matrix $\boldsymbol{\widetilde{A}}$ is invariant under change of boundary condition, it does present the Brune parametrization as the one which leaves the level residue widths $\left\{\boldsymbol{\widetilde{a_j}}\right\}$ invariant to $B_c$ when transforming to Siegert-Humblet parameters though (\ref{eq::radioactive widths rj from aj Brune}).

Conversely, the Siegert-Humblet Kapur-Peierls matrix resonance expansion (\ref{eq::RL Mittag Leffler}) completes Brune's parametrization in that it generates the boundary condition $B_c$ independent poles $\left\{\mathcal{E}_j\right\}$ and radioactive widths $\left\{\boldsymbol{r_j}\right\}$ that explicitly invert Brune's alternative physical level matrix (\ref{eq::Brune physical level matrix}) to yield (\ref{eq:Brune A Mittag-Leffler}).

In practice we are most often presented with the non-degenerate case of rank-one residues (eigenvalue multiplicity of $M_j = 1$), thus for clarity of reading and without loss of generality, we henceforth drop the superscript ``$m$'' and summation over the multiplicity, unless it is of specific interest.

\subsubsection{\label{sec:R_L parameters} $\boldsymbol{R}_{L}$ complex parameters: local expansion, multi-sheeted Riemann surface, poles \& residues }

We here discuss some subtle points in line with section \ref{subsubsec::Ambiguity in shift and penetration}, concerning the continuation of R-matrix operators to complex energies $E\in \mathbb{C}$, which is required in the procedure to calculate the Siegert-Humblet parameters.

We first start with a numerical note.
Numerically, solving the generalized eigenvalue problems (\ref{eq:R_L eigenproblem}) or (\ref{eq: invA eigenproblem for RL poles}) falls into the well-known class of nonlinear eigenvalue problems, for which algorithms we direct the reader to Heinrich Voss's chapter 115 in the Handbook of Linear Algebra \cite{Handbook_of_linear_algebra}.
We will just state that instead of the Rayleigh-quotient type of methods expressed in \cite{Handbook_of_linear_algebra}, it can sometimes be computationally advantageous to first find the poles $\left\{\mathcal{E}_j\right\}$ by solving the channel determinant problem, $\mathrm{det}\left( \left.\boldsymbol{R}_{L}^{-1}(E)\right|_{E = \mathcal{E}_j} \right) = 0$, analogous to (\ref{eq:R_S by Brune det search}), or the corresponding level determinant one, $\mathrm{det}\left( \left.\boldsymbol{A}^{-1}(E)\right|_{E = \mathcal{E}_j} \right) = 0$, and then solve the associated linear eigenvalue problem. Methods tailored to find all the roots of this problem where introduced in \cite{Pavla_PHYSOR_conversion}, or in equations (200) and (204) of \cite{Frohner_Jeff_2000}.
Notwithstanding, from a numerical standpoint, having the two approaches is beneficial in that solving (\ref{eq:R_L eigenproblem}) will be advantageous over solving (\ref{eq: invA eigenproblem for RL poles}) when the number of levels $N_\lambda$ far exceeds the number of channels $N_c$, and conversely.

Let us now provide some remarks on the thorny question of the multi-sheeted nature of the problem.
When solving problem (\ref{eq:R_L eigenproblem}), or (\ref{eq: invA eigenproblem for RL poles}), to obtain the Sieger-Humblet poles $\left\{\mathcal{E}_j\right\}$ and residues $\left\{\boldsymbol{r_j}\right\}$, or $\left\{\boldsymbol{a_j}\right\}$, it is necessary to compute the $\boldsymbol{L^0}$ matrix function $\boldsymbol{L^0}(E) \coloneqq  \boldsymbol{L^0}(\boldsymbol{\rho}(E)) $ for complex energies $E\in \mathbb{C}$. As discussed in \ref{subsec:Energy dependence and wavenumber mapping}, mapping (\ref{eq:rho_c(E) mapping}) generates a multi-sheeted Riemann surface with $2^{N_c}$ branches (with the threshold values $E_{T_c}$ as branch points), corresponding to the choice for each channel $c$, of the sign of the square root in $\boldsymbol{\rho}(E)$. This means that when searching for the poles, one has to keep track of these choices and specify for each pole $\mathcal{E}_j$ on what sheet it is found. Every pole $\mathcal{E}_j$ must thus come with the full reporting of these $N_c$ signs, i.e.
\begin{equation}
\Big\{\mathcal{E}_j, + , +, -, \hdots, +, - \Big\}
\label{eq:: pole E_j sheet reporting}
\end{equation}
The $\left\{\mathcal{E}_j, + , +, \hdots, +, + \right\}$ sheet is called the \textit{physical sheet}, and we here call the poles on that sheet the \textit{principal poles}. All other sheets are called \textit{unphysical} and the poles laying on these sheets are called \textit{shadow poles}. Often, the principal poles are responsible for the resonant behavior, with shadow poles only contributing to background behavior, but cases have emerged where the shadow poles contribute significantly to the resonance structure, as reported in \cite{Hale_1987}, and G. Hale there introduced a quantity called \textit{strength} of a pole (c.f. eq. (7) in \cite{Hale_1987}, or paragraph after eq. (2.11) XI.2.b, p.306, and section XI.4, p.312 in \cite{Lane_and_Thomas_1958}) to quantify the impact a pole $\mathcal{E}_j$ will have on resonance behavior, by comparing the residue $r_{j,c}$ to the Wigner-Eisenbug widths $\gamma_{\lambda,c}$.\\
As discussed in section \ref{sec:R_S Brune transform}, there is an ambiguity of definitions for the shift and penetration functions, which in turn entail various possible Brune parameters. 
When solving (\ref{eq:R_L eigenproblem}) or (\ref{eq: invA eigenproblem for RL poles}) for the Siegert-Humblet poles and residues $\left\{ \mathcal{E}_j, r_{j,c}\right\}$, there are no such ambiguities on the definition of $\boldsymbol{L}$. From (\ref{eq:L expression}), we extend $\boldsymbol{L}(E)$ to complex energies by simply performing analytic continuation of the outgoing wave functions $O_c(\rho_c)$, as for lemma \ref{lem::Mittag-Leffler of L_c Lemma} in section \ref{subsec:External_region_waves}. This means the Siegert-Humblet parameters are uniquely defined, as long as we specify for each channel $c$ what sheet of the Riemann surface from mapping (\ref{eq:rho_c(E) mapping}) was chosen, as in (\ref{eq:: pole E_j sheet reporting}).

As it was the case for the Brune parameters, which counted more solutions than levels ($N_S \geq N_\lambda$), there are more Siegert-Humblet poles $\left\{\mathcal{E}_j\right\}$ than Wigner-Eisenbud levels $\left\{E_\lambda\right\}$. For massive neutral particles, by proceeding in an analogous way as for (\ref{eq: capped multiplicities N_S}), applying the diagonal divisibility and capped multiplicities lemma \ref{lem::diagonal divisibility and capped multiplicities} to the determinant of the Kapurl-Peierls operator $\boldsymbol{R}_L$ in (\ref{eq:R_L eigenproblem}) -- but this time in $\rho_c$ space (c.f. comment after (\ref{eq: capped multiplicities N_S}) in proof of theorem \ref{theo::Choice of Brune poles}) -- and then looking at the order of the resulting rational fractions in $\rho_c$ and the number of times one must square the polynomials to unfold all $\rho_c = \mp \sqrt{\cdot}$ sheets, we where able to establish that the number $N_L$ of poles in wavenumber $\rho$-space is:
\begin{equation}
N_L = \left(2 N_\lambda  + \sum_{c=1}^{N_c} \ell_c \right)\; \times 2^{(N_{E_{T_c} \neq E_{T_{c'}}} - 1 ) }
\label{eq::NL number of poles}
\end{equation}
where $N_{E_{T_c} \neq E_{T_{c'}}}$ designates the number of different thresholds (including the obvious $E_{T_c} = 0$ zero threshold).
Again, as for (\ref{eq: capped multiplicities N_S}), one should add the precision that in the sum over the channels in (\ref{eq::NL number of poles}), the multiplicity of eventual $L_c(\rho_c)$ repeated over many different channels $L_c(\rho_c) = L_{c'\neq c}(\rho_{c'})$ is capped by $N_\lambda$, which in practice would only occur in the rare cases where only one or two levels occurs for many channels with same angular momenta (and, of course, total angular momenta and parity $J^\pi$). 
One can observe that the number $N_L$ of Siegert-Humblet poles adds-up the number of levels $N_\lambda$ and the number of poles of $\boldsymbol{L}$ (which is $\sum_{c=1}^{N_c} \ell_c$ for neutral massive particles, and is infinite in the Coulomb case, c.f. discussion in section \ref{sec:Invariance to chanel radius}). Moreover, $N_L$ is duplicated with each new sheet of the Riemann surface from mapping (\ref{eq:rho_c(E) mapping}) --- that is associated to a new threshold, hence the $N_{E_{T_c} \neq E_{T_{c'}}}$.
Interestingly, comparing $N_L$ from (\ref{eq::NL number of poles}) with the $N_S$ Brune poles from (\ref{eq:N_S Brune poles}) --- which are in $E$-space and must thus be doubled to obtain the number of $\rho$-space poles --- we note that the analytic continuation of the shift factor $\boldsymbol{S}$ adds a virtual pole for each pole of $\boldsymbol{L}$ when unfolding the sheets of mapping (\ref{eq:rho_c(E) mapping}) by being a function of $\rho_c^2(E)$. This can readily be observed in our trivial one level one p-wave ($\ell=1$) channel, where $S(E) = -\frac{1}{1 + \rho^2(E)}$ introduces two poles at $\rho(E) = \pm \mathrm{i}$, while $L(E) = \frac{- 1 + \mathrm{i}\rho(E) + \rho^2(E)}{1 - \mathrm{i}\rho(E)}$ only counts one pole, at $\rho(E) = \mathrm{i}$.

It is important to grasp the meaning of the Mittag-Leffler expansion (\ref{eq::RL Mittag Leffler}) --- or (\ref{eq::A Mittag Leffler}) and (\ref{eq:Brune A Mittag-Leffler}). These are local expressions in that the branch-point structure of the Riemann sheet does not allow these Mittag-Leffler expansions to hold for all complex energy $E\in \mathbb{C}$. However, away from the branch-points $E_{T_c}$ the form locally stands, and the holomorphic part then has an analytic expansion $\boldsymbol{\mathrm{Hol}}_{\boldsymbol{R}_{L}}(E) \coloneqq \sum_{n\geq0} \boldsymbol{c}_n E^n$, which means in a neighborhood $\mathcal{V}(E)$ of $E\in\mathbb{C}$ away from the thrsholds $\left\{E_{T_c}\right\}$ the following expansion holds:
\begin{equation}
\boldsymbol{R}_{L}(E) \underset{\mathcal{V}(E)}{=} \sum_{j\geq 1} \frac{\boldsymbol{r_j}\boldsymbol{r_j}^\mathsf{T}}{E - \mathcal{E}_j} +\sum_{n\geq0} \boldsymbol{c}_n E^n
\label{eq::RL Mittag Leffler expanded}
\end{equation}
This has two major consequences for the Siegert-Humblet expansion.
First, contrarily to Brune's parameters $\left\{\widetilde{E_i}, \widetilde{\gamma_{i,c}}\right\}$, the Siegert-Humblet set of poles and radioactive widths $\left\{\mathcal{E}_j , r_{j,c} \right\}$ do not suffice to uniquely determine the energy behavior of the scattering matrix $\boldsymbol{U}(E)$: one needs to locally add the expansion coefficients $\left[\boldsymbol{c}_n\right]_{c,c'}$ of the entire part $\boldsymbol{\mathrm{Hol}}_{\boldsymbol{R}_{L}}(E) \coloneqq \sum_{n\geq0} \boldsymbol{c}_n E^n$.
Second, since the set of coefficients $\left\{\boldsymbol{c}_n\right\}$ is \textit{a priori} infinite (and the poles set is too for the Coulomb case), this means that numerically the Siegert-Humblet expansion can only be used to compute local approximations of the scattering matrix, which can nonetheless reach any target accuracy by expanding the number of $\left\{  \mathcal{E}_j\right\}_{j \in \llbracket 1,N_L \rrbracket} $ poles included and the order of the truncation $N_{\mathcal{V}(E)}$ in $\left\{\boldsymbol{c}_n\right\}_{n\in\llbracket1,N_{\mathcal{V}(E)}\rrbracket}$.
In practice, this means that to compute the scattering matrix one needs to provide the Siegert-Humblet parameters $\left\{\mathcal{E}_j, r_{j,c}\right\}$, cut the energy domain of interest into local windows $\mathcal{V}(E)$ away from threshold branch-points $\left\{E_{T_c}\right\}$, and provide a set of local coefficients $\left\{\boldsymbol{c}_n\right\}_{n\in\llbracket1,N_{\mathcal{V}(E)}\rrbracket}$ for each window.

\subsubsection{\label{sec:Linking the R-matrix parametrization to the Humblet-Rosenfeld scattering matrix expansion} Linking the R-matrix parametrization to the Humblet-Rosenfeld scattering matrix expansion}

So far, we have started from the R-matrix Wigner-Eisenbud parameters $\left\{E_\lambda, \gamma_{\lambda,c}\right\}$ to construct the poles and residues of the Kapur-Peierls operator $\boldsymbol{R}_{L}$, through (\ref{eq:R_L eigenproblem}) and (\ref{eq:R_L residues widths}).
Plugging its associated expansion (\ref{eq::RL Mittag Leffler}) into the expression of the scattering matrix (\ref{eq:U expression}) then yields the Mitteg-Leffler expansion of the scattering matrix:
\begin{equation}
\boldsymbol{U}(E) \underset{\mathcal{V}(E)}{=} \boldsymbol{w}\sum_{j\geq 1} \frac{\boldsymbol{u_j}\boldsymbol{u_j}^\mathsf{T}}{E - \mathcal{E}_j} + \boldsymbol{\mathrm{Hol}}_{\boldsymbol{U}}(E)
\label{eq::U Mittag Leffler}
\end{equation}
where $\boldsymbol{w} \coloneqq 2\mathrm{i}\Id{}$ is the wronskian (\ref{eq:wronksian expression}), and the scattering residue widths $\boldsymbol{u_j}$ are defined as:
\begin{equation}
\boldsymbol{u_j} \coloneqq \left[\boldsymbol{\rho}^{1/2}\boldsymbol{O}^{-1}\right]_{E = \mathcal{E}_j}\boldsymbol{r_j}
\label{eq::u_j scattering residue width}
\end{equation}
In writing (\ref{eq::U Mittag Leffler}), we have used the fact that all the resonances of the scattering matrix $\boldsymbol{U}(E)$ come from the Kapur-Peierls poles $\left\{\mathcal{E}_j\right\}$: the poles $\left\{\omega_k\right\}$ of the outgoing wave function $\boldsymbol{O}(E)$ cancel out in (\ref{eq:U expression}) and are thus not present in the scattering matrix, this we demonstrate in theorem \ref{theo::Analytic continuation of scattering matrix cancels spurious poles}, section \ref{subsec::Spurious poles cancellation for analytically continued scattering matrix}. Cauchy's residues theorem then allows us to evaluate the residues at the pole value to obtain (\ref{eq::u_j scattering residue width}).
As for (\ref{eq::RL Mittag Leffler}), if a resonance were to be degenerate with multiplicity $M_j$, the residues would no longer be rank-one, but instead the scattering matrix residue associated to pole $\mathcal{E}_j$ would be $\sum_{m=1}^{M_j}\boldsymbol{u_j^m}{\boldsymbol{u_j^m}}^\mathsf{T}$, with $\boldsymbol{u_j^m} \coloneqq \left[\boldsymbol{\rho}^{1/2}\boldsymbol{O}^{-1}\right]_{E = \mathcal{E}_j}\boldsymbol{r_{j}^m}$.

Expression (\ref{eq::U Mittag Leffler}) exhibits the advantage that the energy dependence of the scattering matrix $\boldsymbol{U}(E)$ is untangled in a simple sum. All the resonance behavior stems from the complex poles and residue widths $\left\{\mathcal{E}_j , u_{j,c} \right\}$, which yield the familiar Breit-Wigner profiles (Cauchy-Lorentz distributions) for the cross section. Conversely, all the threshold behavior and the background are described by the holomorphic part $\boldsymbol{\mathrm{Hol}}_{\boldsymbol{U}}(E)$, which can be expanded in various forms, for instance analytically (\ref{eq::Hol_U expansion}).

This establishes the important bridge between the R-matrix parametrizations and the Humblet-Rosenfeld expansions of the scattering matrix. More precisely, Mittag-Leffler expansion (\ref{eq::U Mittag Leffler}) is identical to the Humblet-Rosenfeld expansions (10.22a)-(10.22b) in \cite{Theory_of_Nuclear_Reactions_I_resonances_Humblet_and_Rosenfeld_1961} for the neutral particles case, and (5.4a)-(5.4b) in \cite{Theory_of_Nuclear_Reactions_IV_Coulomb_Humblet_1964} for the Coulomb case. We thus here directly connect the R-matrix parameters with the Humblet-Rosenfeld resonances, parametrized by their partial widths and real and imaginary poles, as described in \cite{Theory_of_Nuclear_Reactions_III_Channel_radii_Humblet_1961_channel_Radii}.
In particular, the poles $\left\{\mathcal{E}_j\right\}$ from (\ref{eq:E_j pole def}), found by solving (\ref{eq:R_L eigenproblem}), are exactly the ones defined by equations (9.5) and (9.8) in \cite{Theory_of_Nuclear_Reactions_I_resonances_Humblet_and_Rosenfeld_1961}.
The scattering residue widths $\left\{u_{j,c}\right\}$, calculated from (\ref{eq::u_j scattering residue width}), then correspond to the Humblet-Rosenfeld complex residues (10.12) in \cite{Theory_of_Nuclear_Reactions_I_resonances_Humblet_and_Rosenfeld_1961}, from which they build their quantities $\left\{G_{c,n}\right\}$ appearing in expansions (10.22a)-(10.22b) in \cite{Theory_of_Nuclear_Reactions_I_resonances_Humblet_and_Rosenfeld_1961}, or (5.4a)-(5.4b) in \cite{Theory_of_Nuclear_Reactions_IV_Coulomb_Humblet_1964}.
Finally, the holomorphic part $\boldsymbol{\mathrm{Hol}}_{\boldsymbol{U}}(E)$ corresponds to the regular function $Q_{c,c'}(E)$ defined between (10.14a) and (10.14b) in \cite{Theory_of_Nuclear_Reactions_I_resonances_Humblet_and_Rosenfeld_1961}.

Just as Humblet and Rosenfeld did with $Q_{c,c'}(E)$ in section 10.2 of \cite{Theory_of_Nuclear_Reactions_I_resonances_Humblet_and_Rosenfeld_1961} and section 4 of \cite{Theory_of_Nuclear_Reactions_IV_Coulomb_Humblet_1964}, we do not give here an explicit way of calculating this holomorphic contribution $\boldsymbol{\mathrm{Hol}}_{\boldsymbol{U}}(E)$ other than stating that it is possible to expand it in various ways. Far from a threshold, an analytic series in $E$ can stand:
\begin{equation}
\boldsymbol{\mathrm{Hol}}_{\boldsymbol{U}}(E) \underset{\mathcal{V}(E)}{=}  \sum_{n\geq0} \boldsymbol{s}_n E^n
\label{eq::Hol_U expansion}
\end{equation}
In the immediate viscinity of a threshold, the asymptotic threshold behavior will prevail (for massive particles, $U_{c,c'}\sim k_c^{\ell_c + 1} k_{c'}^{\ell_{c'}}$, c.f. eq.(10.5) in \cite{Theory_of_Nuclear_Reactions_I_resonances_Humblet_and_Rosenfeld_1961}, or \cite{Wigner_Thresholds_1948}), yielding an expansion of the form:
\begin{equation}
\boldsymbol{\mathrm{Hol}}_{\boldsymbol{U}}(E) \underset{\mathcal{V}(E_{T_c})}{=} \sum_{n\geq0} \boldsymbol{s}_n k_c^n(E)
\label{eq::Hol_U threshold expansion}
\end{equation}
Though there is no explicit way of linking these expansions (\ref{eq::Hol_U threshold expansion}) or (\ref{eq::Hol_U expansion}) to the R-matrix Wigner-Eisenbud parameters $\left\{E_\lambda, \gamma_{\lambda, c}\right\}$, this means that the same approach as that discussed in the paragraph following equation (\ref{eq::RL Mittag Leffler expanded}) can be taken: one can provide a local set of coefficients $\left\{\boldsymbol{s}_n\right\}_{\mathcal{V}(E)}$ to expand the holomorphic part of the scattering matrix $\boldsymbol{\mathrm{Hol}}_{\boldsymbol{U}}(E)$, and then calculate the scattering matrix from the Mittag-Leffler expansion (\ref{eq::U Mittag Leffler}).

An important question is that of the radius of convergence of the Mittag-Leffler expansion, in other terms how big can the vicinity $\mathcal{V}(E)$ be?
Humblet and Rosenfeld analyze this problem in section 1.4 of \cite{Theory_of_Nuclear_Reactions_I_resonances_Humblet_and_Rosenfeld_1961}, and perform the Mittag-Leffer expansion (1.50). In the first paragraph of p.538 it is stated that Humblet demonstrated in his Ph.D. thesis that the Mittag-Leffler series will converge for $M\geq 1$ for $U(k)$, though this does not investigate the multi-channel case, and thus the multi-sheeted nature of the Riemann surface stemming from mapping (\ref{eq:rho_c(E) mapping}). They assume at the beginning of section 10.2 that this property stands in the multi-channel case and yet continue their discussion with a choice of $M=0$ that would leave the residues diverging according to their expansion (1.50). This is one reason why we chose in this article to start from a local Mittag-Leffler expansion, and then search for its domain of convergence.
General mathematical scattering theory shows that the Mittag-Leffler expansion holds at least on the whole physical sheet (c.f. theorem 0.2 p.139 of \cite{Guillope_1989_ENS}).
Though it is possible the Mittag-Leffler expansion might also hold separately on each sheet, in practice this requirement is however not needed since it is often computationally more advantageous to break down an energy region between two consecutive thresholds $\left[ E_{T_c} , E_{T_{c}+1} \right]$ into smaller vicinities.

As we see, by performing the Mittag-Leffler expansion (\ref{eq::U Mittag Leffler}), we have traded-off a finite set of real, unwound, Wigner-Eisenbud parameters $\left\{E_\lambda,\gamma_{\lambda,c}\right\}$ that completely parametrized the energy dependence of the scattering matrix through (\ref{eq:U expression}), with an infinite set of complex Siegert-Humblet parameters $\left\{\mathcal{E}_j,r_{j,c}\right\}$ and some local coefficients $\left\{\boldsymbol{s_n}\right\}_{\mathcal{V}(E)}$ for the holomorphic part --- all of which are intricately intertwined through (\ref{eq:R_L eigenproblem}) which makes them dwell on a sub-manifold of the multi-sheeted Riemann surface of mapping (\ref{eq:rho_c(E) mapping}). This additional complexity of the Siegert-Humblet parameters comes at the gain of a simple parametrization of the energy dependence for the scattering matrix: the poles and residues expansion (\ref{eq::U Mittag Leffler}). For computational purposes, this may be a trade-off worth doing.

\section{\label{sec:Invariance to chanel radius}Invariance with respect to channel radii}

Section \ref{sec:Invariance to B_c, Brune and Siegert-Humblet} provided new insights on the Wigner-Eisenbud, Brune, and Siegert-Humblet parametrizations of R-matrix theory with respect to their invariance to the arbitrary boundary condition $B_c$. As we saw, both the Brune and the Siegert-Humblet parameters are invariant under change of $B_c$, but not under change of channel radius $a_c$.
This section is dedicated to invariance properties of R-matrix parameters to a change in channel radius $a_c$. This problem is both harder and less studied than that of the invariance to the boundary conditions $B_c$. To the best of our knowledge, the only previous results on this topic are the partial differential equations on the Wigner-Eisenbud $\left\{E_\lambda , \gamma_{\lambda,c}\right\}$ parameters Teichmann derived in his Ph.D. thesis (c.f. eq. (2.29) and (2.31) sections III.2. p.27 of \cite{Teichmann_thesis_1949}), and a recent study of the limit case $a_c \to 0$ in \cite{G_Hale_channel_radius_limit_2014_PhysRevC.89.014623}.
We here focus on the Siegert-Humblet parameters $\big\{\mathcal{E}_j , r_{j,c}\big\}$. Our main result of this section resides in theorem \ref{theo::r_j,c uncer change of a_c}, which establishes a way of converting the Siegert-Humblet radioactive residue widths $\left\{ r_{j,c} \right\}$ under a change of channel radius $a_c$.

Before introducing theorem \ref{theo::r_j,c uncer change of a_c}, we bring forth the observation that the scattering matrix $\boldsymbol{U}$ is invariant under change of channel radius $a_c$, i.e. for any channel $c$ we have:
\begin{equation}
\begin{IEEEeqnarraybox}[][c]{rcl}
\frac{\partial \boldsymbol{U}}{\partial a_c}  & \ = \ &  \boldsymbol{0}
\IEEEstrut\end{IEEEeqnarraybox}
\label{eq: dU/dac = 0}
\end{equation}
Since theorem \ref{theo:: poles of U are Siegert-Humblet poles} will show that the poles of the scattering matrix are exactly the ones of the Kapur-Peierls operator $\boldsymbol{R}_{L}$, which are the Siegert-Humblet poles $\left\{\mathcal{E}_j\right\}$, invariance (\ref{eq: dU/dac = 0}) entails that the poles are invariant under change of channel radius $a_c$, i.e.
\begin{equation}
\begin{IEEEeqnarraybox}[][c]{rcl}
\frac{\partial \mathcal{E}_j}{\partial a_c} & \ = \ &  0
\IEEEstrut\end{IEEEeqnarraybox}
\label{eq: dp/dac = 0}
\end{equation}

This is not the case for the other Siegert-Humblet parameters --- the radioactive widths $\left\{ r_{j,c}\right\}$. However, one can use invariance (\ref{eq: dU/dac = 0}) to differentiate the scattering matrix $\boldsymbol{U}$ expression (\ref{eq:U expression}).
Noticing in that process that definition (\ref{eq:L expression}) and the Bloch operator projection onto the channel surface, $\rho_c = k_c a_c$, entail
\begin{equation}
\begin{IEEEeqnarraybox}[][c]{rcl}
\frac{\partial \rho_c^{1/2}O_c^{-1} }{\partial a_c}  & \ = \ & \frac{1}{a_c}\rho_c^{1/2}O_c^{-1}\left[ \frac{1}{2} - L_c\right]
\IEEEstrut\end{IEEEeqnarraybox}
\label{eq: drhoO-1/dac = 0}
\end{equation}
this enables us to establish the following partial differential equations on the Kapur-Peierls matrix operator $\boldsymbol{R}_{L}$ elements:
for the diagonal element,
\begin{equation}
\begin{IEEEeqnarraybox}[][c]{rcl}
a_c \frac{\partial {R_L}_{cc}}{\partial a_c} + (1- 2 L_c) {R_L}_{cc} - 1 & \ = \ &  0
\IEEEstrut\end{IEEEeqnarraybox}
\label{eq: dR//da cc}
\end{equation}
and for off-diagonal ones,
\begin{equation}
\begin{IEEEeqnarraybox}[][c]{rcl}
a_c \frac{\partial {R_L}_{cc'}}{\partial a_c} + (\frac{1}{2}- L_c) {R_L}_{cc'}  & \ = \ &  0
\IEEEstrut\end{IEEEeqnarraybox}
\label{eq: dR//da cc'}
\end{equation}
which can be synthesized into expression,
\begin{equation}
\begin{IEEEeqnarraybox}[][c]{rcl}
\boldsymbol{a} \frac{\partial \boldsymbol{R}_{L}}{\partial \boldsymbol{a}} + ( \frac{1}{2} \Id{} - \boldsymbol{L}) \boldsymbol{R}_{L} + \Id{}\circ\left[ (\frac{1}{2}\Id{} - \boldsymbol{L}) \boldsymbol{R}_{L} - \Id{}  \right] & \ = \ &  \boldsymbol{0}
\IEEEstrut\end{IEEEeqnarraybox}
\label{eq: dRL/da }
\end{equation}
where $\circ$ designates the Hadammard matrix product, and where we used the notation:
\begin{equation}
\begin{IEEEeqnarraybox}[][c]{rcl}
\left[\frac{\partial \boldsymbol{R}_{L}}{\partial \boldsymbol{a}}\right]_{cc'}  & \ \coloneqq \ &  \frac{\partial {R_L}_{cc'}}{\partial a_c}
\IEEEstrut\end{IEEEeqnarraybox}
\label{eq: dRL/da def}
\end{equation}
Equivalently, applying property (\ref{eq::inv M derivatie property}) on partial differential equation (\ref{eq: dRL/da def}) yields:
\begin{equation}
\begin{IEEEeqnarraybox}[][c]{rcl}
\boldsymbol{a} \frac{\partial \boldsymbol{R}_{L}^{-1}}{\partial \boldsymbol{a}} - \boldsymbol{R}_{L}^{-1} ( \frac{1}{2} \Id{} - \boldsymbol{L})  - \Id{}\circ\left[ \boldsymbol{R}_{L}^{-1}(\frac{1}{2}\Id{} - \boldsymbol{L})  - \boldsymbol{R}_{L}^{-2}  \right] & \ = \ &  \boldsymbol{0}
\IEEEstrut\end{IEEEeqnarraybox}
\label{eq: dinvRL/da }
\end{equation}

These first order partial differential equations are inconvenient to solve in that they are channel-dependent, and would thus give rise to equations for each cross term.

\subsection{\label{subsec::Radioactive width transformation under a change of channel radius} Radioactive width transformation under a change of channel radius}

A striking property of the R-matrix parametrizations is that they separate the channel contribution to each resonance, meaning that to compute, for instance, the $R_{c,c'}$ element in (\ref{eq:R expression}), one only requires the widths for each level of each channel, $\gamma_{\lambda,c}$, and not some new parameter for each specific channel pair $c,c'$ combination.
In this spirit, we here demonstrate in theorem \ref{theo::r_j,c uncer change of a_c} that the Siegert-Humblet radioactive widths $r_{j,c}$ play a similar role in that their transformation under a change of channel radius only depends on that given channel.

\begin{theorem}\label{theo::r_j,c uncer change of a_c} \textsc{Siegert-Humblet radioactive residue width $r_{j,c}$ transformation under change of channel radius $a_c$.} \\
Invariance of the scattering matrix to channel radii $a_c$ sets the following first-order linear partial differential equation on the radioactive widths $\left\{r_{j,c}\right\}$ of the Kapur-Peierls $\boldsymbol{R}_{L}$ operator residues,
\begin{equation}
\begin{IEEEeqnarraybox}[][c]{rcl}
a_c \frac{\partial {r}_{j,c}}{\partial a_c} + (\frac{1}{2}- L_c) {r}_{j,c} & \ = \ &  0
\IEEEstrut\end{IEEEeqnarraybox}
\label{eq: drcj/dac }
\end{equation}
which can be formally solved as,
\begin{equation}
\begin{IEEEeqnarraybox}[][c]{rcl}
{r}_{j,c}(a_c) = {r}_{j,c}(a_c^{(0)}) \, \sqrt{\frac{a_c^{(0)}}{a_c}} \, \exp{\left(\int_{a_c^{(0)}}^{a_c} \frac{L_c(k_c x)}{x} \mathrm{d}x\right)}
\IEEEstrut\end{IEEEeqnarraybox}
\label{eq: rcj(ac) integral form}
\end{equation}
and explicitly integrates to:
\begin{equation}
\begin{IEEEeqnarraybox}[][c]{rcl}
\frac{{r}_{j,c}(a_c)}{{r}_{j,c}(a_c^{(0)})} =  \frac{O_c(\rho_c(a_c))}{O_c(\rho_c(a_c^{(0)}))} \sqrt{\frac{a_c^{(0)}}{a_c}}
\IEEEstrut\end{IEEEeqnarraybox}
\label{eq: rcj(ac) explicit}
\end{equation}
\end{theorem}

\begin{proof}
Since we demonstrated the invariance (\ref{eq: dp/dac = 0}), the Mittag-Leffler expansion (\ref{eq::U Mittag Leffler}) then entails that $\boldsymbol{u_j}$ from (\ref{eq::u_j scattering residue width}) satisfies invariance: $\frac{\partial \boldsymbol{u_j}}{\partial a_c}   \ = \   \boldsymbol{0}$.
Applying result (\ref{eq: drhoO-1/dac = 0}) to the latter then yields partial differential equation (\ref{eq: drcj/dac }), the direct integration of which readily yields (\ref{eq: rcj(ac) integral form}). Since $L_c(\rho_c) \coloneqq \frac{\rho_c}{O_c(\rho_c)}\frac{\partial O_c(\rho_c)}{\partial \rho_c}$, (\ref{eq: rcj(ac) integral form}) integrates explicitly to (\ref{eq: rcj(ac) explicit}). This result also stands for any degenerate state of multiplicity $M_j$, where for each radioactive width $\boldsymbol{r_j^m}$ we have:
\begin{equation}
\begin{IEEEeqnarraybox}[][c]{rcl}
\frac{{r}_{j,c}^m(a_c)}{{r}_{j,c}^m(a_c^{(0)})} =  \frac{O_c(\rho_c(a_c))}{O_c(\rho_c(a_c^{(0)}))} \sqrt{\frac{a_c^{(0)}}{a_c}}
\IEEEstrut\end{IEEEeqnarraybox}
\label{eq: rcj(ac) explicit degenerate m}
\end{equation}
\end{proof}

Note that for neutral particles (massive or massless) s-waves ($\ell=0$), the outgoing wave function is $O_c(\rho(a_c)) = \mathrm{e}^{\mathrm{i}k_ca_c}$ (c.f. table \ref{tab::L_values_neutral}), so that (\ref{eq: rcj(ac) explicit}) yields ${r}_{j,c}(a_c) = {r}_{j,c}(a_c^{(0)}) \, \sqrt{\frac{a_c^{(0)}}{a_c}} \, \mathrm{e}^{\mathrm{i}k_c\left(a_c - a_c^{(0)}\right)}$. Alternatively, direclty integrating (\ref{eq: rcj(ac) integral form}) with the outgoing-wave reduced logarithmic derivative expression $L_c(\rho(a_c)) = \mathrm{i}k_c a_c$ yields back the same result. Thus for s-wave neutral channels subject to a change of channel radius, the modulus of the radioactive widths decreases proportionally to the inverse square root of the channel radius $a_c$, at least for real wavenumbers $k_c \in \mathbb{R}$, i.e. real energies above the channel threshold.

\begin{table*}
\caption{\label{tab::roots of the outgoing wave functions} Roots $\big\{\omega_n\big\}$ of the outgoing wave function $O_{\ell}(\rho)$, algebraic solutions for neutral particles up to $\ell \leq 4$.}
\begin{ruledtabular}
\begin{tabular}{ll}
$\ell$ &  \quad \quad \quad \quad \quad \quad $\big\{\omega_n\big\}$ \tabularnewline
\hline \hline
0  &  \quad \quad \quad \quad \quad \quad $\left\{ \emptyset \right\}$  \tabularnewline
1  &  \quad \quad\quad \quad \quad \quad $ \big\{  -\mathrm{i}\big\}$  \tabularnewline
2  &  \quad \quad \quad \quad\quad \quad  $ \Big\{ \frac{1}{2}\left(-\sqrt{3} - 3 \mathrm{i}\right)  \quad , \quad \frac{1}{2}\left(\sqrt{3} - 3 \mathrm{i}\right) \Big\}$
\tabularnewline
3  & \quad \quad \quad \quad\quad \quad $ \left\{  \omega_1^{\ell = 3}  \approx -1.75438 - 1.83891 \mathrm{i}  \quad , \quad \omega_2^{\ell = 3} \approx - 2.32219 \mathrm{i} \quad , \quad
 \omega_3^{\ell = 3} \approx 1.75438 - 1.83891 \mathrm{i}  \right\}$
\tabularnewline
\multicolumn{2}{c}{$ \begin{array}{l}
\omega_1^{\ell = 3} \coloneqq -2\mathrm{i}-\frac{1}{2}\left(\sqrt{3} - \mathrm{i}\right)\sqrt[3]{\frac{1}{2}\left(1+\sqrt{5}\right)} - \frac{\sqrt{3} + \mathrm{i}}{2^{2/3}\sqrt[3]{1 + \sqrt{5}}}  \\
\omega_2^{\ell = 3} \coloneqq \mathrm{i}\left(-2 + \sqrt[3]{\frac{2}{1+\sqrt{5}}} - \sqrt[3]{\frac{1}{2}\left(1+\sqrt{5} \right)} \right) \\
\omega_3^{\ell = 3} \coloneqq -2\mathrm{i}+\frac{1}{2}\left(\sqrt{3} + \mathrm{i}\right)\sqrt[3]{\frac{1}{2}\left(1+\sqrt{5}\right)} + \frac{\sqrt{3} - \mathrm{i}}{2^{2/3}\sqrt[3]{1 + \sqrt{5}}}
\end{array} $}
\tabularnewline
4  & \quad \quad \quad \quad\quad \quad$ \left\{ \begin{array}{lll}
  \omega_1^{\ell = 4}  \approx -2.65742 - 2.10379\mathrm{i}  & \omega_2^{\ell = 4} \approx -0.867234 - 2.89621 \mathrm{i} \\
  \omega_3^{\ell = 4} \approx 0.867234 - 2.89621 \mathrm{i} & \omega_4^{\ell = 4} \approx 2.65742 - 2.10379\mathrm{i}
\end{array}\right\} $
\tabularnewline
\multicolumn{2}{c}{$ \begin{array}{l}
\omega_1^{\ell = 4} \coloneqq  -\frac{5 \mathrm{i}}{2} - \frac{1}{2}\sqrt{5 + \frac{15^{2/3}}{\sqrt[3]{\frac{1}{2}\left(5 + \mathrm{i}\sqrt{35}\right)}} + \sqrt[3]{\frac{15}{2}\left(5 + \mathrm{i}\sqrt{35} \right)}}  - \frac{1}{2}\sqrt{10-\frac{15^{2/3}}{\sqrt[3]{\frac{1}{2}\left(5 + \mathrm{i}\sqrt{35} \right)}} - \sqrt[3]{\frac{15}{2}\left(5 + \mathrm{i}\sqrt{35}\right)} - \frac{10\mathrm{i}}{\sqrt{5 + \frac{15^{2/3}}{\sqrt[3]{\frac{1}{2}\left(5 + \mathrm{i}\sqrt{35}\right)}} + \sqrt[3]{\frac{15}{2}\left(5 + \mathrm{i}\sqrt{35}\right)} }}}   \\
\omega_2^{\ell = 4} \coloneqq  -\frac{5 \mathrm{i}}{2} - \frac{1}{2}\sqrt{5 + \frac{15^{2/3}}{\sqrt[3]{\frac{1}{2}\left(5 + \mathrm{i}\sqrt{35}\right)}} + \sqrt[3]{\frac{15}{2}\left(5 + \mathrm{i}\sqrt{35} \right)}} + \frac{1}{2}\sqrt{10-\frac{15^{2/3}}{\sqrt[3]{\frac{1}{2}\left(5 + \mathrm{i}\sqrt{35} \right)}} - \sqrt[3]{\frac{15}{2}\left(5 + \mathrm{i}\sqrt{35}\right)} - \frac{10\mathrm{i}}{\sqrt{5 + \frac{15^{2/3}}{\sqrt[3]{\frac{1}{2}\left(5 + \mathrm{i}\sqrt{35}\right)}} + \sqrt[3]{\frac{15}{2}\left(5 + \mathrm{i}\sqrt{35}\right)} }}}  \\
\omega_3^{\ell = 4} \coloneqq  -\frac{5 \mathrm{i}}{2} + \frac{1}{2}\sqrt{5 + \frac{15^{2/3}}{\sqrt[3]{\frac{1}{2}\left(5 + \mathrm{i}\sqrt{35}\right)}} + \sqrt[3]{\frac{15}{2}\left(5 + \mathrm{i}\sqrt{35} \right)}} - \frac{1}{2}\sqrt{10-\frac{15^{2/3}}{\sqrt[3]{\frac{1}{2}\left(5 + \mathrm{i}\sqrt{35} \right)}} - \sqrt[3]{\frac{15}{2}\left(5 + \mathrm{i}\sqrt{35}\right)} + \frac{10\mathrm{i}}{\sqrt{5 + \frac{15^{2/3}}{\sqrt[3]{\frac{1}{2}\left(5 + \mathrm{i}\sqrt{35}\right)}} + \sqrt[3]{\frac{15}{2}\left(5 + \mathrm{i}\sqrt{35}\right)} }}} \\
\omega_4^{\ell = 4} \coloneqq  -\frac{5 \mathrm{i}}{2} + \frac{1}{2}\sqrt{5 + \frac{15^{2/3}}{\sqrt[3]{\frac{1}{2}\left(5 + \mathrm{i}\sqrt{35}\right)}} + \sqrt[3]{\frac{15}{2}\left(5 + \mathrm{i}\sqrt{35} \right)}} + \frac{1}{2}\sqrt{10-\frac{15^{2/3}}{\sqrt[3]{\frac{1}{2}\left(5 + \mathrm{i}\sqrt{35} \right)}} - \sqrt[3]{\frac{15}{2}\left(5 + \mathrm{i}\sqrt{35}\right)} + \frac{10\mathrm{i}}{\sqrt{5 + \frac{15^{2/3}}{\sqrt[3]{\frac{1}{2}\left(5 + \mathrm{i}\sqrt{35}\right)}} + \sqrt[3]{\frac{15}{2}\left(5 + \mathrm{i}\sqrt{35}\right)} }}}
\end{array} $}
\end{tabular}
\end{ruledtabular}
\end{table*}

\subsection{\label{subsec::Elemental solutions under pole expansion} Elemental solutions through pole expansion}

Having just established in theorem \ref{theo::r_j,c uncer change of a_c} the integral form (\ref{eq: rcj(ac) integral form}), which integrates to (\ref{eq: rcj(ac) explicit}) as a function of the outgoing wavefunction $O_c(\rho_c)$, we here set-out in theorem \ref{theo::r_j,c uncer change of a_c explicit} to express (\ref{eq: rcj(ac) integral form}) as a function of the roots $\left\{ \omega_n \right\}$ of the outgoing wavefunction.

\begin{theorem}\label{theo::r_j,c uncer change of a_c explicit} \textsc{Elemental pole expansion solution of radioactive residue $r_{j,c}$ transformation under change of channel radius $a_c$.} \\
Transformation (\ref{eq: rcj(ac) explicit}) of Siegert-Humblet radioactive widths $\left\{r_{j,c}\right\}$ under change of channel radius $a_c^{(0)} \to a_c$, can be expressed elementally as:
\begin{equation}
\begin{IEEEeqnarraybox}[][c]{rcl}
\frac{{r}_{j,c}(a_c)}{{r}_{j,c}(a_c^{(0)}) } = \sqrt{\frac{a_c^{(0)}}{a_c}} \left( \frac{ a_c^{(0)} }{  a_c}\right)^{\ell} \mathrm{e}^{\mathrm{i}k_c\left(a_c - a_c^{(0)}\right)} \prod_{n\geq1}\left( \frac{k_c a_c - \omega_n}{k_c a_c^{(0)} - \omega_n}\right)
\IEEEstrut\end{IEEEeqnarraybox}
\label{eq: rcj(ac) - integrated by Mittag-Leffler}
\end{equation}
where $\left\{ \omega_n \right\}$ are the roots of the outgoing wave function $\Big\{ \omega_n \; | \; O_c(\omega_n) = 0 \Big\}$.\\
In the Coulomb case, there are an infinite number of such roots $\left\{ \omega_n \right\}$. \\
For neutral particle channel $c $ with angular momentum $\ell$, there exists exactly $\ell$ roots $\left\{ \omega_n \right\}_{n\in\llbracket 1, \ell \rrbracket}$, the exact and algebraically solvable values of which are reported in table \ref{tab::roots of the outgoing wave functions}, up to angular momentum $\ell = 4$.
\end{theorem}

\begin{proof}
    The proof is the element-wise integration, invoking Fubini's theorem to permute sum and integral, of (\ref{eq: rcj(ac) integral form}) using the Mittag-Leffler pole expansion (\ref{eq::Explicit Mittag-Leffler expansion of L_c}) of $L_c(\rho)$, which we established in lemma \ref{lem::Mittag-Leffler of L_c Lemma}.
    In the case of neutral particles there is a finite number of roots $\left\{ \omega_n \right\}$ so that the product in (\ref{eq: rcj(ac) - integrated by Mittag-Leffler}) is finite.
    Note that in the charged particles case, there is an infinity of roots $\left\{\omega_n\right\}$, and the Weirstrauss factorization theorem would thus usually require (\ref{eq: rcj(ac) - integrated by Mittag-Leffler}) to be cast in a Hadamard canonical representation with Weierstrass elementary factors. However, in (\ref{eq: rcj(ac) - integrated by Mittag-Leffler}), the product elements tend towards unity as $n$ goes to infinity $\left( \frac{k_c a_c - \omega_n}{k_c a_c^{(0)} - \omega_n}\right) \underset{n\to \infty}{\longrightarrow} 1$, so that the infinite product in (\ref{eq: rcj(ac) - integrated by Mittag-Leffler}) should still converge.
\end{proof}

Theorems \ref{theo::r_j,c uncer change of a_c} and \ref{theo::r_j,c uncer change of a_c explicit} make explicit the behavior of the radioactive widths $\big\{r_{j,c}\big\}$ under a change of channel radius $a_c$.
Strikingly, only the Kapur-Peierls matrix $\boldsymbol{R}_{L}$ appears in this change of variable. 
This means that the R-matrix $\boldsymbol{R}$ and the $\boldsymbol{L^0}$ matrix function suffice to both compute the Siegert-Humblet parameters $\big\{\mathcal{E}_j, r_{j,c}\big\}$ from (\ref{eq:R_L eigenproblem}), and to change the radioactive widths $\big\{r_{j,c}\big\}$ under a change of channel radius $a_c$.
This novel result portrays the Siegert-Humblet parameters as allowing a simple energy dependence to the scattering matrix (\ref{eq::U Mittag Leffler}) --- albeit locally and needing the expansion coefficients (\ref{eq::Hol_U expansion}) --- all the while being boundary condition $B_c$ independent and easy to transform under a change of channel radius $a_c$.

\section{\label{sec::Continuing the scattering matrix to complex energies while closing channels below thresholds} Continuing the scattering matrix to complex energies while closing channels below thresholds}

In section 5.2 of \cite{Theory_of_Nuclear_Reactions_I_resonances_Humblet_and_Rosenfeld_1961}, Humblet and Rosenfeld continue the scattering matrix to complex wave numbers $k_c \in \mathbb{C}$, and define corresponding open and closed channels. They however never point to the conumdrum that this entails: in their approach, the scattering matrix seemingly does not annul itself below threshold.
This is contrary to the apporach taken by Lane \& Thomas, where they explicitly annul the elements of the scattering matrix below thresholds, as stated in the paragraph between equations (2.1) and (2.2) of section VII.1. p.289 \cite{Lane_and_Thomas_1958}.
Claude Bloch ingeniously circumvents the problem by explicitly stating after eq. (50) in \cite{Bloch_1957} that the scattering matrix is a matrix of the open channels only, meaning its dimensions change as more channels open when energy $E$ increases past new thresholds $E > E_{T_c}$. In his approach, sub-threshold elements of the scattering matrix need not be annulled, one simply does not consider them.

Echoing section \ref{subsubsec::Ambiguity in shift and penetration}, we dedicate this section to this question of how to extend the scattering matrix to complex wavenumbers $k_c\in\mathbb{C}$, while closing the channels below threshold.

\subsection{\label{subsec::Forcing sub-threshold elements to zero} Forcing sub-threshold elements to zero: the legacy of Lane \& Thomas}

To close the channels for real energies below threshold, the simplest approach is the one proposed by Lane \& Thomas in \cite{Lane_and_Thomas_1958}. This approach exploits the ambiguity, discussed in section \ref{subsubsec::Ambiguity in shift and penetration}, when it comes to defining the shift $\boldsymbol{S}(E)$ and penetration $\boldsymbol{P}(E)$ factors for complex energies $E \in \mathbb{C}$. With decomposition (\ref{eq:: L = S + iP}), the scattering matrix expressions (\ref{eq:U expression}) can be re-written, for real energies above threshold, according to section VII.1. equation (1.6b) in \cite{Lane_and_Thomas_1958}:
\begin{equation}
\begin{IEEEeqnarraybox}[][c]{rcl}
\boldsymbol{U} & \ = \ &  \boldsymbol{\Omega} \left( \Id{}+ \boldsymbol{w} \boldsymbol{\mathfrak{P}}^{1/2}\boldsymbol{R}_{L}\boldsymbol{\mathfrak{P}}^{1/2}  \right) \boldsymbol{\Omega}
\IEEEstrut\end{IEEEeqnarraybox}
\label{eq:U expression L&T}
\end{equation}
with the values defined for energies above the thresholds in III.3.a. p.271 of \cite{Lane_and_Thomas_1958}:
\begin{equation}
\begin{IEEEeqnarraybox}[][c]{rcl}
\boldsymbol{\Omega} & \ \coloneqq \ & \boldsymbol{O}^{-1}\boldsymbol{I} \\
\boldsymbol{\mathfrak{P}} & \ \coloneqq \ & \boldsymbol{\rho} \boldsymbol{O}^{-1}\boldsymbol{I}^{-1}
\IEEEstrut\end{IEEEeqnarraybox}
\label{eq:Omega and mathfrak P expression L&T}
\end{equation}
Let us note that the Mittag-Leffler expansion (\ref{eq::RL Mittag Leffler}) of the Kapur-Peierls matrix $\boldsymbol{R}_{L}$ operator can still be performed.

Lane \& Thomas do not specify how they would continue the quantities (\ref{eq:Omega and mathfrak P expression L&T}) for negative energies, as they state ``we need not be concerned with stating similar relations for the negative energy channels" (c.f. paragraph after equation (4.7c), p.271.), but they do specify that $\boldsymbol{P} = \boldsymbol{0}$ below threshold energies and $\boldsymbol{P} = \boldsymbol{\mathfrak{P}}$ above.
This means that plugging-in $\boldsymbol{P} $ in place of $\boldsymbol{\mathfrak{P}}$ in (\ref{eq:Omega and mathfrak P expression L&T}) has the convenient property of automatically closing the reaction channels below threshold, since in that case $U_{c,c'} = \Omega_c \Omega_{c'}$, which annuls the off-diagonal terms of the cross section (the reaction channels $c\neq c'$) when plugged into equation (1.10) in \cite{Lane_and_Thomas_1958} VIII.1. p.291.
Note that this approach only annuls the off-diagonal terms of the scattering cross section, leaving non-zero cross sections for the diagonal $\sigma_{cc}(E)$, even below threshold.  Indeed, equation (4.5a) section III.4.a., p.271 of \cite{Lane_and_Thomas_1958} gives $\Omega_c = \mathrm{e}^{\mathrm{i}(\omega_c - \phi_c)}$, whilst the cross section is begotten by the amplitudes of the \textit{transmission matrix} $\boldsymbol{T}(E)$, defined as $T_{cc'} \coloneqq \delta_{cc'}\mathrm{e}^{2\mathrm{i}\omega_c} - U_{cc'} $ in (2.3), section VIII.2., p.292. For sub-threshold real energies, the diagonal term of the transmission matrix is thus equal to $T_{cc} = \mathrm{e}^{2\mathrm{i}\omega_c}\left( 1 - \mathrm{e}^{-2\mathrm{i}\phi_c} \right)$. This means that in the Lane \& Thomas approach, all channels $c'\neq c$ are force-closed to zero below the incomming channel threshold $E < E_{T_c}$, except for the $c \to c$ reaction, which is pudiquely overlooked as non-physical.

Of course, this approach comes at the cost of sacrificing the analytic properties of the scattering matrix $\boldsymbol{U}$: since $P_c = \Im\left[L_c\right]$, the penetration factor is no longer meromorphic and thus neither is $\boldsymbol{U}$, going directly against a vast amount of literature on the analytic properties of the scattering matrix \cite{Eden_and_Taylor, Theory_of_Nuclear_Reactions_I_resonances_Humblet_and_Rosenfeld_1961, Theory_of_Nuclear_Reactions_II_optical_model_Rosenfeld_1961, Theory_of_Nuclear_Reactions_III_Channel_radii_Humblet_1961_channel_Radii, Theory_of_Nuclear_Reactions_IV_Coulomb_Humblet_1964, Theory_of_Nuclear_Reactions_V_low_energy_penetrations_Jeukenne_1965, Theory_of_Nuclear_Reactions_VI_unitarity_Humblet_1964, Theory_of_Nuclear_Reactions_VII_Photons_Mahaux_1965, Theory_of_Nuclear_Reactions_VIII_evolutions_Rosenfeld_1965, Theory_of_Nuclear_Reactions_IX_few_levels_approx_Mahaux_1965, Humblet_1990, Pacific_Journal_Mittag_Leffler_and_spectral_theory_1960, Guillope_1989_ENS, Dyatlov}.
This is the approach presently taken by the SAMMY code at Oak Ridge National Laboratory \cite{SAMMY_2008}, and upon which thus rest numerous ENDF evaluations \cite{ENDF_VIII}.

We would like to note that under careful reading, this might not actually have been the approach intended by Lane \& Thomas in \cite{Lane_and_Thomas_1958}. Indeed, Lane \& Thomas never specify how to prolong the $\boldsymbol{\mathfrak{P}}$ to sub-threshold energies, and in equation (\ref{eq:U expression L&T}) it is $\boldsymbol{\mathfrak{P}}$ that is present and not $\boldsymbol{P}$.
They do however note in the paragraph between equations (2.1) and (2.2) of section VII.1. p.289, that ``as there are no physical situations in which the $I_c^-$ occur, the components of the [scattering matrix] are not physically significant and one might as well set them equal to zero as can be seen from (1.6b). This may be accomplished without affecting the [positive energy channels] by setting the negative energy components of the Wronskian matrix to zero; $w_c^-=0$. (This means that the $O_c^-$ and $I_c^-$ are not linearly independent.)".
The choice of wording is here important. Indeed, it says that it is possible to set the Wronskian to zero to close channels below the threshold, though it is not necessary. This is yet another way of closing subthreshold channels that would allow to keep the analytic properties of the scattering matrix, with $\boldsymbol{\mathfrak{P}} \coloneqq  \boldsymbol{\rho} \boldsymbol{O}^{-1}\boldsymbol{I}^{-1} $ still analytically continued, albeit at the cost of not knowing when in the complex plane should the Wronskian $w_c$ be set to zero --- perhaps only on $\mathbb{R}_-$, which would then become a branch line.\\
As we will show in theorem \ref{theo::Analytic continuation of scattering matrix cancels spurious poles} of section \ref{subsec::Analytic continuation of the scattering matrix}, as long as the scattering matrix is continued so as to keep the Wronskian relation (\ref{eq:wronksian expression}) intact, the poles of the outgoing scattering wave function $O_c$ cancel out of the scattering expressions (\ref{eq:U expression}) and (\ref{eq:U expression L&T}). The Wronskian condition (\ref{eq:wronksian expression}) is conserved when keeping $\boldsymbol{\mathfrak{P}} $ from (\ref{eq:Omega and mathfrak P expression L&T}) --- instead of the definition $P_c \coloneqq \Im\left[L_c\right]$, which cannot respect the Wronskian relation (\ref{eq:wronksian expression}) --- so that this approach of setting the Wronskian to zero below threshold while analytically continuing the penetration and shift factors would indeed cancel out the spurious poles of the outgoing wave functions $O_c$.

\subsection{\label{subsec::Analytic continuation of the scattering matrix} Analytic continuation of the scattering matrix}

In opposition to the Lane \& Thomas approach, an entire field of physics and mathematics has studied the analytic continuation of the scattering matrix to complex wavenumbers $k_c\in \mathbb{C}$ \cite{Dyatlov, Guillope_1989_ENS, S-matrix-complex_resonance_parameters_Csoto_1997, Favella_and_Reineri_1962, complex_energy_above_inoization_threshold_1969, Theory_of_Nuclear_Reactions_I_resonances_Humblet_and_Rosenfeld_1961, Theory_of_Nuclear_Reactions_II_optical_model_Rosenfeld_1961, Theory_of_Nuclear_Reactions_III_Channel_radii_Humblet_1961_channel_Radii, Theory_of_Nuclear_Reactions_IV_Coulomb_Humblet_1964, Theory_of_Nuclear_Reactions_V_low_energy_penetrations_Jeukenne_1965, Theory_of_Nuclear_Reactions_VI_unitarity_Humblet_1964, Theory_of_Nuclear_Reactions_VII_Photons_Mahaux_1965, Theory_of_Nuclear_Reactions_VIII_evolutions_Rosenfeld_1965}.

As we saw, there is no ambiguity as to how to continue the $\boldsymbol{L^0}$ matrix function, and thus the Kapur-Peierls matrix $\boldsymbol{R}_{L}$, to complex wave numbers. Indeed, as discussed in section \ref{subsec:External_region_waves}, the incoming $I_c(\rho_c)$ and outgoing $O_c(\rho_c)$ wave functions can be analytically continued to complex wavenumbers $k_c\in \mathbb{C}$, and through the multi-sheeted mapping (\ref{eq:rho_c(E) mapping}) to complex energies $E\in \mathbb{C}$. This naturally yields the meromorphic continuation of the scattering matrix to complex energies (\ref{eq::U Mittag Leffler}).

The shortcoming of this analytic continuation approach is that it cannot annul the channel elements of the scattering matrix for sub-threshold energies $E<E_{T_c}$. Indeed, analytic continuation (\ref{eq::U Mittag Leffler}) means the scattering matrix $\boldsymbol{U}$ is a meromorphic operator from $\mathbb{C}$ to $\mathbb{C}$ on the multi-sheeted Riemann surface of mapping (\ref{eq:rho_c(E) mapping}). Unicity of the analytic continuation then means that if the scattering matrix elements are zero below their threshold, $U_{c,c'}(E) = 0 \; , \; \forall E - E_{T_c} \in \mathbb{R}_-$, then it is identically zero for all energies on that sheet of the manifold , $U_{c,c'}(E) = 0 \; ,\; \forall E \in \mathbb{C}$. Thus, the analytic continuation formalism cannot set elements of the scattering matrix to be identically zero bellow thresholds $\left\{E_{T_c}\right\}$.

This apparent inability to close channels below thresholds is the principal reason why the nuclear data community has stuck to definition (\ref{eq:: Def S = Re[L], P = Im[L]}), legacy of Lane \& Thomas, when computing the scattering matrix in equation (\ref{eq:U expression}). This has been the subject of an ongoing controversy in the field on how to continue the scattering matrix to complex wave numbers.

This article argues that analytic continuation (\ref{eq:: Def S and P analytic continuation from L}) is the physically correct way of continuing the scattering matrix to complex energies.
To support this, we advance and demonstrate three new arguments: analytic continuation cancels out spurious poles otherwise introduced by the outgoing wavefunctions $O_c$ (theorem \ref{theo::Analytic continuation of scattering matrix cancels spurious poles}); analytic continuation respects generalized unitarity (theorem \ref{theo::satisfied generalized unitarity condition}); and, for massive particles, analytic continuation of real wavenumber expressions to sub-threshold energies naturally sees the transmission matrix evaness on the physical sheet (theorem \ref{theo::evanescence of sub-threshold transmission matrix}), while always closing the channels by annulling the cross section (theorem \ref{theo::Analytic continuation annuls sub-threshold cross sections}).

\subsection{\label{subsec::Spurious poles cancellation for analytically continued scattering matrix} Spurious poles cancellation for analytically continued scattering matrix}

We here show that the outgoing-wave $\boldsymbol{O}$ \textit{a priori} introduces spurious poles to the scattering matrix $\boldsymbol{U}$; but show how these are cancelled out if the R-matrix $\boldsymbol{R}$, the wave-functions $\boldsymbol{O}$, $\boldsymbol{I}$, and thus the $\boldsymbol{L^0}$ matrix function, are analytically continued to complex wavenumbers, while maintaining a constant Wronskian (\ref{eq:wronksian expression}).

\subsubsection{\label{subsubsec:: Semi-simple poles in R-matrix theory} Assuming semi-simplicity of poles in R-matrix theory}

Let us first start with a note on high-order poles.
Being a high-order pole, as opposed to a simple pole, can bear various meanings. In our context, the three following definitions are of interest: a) \textit{Laurent order}: the order of the polar expansion in the Laurent development in the vicinity of a pole; b) \textit{Algebraic multiplicity}: the multiplicity of the root of the resolvant at a pole value; c) \textit{Geometric multiplicity}: the dimension of the associated nullspace.

From equation (\ref{eq::RL Mittag Leffler degenerate state}) and throughout the article, we have treated the case of degenerate states where the geometric mulpitplicity $M_j > 1$ was higher than one, leading to rank-$M_j$ residues.
We have however always assumed the Laurent order to be one: in equation (\ref{eq::RL Mittag Leffler degenerate state}), the residues might be rank-$M_j$, but the Laurent order is still unity (no $\frac{1}{(E-\mathcal{E}_j)^2}$ or higher Laurent orders).

In the general case the Laurent order is greater than one but it does not equal geometric or algebraic multiplicity.
In terms of Jordan normal form, if the Jordan cells had sizes $n_1, ..., n_{m_g}$, then the geometric multiplicity is equal to $m_g$, the algebraic multiplicity $m_a$ is the sum $ m_a = n_1 + ... + n_{m_g}$, and the Laurent order is the maximum $\mathrm{max}\left\{n_1, ..., n_m\right\}$.

Alternatively, these can be defined as follows:
Let $\boldsymbol{M}(z)$ be a complex symmetric meromorphic matrix operator, with a root at $z=z_0$ (i.e. $\boldsymbol{M}(z_0)$ is non-invertible).
The algebraic multiplicity $m_a$ is the first non-zero derivative of the determinant, i.e. the first integer $m_a\in\mathbb{N}$ such that $\left.\frac{\mathrm{d}^{m_a}}{\mathrm{d}z^{m_a}}\mathrm{det}\Big( \boldsymbol{M}(z) \Big)\right|_{z=z_0} \neq 0$ ; alternatively, using Cauchy's theorem, the first integer $m_a$ such that $\oint_{\mathcal{C}_{z_0}}  \frac{\boldsymbol{M}(z)}{(z-z_0)^{m_a}} \mathrm{d}z = \boldsymbol{0}$.
The geometric multiplicity $m_g$ is the dimension of the kernel (nullspace), i.e. $m_g = \mathrm{dim}\left( \mathrm{Ker}\left(\boldsymbol{M}(z_0)\right)\right)$.
In general the algebraic multiplicity is greater than the geometric one: $m_a \geq m_g$.

$\boldsymbol{M}(z_0)$ is said to be \textit{semi-simple} if its geometric and algebraic multiplicities are equal, i.e. $m_a = m_g$ (c.f. theorem 2, p.120 in \cite{Elements_de_Mathematique_Algebre_Bourbaki_1959}).
Semi-simplicity can be established using the following result: $\boldsymbol{M}(z_0)$ is semi-simple if, and only if, for each nonzero $v \in \mathrm{Ker}\left( \boldsymbol{M}(z_0)\right)$, there exists $ w \in \mathrm{Ker}\left( \boldsymbol{M}(z_0)\right)$ such that:
\begin{equation}
v^\mathsf{T} \Bigg(\left. \frac{\mathrm{d}\boldsymbol{M}}{\mathrm{d}z} \right|_{z=z_0}  \Bigg)w \neq 0
\label{eq::semisimple condition}
\end{equation}

If an operator $\boldsymbol{M}(z_0)$ is semi-simple at a root $z_0$, then $z_0$ is a pole of Laurent-order one for the inverse operator $\boldsymbol{M}^{-1}(z) \underset{\mathcal{V}(z_0)}{\sim} \frac{\boldsymbol{\widetilde{M}}}{z-z_0} $.
For Hermitian operators, the semi-simplicity property is guaranteed. However, resonances seldom correspond to Hermitian operators. In our case, the resonances correspond to the poles of the scattering matrix $\boldsymbol{U}(E)$, which is not self-adjoint but complex symmetric $\boldsymbol{U}^\mathsf{T} = \boldsymbol{U}$ (c.f. equation (2.15) section VI.2.c p.287 in \cite{Lane_and_Thomas_1958}).
For complex symmetric operators, semi-simplicity is not guaranteed in general, even when discarding the complex case of quasi-nul vectors.

In the case of R-matrix theory, we were able to find cases where the geometric multiplicity of the scattering matrix does not match the algebraic one, thus R-matrix theory does not always yield semi-simple scattering matrices, and the Laurent development orders of the resonance poles can be higher.
For instance, we can devise examples of non-semi-simple inverse level matrices from definition (\ref{eq:inv_A expression}) by choosing resonance parameters such that the algebraic multiplicity is striclty greated than the geometric one.

However, one can also observe in these simple cases that the space of parameters for which semi-simplicity is broken is a hyper-plane of the space of R-matrix parameters. 
This gives credit to the traditional physicis arguments that the probability of this occurring is quasi-nul: R-matrix theory can yield scattering matrices with Laurent orders higher than one, but this is extremely unlikely; a mathematical approach of generic simplicity of resonances can be found in chapter 4 ``Black Box Scattering in $\mathbb{R}^n$" of \cite{Dyatlov}, in particular theorem 4.4 (Meromorphic continuation for black box Hamiltonians), theorem 4.5 (Spectrum of black box Hamiltonians), theorem 4.7 (Singular part of RV($\lambda$) for black box Hamiltonians), and theorem 4.39 (Generic simplicity of resonances for higher dimensional black box with potential perturbation).
In other terms, we assume semi-simplicity is almost always guaranteed through R-matrix parametrizations.

Henceforth, we use this agrument to continue assuming the Kapur-Peirels matrix $\boldsymbol{R}_{L}$ is usually semi-simple, and thus the Laurent order of the radioactive poles $\big\{ \mathcal{E}_j \big\}$ in (\ref{eq::RL Mittag Leffler degenerate state}) is, in practice, one.

But let us be aware that in general scattering theory, the scattering operator may exhibit high-order poles \cite{Dyatlov, Guillope_1989_ENS, Owusu_2009}, and efforts are being made to have these ``exceptional points'' of second order arise in the specific case of nuclear interactions \cite{Exceptional_points_scattering, Exceptional_points_Hamiltonian_Michel}. The traditional R-matrix assumption where the poles of the scattering matrix are almost-always of Laurent-order one is unable to describe these physical phenomena.

\subsubsection{\label{subsubsec:: Spurious poles from O} Outgoing wave $\boldsymbol{O}$ introduces spurious poles}

We have given reasons to assume that the poles of the Kapur-Peirels matrix $\boldsymbol{R}_{L}$ are simple (i.e. or Laurent order one), however, looking at (\ref{eq:U expression}) shows that the roots of the outgoing wave functions $\boldsymbol{O}$ could endow the scattering matrix with additional poles, through $\boldsymbol{O}^{-1}$, and that these poles could potentially have higher Laurent orders, since $\boldsymbol{O}^{-1}$ appears twice in expression (\ref{eq:U expression}). 

We here establish through Lemma \ref{lem::Diagonal semi-simplicity Lemma} that $\boldsymbol{O}^{-1}$ is semi-simple because it is diagonal with simple roots. This result will be used in theorem \ref{theo::Analytic continuation of scattering matrix cancels spurious poles} to show that these poles cancel out of the scattering matrix if the Wronskian condition (\ref{eq:wronksian expression}) is maintained.

\begin{lem}\label{lem::Diagonal semi-simplicity Lemma}
\textsc{Diagonal Semi-Simplicity} --
If a diagonal matrix $\boldsymbol{D}^{-1}(z)$ is composed of elements with simple roots $\big\{\omega_k\big\}$, then its inverse is semi-simple, i.e. when a pole $\omega_k$ of a diagonal matrix $\boldsymbol{D}(z)$ has an algebraic multiplicity $M_k > 1$ the Laurent development order of the pole remains 1 while the associated residue matrix is of rank $M_k$, and can be expressed as:
\begin{equation}
\boldsymbol{D}(z) \underset{\mathcal{V}(z=\omega_k)}{=} \boldsymbol{D}_0 + \frac{\boldsymbol{D}_k}{z-\omega_k}
\end{equation}
with
\begin{equation}
\begin{IEEEeqnarraybox}[][c]{rcl}
\boldsymbol{D}_k & = & \sum_{m=1}^{M_k}\frac{\boldsymbol{v_k}\boldsymbol{v_k}^\mathsf{T} }{\boldsymbol{v_k}^\mathsf{T} {\boldsymbol{D_0^{-1}}}^{(1)} \boldsymbol{v_k}}
\IEEEstrut\end{IEEEeqnarraybox}
\label{eq::higher algebraic multiplicity scattering poles}
\end{equation}
\end{lem}

\begin{proof}
Without loss of generality, a change of variables can be performed so as to set $\omega_k = 0$.
Let $\boldsymbol{D}(z) = \boldsymbol{\mathrm{diag}}\left( d_1(z), d_2(z), \hdots, d_1(z), d_j(z), d_n(z) \right) $ be a diagonal meromorphic complex-valued operator, which admits a pole at $z=0$.
$\boldsymbol{D^{-1}}(z) = \boldsymbol{\mathrm{diag}}\left( d_1^{-1}(z), d_2^{-1}(z), \hdots, d_1^{-1}(z), d_j^{-1}(z), d_n^{-1}(z) \right) $ is well known, and
$\mathrm{det}\left( \boldsymbol{D^{-1}} \right)(z=0) = d_1^{-1}(z)^2 \prod_{j \neq 1 }d_j^{-1}(z)$.
Let us assume only $d_1^{-1}(z=0) = 0$, with a simple root, so that $d_1(z) \underset{\mathcal{V}(z=0)}{=} d0_1 + \frac{R_1}{z}$.
Then $\mathrm{det}\left( \boldsymbol{D^{-1}}(z) \right)(z=0) = d_1^{-1}(z)^2 \prod_{j \neq 1 }d_j^{-1}(z)$ has a double root: the algebraic multiplicity is thus $2$.
However, it is immediate to notice that:
\begin{equation*}
\begin{IEEEeqnarraybox}[][c]{rcl}
\boldsymbol{D}(z) & = & \boldsymbol{\mathrm{diag}}\left( d_1(z), d_2(z), \hdots, d_1(z), d_j(z), d_n(z) \right) \\  & \underset{\mathcal{V}(z=0)}{=} & \boldsymbol{\mathrm{diag}}\left( d0_1, d_2(z), \hdots, d0_1(z), d_j(z), d_n(z) \right) \\
& & + \frac{1}{z} \boldsymbol{\mathrm{diag}}\left( R_1, 0, \hdots, R_1, 0, 0 \right)
\IEEEstrut\end{IEEEeqnarraybox}
\end{equation*}
This means the Laurent development order remains 1, albeit the algebraic multiplicity of the pole is 2 (or higher $M_k$). Thus, it can be written that:
\begin{equation*}
\boldsymbol{D}(z) \underset{\mathcal{V}(z=0)}{=} \boldsymbol{D}_0 + \frac{\boldsymbol{D}_1}{z}
\end{equation*}
When solving for the non-linear Eigenproblem
\begin{equation*}
\boldsymbol{D^{-1}}(z) \boldsymbol{v} = \boldsymbol{0}
\end{equation*}
the kernel is no longer an eigenline, but instead spans $\left( \boldsymbol{v_1}, \boldsymbol{v_2}  \right)$,
\begin{equation*}
\mathrm{Ker}\left(\boldsymbol{D^{-1}}_0\right) = \mathrm{span}\left( \boldsymbol{v_1}, \boldsymbol{v_2}  \right)
\end{equation*}
with $\boldsymbol{v_1} = a_1 \left[ 1, 0 , \hdots, 0, 0 ,0 \right]^\mathsf{T} $ and $\boldsymbol{v_2} = a_2 \left[ 0, 0 , \hdots, 1, 0 ,0 \right]^\mathsf{T} $.
Then, following Gohberg-Sigal's theory \cite{Gohberg_Sigal_1971}, the fundamental property:
\begin{equation*}
\boldsymbol{D^{-1}}\boldsymbol{D} = \Id{}
\end{equation*}
and the Laurent development around the pole:
\begin{equation*}
\begin{IEEEeqnarraybox}[][c]{rcl}
\boldsymbol{D^{-1}}(z) & \underset{\mathcal{V}(z=0)}{=} & \boldsymbol{D^{-1}}_0 + z {\boldsymbol{D^{-1}}}^{(1)}_0 + \mathcal{O}\left(z^2\right)
\IEEEstrut\end{IEEEeqnarraybox}
\end{equation*}
yield the relations:
\begin{equation*}
\begin{IEEEeqnarraybox}[][c]{rcl}
\boldsymbol{D^{-1}}_0  \boldsymbol{D}_0 + {\boldsymbol{D^{-1}}}^{(1)}_0 \boldsymbol{D}_1  & =  & \Id{} \\
\boldsymbol{D^{-1}}_0  \boldsymbol{D}_1 & =  & \boldsymbol{0}
\IEEEstrut\end{IEEEeqnarraybox}
\end{equation*}
Constructing $\boldsymbol{D}_1$ to satisfy the latter then entails \begin{equation*}
\begin{IEEEeqnarraybox}[][c]{rcl}
\boldsymbol{D}_1 & = & \frac{\boldsymbol{v_1}\boldsymbol{v_1}^\mathsf{T} }{\boldsymbol{v_1}^\mathsf{T} {\boldsymbol{D_0^{-1}}}^{(1)} \boldsymbol{v_1}} +  \frac{\boldsymbol{v_2}\boldsymbol{v_2}^\mathsf{T} }{\boldsymbol{v_2}^\mathsf{T} {\boldsymbol{D_0^{-1}}}^{(1)} \boldsymbol{v_2}}
\IEEEstrut\end{IEEEeqnarraybox}
\end{equation*}
where the transpose is used because the matrix is complex symmetric.
This reasoning immediately generalizes to expression (\ref{eq::higher algebraic multiplicity scattering poles}).
\end{proof}

Let $\left\{\omega_k\right\}$ be all the roots of the outgoing-wave functions (i.e. the poles of inverse outgoing wave $\boldsymbol{O}^{-1}$), which we can find by solving the non-linear Eigenvalue problem:
\begin{equation}
\boldsymbol{O}(\omega_k) {\boldsymbol{{w_{k}}}}_m = \boldsymbol{0}
\end{equation}
Since $\boldsymbol{O}^{-1}$ is diagonal, Lemma \ref{lem::Diagonal semi-simplicity Lemma} entails it is semi-simple: the algebraic multiplicities are equal to the geometric multiplicities, and thus the poles $\left\{\omega_k\right\}$ all have Laurent-order one.

Situations can arise where same-charge channels within the same total angular momentum $J^\pi$ will carry same angular momenta $\ell_c = \ell_{c'}$ and equal channel radii $a_c = a_{c'}$.
In that case, the geometric multiplicity $M_k$ of pole $\omega_k$ will be equal to the number of channels sharing the same functional outgoing waves $O_c = O_{c'}$.
Diagonal semi-simplicity Lemma \ref{lem::Diagonal semi-simplicity Lemma} then establishes that the residue of $\boldsymbol{O}^{-1}$ associated to pole $\omega_k$ is now a diagonal rank-$M_k$ matrix, $\boldsymbol{D_k}$, expressed as:
\begin{equation}
\boldsymbol{D_k} = \sum_{m=1}^{M_k} \frac{{\boldsymbol{w_k}}_m {\boldsymbol{w_k}}_m^\mathsf{T}}{{\boldsymbol{w_k}}_m^\mathsf{T} \boldsymbol{O}^{(1)}(\omega_k) {\boldsymbol{w_k}}_m}
\label{eq::D_k residue of O-1}
\end{equation}
where $\boldsymbol{O}^{(1)}(\omega_k)$ designates the first derivative of $\boldsymbol{O}$, evaluated at the pole value $\omega_k$.
This establishes the existence of higher-rank residues associated to the inverse outgoing wave function $\boldsymbol{O}^{-1}$.

Notice that if the channel radii $\big\{a_c\big\}$ where chosen at random, these high-rank residues would almost never emerge (null probability).
However, since $a_c$ is chosen arbitrarily in the context of R-matrix theory, it is often the case that evaluators set $a_c$ to a fixed value for multiple different channels, and even across isotopes.
This means that in practice these high-rank residues are legion.

\subsubsection{\label{subsubsec:: Poles cancelation for analytical U} Poles from the outgoing waves $\boldsymbol{O}$ cancel out of the analitically continued scattering matrix $\boldsymbol{U}$}

We just established that the poles $\big\{\omega_k\big\}$ of the inverse outgoing wave $\boldsymbol{O}^{-1}$ had Laurent-order one, potentially with higher-rank residues (\ref{eq::D_k residue of O-1}). At first sight, (\ref{eq:U expression}) seems to entail these $\left\{\omega_k\right\}$ poles should also be poles of the scattering matrix $\boldsymbol{U}$, possibly with Laurent order two.
We here establish with theorem \ref{theo::Analytic continuation of scattering matrix cancels spurious poles} that if the Wronskian condition (\ref{eq:wronksian expression}) is satisfied through analytic continuation, the $\big\{\omega_k\big\}$ poles cancel out of the scattering matrix $\boldsymbol{U}$, leaving only the poles $\big\{\mathcal{E}_j\big\}$ of the Kapur-Peierls matrix $\boldsymbol{R}_{L}$ as the scattering poles.

\begin{theorem}\label{theo::Analytic continuation of scattering matrix cancels spurious poles} \textsc{Analytic continuation of scattering matrix cancels spurious poles}. \\
If the Wronskian condition (\ref{eq:wronksian expression}) is satisfied through analytic continuation of the Kapur-Peierls operator $\boldsymbol{R}_{L}$ and the wavefunctions $\boldsymbol{I}$ and $\boldsymbol{O}$, then the poles $\big\{\omega_k\big\}$ of the inverse outgoing wave $\boldsymbol{O}^{-1}$ cancel out of the scattering matrix $\boldsymbol{U}$ in equation (\ref{eq:U expression}).
\end{theorem}

\begin{proof}
Consider the scattering matrix expression $\boldsymbol{U} =  \boldsymbol{O}^{-1} \left[\boldsymbol{I} + 2i \boldsymbol{\rho}^{1/2} \boldsymbol{R}_{L} \boldsymbol{O}^{-1} \boldsymbol{\rho}^{1/2}\right]$ from (\ref{eq:U expression}).
Result (\ref{eq::D_k residue of O-1}) entails that, in the viscinity of $\omega_k$, root of the outgoing wave-function $\boldsymbol{O}$, the residue is locally given by:
\begin{equation}
\begin{IEEEeqnarraybox}[][c]{rcl}
\boldsymbol{U}(z) & \ \underset{\mathcal{V}(E=\omega_k)}{=} \ & \boldsymbol{U}_0(\omega_k) + \frac{ \boldsymbol{D}_k \left[\boldsymbol{I} + 2i \boldsymbol{\rho}^{1/2} \boldsymbol{R}_{L} \boldsymbol{O}^{-1} \boldsymbol{\rho}^{1/2} \right]_{E=\omega_k} }{E-\omega_k}
\IEEEstrut\end{IEEEeqnarraybox}
\label{eq:: U at omega_k}
\end{equation}
We now notice that evaluating the Kapur-Peirels $\boldsymbol{R}_{L}$ operator (\ref{eq:Kapur-Peirels Operator and Channel-Level equivalence}) at the pole value $\omega_k$ yields the following equality:
\begin{equation}
\begin{IEEEeqnarraybox}[][c]{rcl}
\boldsymbol{R}_{L} \boldsymbol{O}^{-1} (\omega_k) {\boldsymbol{w_k}}_m  & \ = \ & - \left[ \boldsymbol{\rho} \boldsymbol{O}^{(1)} \right]^{-1}(\omega_k) {\boldsymbol{w_k}}_m
\IEEEstrut\end{IEEEeqnarraybox}
\label{eq::R_L O at outgoing scattering poles}
\end{equation}
Plugging (\ref{eq::R_L O at outgoing scattering poles}) into the residue of (\ref{eq:: U at omega_k}), and using the fact that (\ref{eq::D_k residue of O-1}) guarantees $\boldsymbol{D_k} $ is a linear combination of $ {\boldsymbol{w_k}}_m {\boldsymbol{w_k}}_m^\mathsf{T} $, we then have the following equality on the residues at poles $\omega_k$:
\begin{equation}
\begin{IEEEeqnarraybox}[][c]{rcl}
\boldsymbol{D}_k \left[\boldsymbol{I} + 2i \boldsymbol{\rho}^{1/2} \boldsymbol{R}_{L} \boldsymbol{O}^{-1} \boldsymbol{\rho}^{1/2} \right]_{E=\omega_k} = \boldsymbol{D}_k \left[ \boldsymbol{I} - 2i {\boldsymbol{O}^{(1)}}^{-1} \right]_{E=\omega_k}
\IEEEstrut\end{IEEEeqnarraybox}
\label{eq::U residues at outgoing scattering poles}
\end{equation}
The rightmost term is diagonal and independent from the resonance parameters.
Since the Wronskian matrix $\boldsymbol{w}$ of the external region interaction (for Coulomb or free particles) is constant, $\boldsymbol{w} = \boldsymbol{O}^{(1)}\boldsymbol{I} - \boldsymbol{I}^{(1)}\boldsymbol{O} = 2i \Id{}$, evaluating at outgoing wave-function root $\omega_k$, one finds $2i \Id{} = \boldsymbol{O}^{(1)}\boldsymbol{I}(\omega_k)$.
Plugging this result into (\ref{eq::U residues at outgoing scattering poles}) annuls the corresponding residue from the scattering matrix, i.e.:
\begin{equation}
\begin{IEEEeqnarraybox}[][c]{c}
\boldsymbol{D}_k \left[\boldsymbol{I} + 2i \boldsymbol{\rho}^{1/2} \boldsymbol{R}_{L} \boldsymbol{O}^{-1} \boldsymbol{\rho}^{1/2} \right]_{E=\omega_k}  = \boldsymbol{0}
\IEEEstrut\end{IEEEeqnarraybox}
\end{equation}
Thus, if the Wronskian condition (\ref{eq:wronksian expression}) is respected, the $\big\{\omega_k\big\}$ poles cancel out of the scattering matrix $\boldsymbol{U}$
\end{proof}

This new result unveils the importance of performing analytic continuation of the outgoing $O_c$ and incoming $I_c$ wave functions and maintain a constant Wronskian (\ref{eq:wronksian expression}), without which the $\big\{\omega_k\big\}$ poles would not cancel from the scattering matrix $\boldsymbol{U}$.

Finally, theorem \ref{theo:: poles of U are Siegert-Humblet poles} is a direct corollary of theorem \ref{theo::Analytic continuation of scattering matrix cancels spurious poles}, and having assumed the poles of the Kapur-Peirels operator almost always be of Laurent order one:

\begin{theorem}\label{theo:: poles of U are Siegert-Humblet poles} \textsc{In R-matrix theory the Scattering matrix poles are the Kapur-Peirels radioactive poles}. \\
In the context of the R-matrix scattering model, when the scattering matrix $\boldsymbol{U}$ is analytically continued to complex energies $E\in\mathbb{C}$ such as to respect the Wronskian condition (\ref{eq:wronksian expression}), the poles of the scattering matrix $\boldsymbol{U}$ are exactly the poles of the Kapur-Peierls operator $\boldsymbol{R}_{L}$, i.e. the Siegert-Humblet poles $\big\{ \mathcal{E}_j \big\}$ from (\ref{eq:R_L eigenproblem}) and (\ref{eq:E_j pole def}).
These poles are almost always of Laurent-order of one.
\end{theorem}

Importantly, both the Lane \& Thomas force-closing of sub-threshold channels \ref{subsec::Forcing sub-threshold elements to zero} or the analytic continuation \ref{subsec::Analytic continuation of the scattering matrix} will yield the same cross section values for real energies above thresholds.
However, we here demonstrated that the choice of analytic continuation in equation (\ref{eq:U expression}), respecting the Wronskian condition (\ref{eq:wronksian expression}), leads to the cancellation from the scattering matrix $\boldsymbol{U}$ of the $\big\{\omega_k\big\}$ spurious poles, which have nothing to do with the resonant states of the scattering system.
This cancellation is thus physically accurate, and would not take place had the choice of $\boldsymbol{\mathfrak{P}} = \boldsymbol{P}$ been made in equation (\ref{eq:U expression L&T}), as discussed in section \ref{subsec::Forcing sub-threshold elements to zero}. Indeed, choosing definition (\ref{eq:: Def S = Re[L], P = Im[L]}), i.e. $\boldsymbol{P} = \Im\left[ \boldsymbol{L}(z) \right] \; \in \mathbb{R}$, will fail to cancel out the $\big\{\omega_k\big\}$ poles. Conversely, defining the penetration by analytic continuation (\ref{eq:: Def S analytical}) as $ \boldsymbol{P}(z) \coloneqq \frac{1}{2i}\left( \boldsymbol{L}(z) - \left[\boldsymbol{L}(z^*)\right]^* \right)  \; \in \mathbb{C}$ will guarantee the cancellation of the $\big\{\omega_k\big\}$ poles from the scattering matrix $\boldsymbol{U}$ if using (\ref{eq:U expression L&T}).
Notice this is almost the definition (\ref{eq:Delta_L def}) of $\boldsymbol{ \Delta L}(\rho)$ we hereafter use in the proof of the generalized unitarity.
Then, to force-close sub-threshold channels, one could set the Wronskian to zero, as proposed by Lane \& Thomas in the paragraph between equation (2.1) and (2.2) of section VII.1. p.289. This shifts the problem to how to maintain the Wronskian condition (\ref{eq:wronksian expression}) while setting the Wronskian to zero below thresholds. 
Alternatively, we here argue in section \ref{subsec::Closure of sub-threshold cross sections through analytic continuation} that this might not be necessary, as analytic continuation can naturally close sub-threshold channels.

\subsection{\label{subsec::Generalized unitarity} Generalized unitarity for analytically continued scattering matrix}

One of us, G. Hale, proved a somewhat more esoteric argument in favor of analytic continuation of the scattering matrix, showing it satisfies generalized unitarity. 

Eden \& Taylor established a generalized unitarity condition, eq. (2.16) in \cite{Eden_and_Taylor}, which extents the one described by Lane \& Thomas, eq. (2.13), VI.2.c. p.287, in that the subset of open channels is unitary (thus conserving probability), but the scattering matrix can still be continued to sub-threshold channels and be non-zero, that is the full scattering matrix of open and closed channels is not unitary but satisfies the generalized unitarity condition. This is also consistent with approaches other than R-matrix to modeling nuclear interactions (c.f. commentary above eq. (3) p.4 in \cite{Exceptional_points_Hamiltonian_Michel}, \cite{complex_energy_above_inoization_threshold_1969}, or\cite{S-matrix-complex_resonance_parameters_Csoto_1997}).

The premises of the problem lies again in the multi-sheeted Rieman surface spawning from mapping (\ref{eq:rho_c(E) mapping}): when considering the scattering matrix $\boldsymbol{U}(E)$ at a given energy $E$ from (\ref{eq:conservation of energy E = E_c = E_c'}), there are multiple possibilities for the choice of wavenumber $k_c$ at each channel. 
Following Eden \& Taylor eq. (2.14a) and eq. (2.14b) \cite{Eden_and_Taylor}, we consider the case of momenta being continued along the following paths in the multi-sheeted Rieman surface: one subset of channels $c$, denoted by $\widehat{\mathfrak{C}}$, is continued as $k_{c \in \widehat{\mathfrak{C}}} \to k_{c \in \widehat{\mathfrak{C}}}^*  $, while all the others are continued as $k_{c \not\in \widehat{\mathfrak{C}}} \to - k_{c \not\in \widehat{\mathfrak{C}}}^*  $, and we collectively denote this continuation $\boldsymbol{k} \to \boldsymbol{\widetilde{k}}$:
\begin{equation}
\boldsymbol{k} \to \boldsymbol{\widetilde{k}} : \left\{ \begin{array}{rcl}
\forall c \in \widehat{\mathfrak{C}}\; , \quad k_c & \to & k_c^* \\
\forall c \not\in \widetilde{C}\; , \quad k_c & \to & -k_c^* \\
\end{array}\right. 
\label{eq:: Eden and Taylor continuation}
\end{equation}
We then seek to reproduce the generalized unitarity property eq. (2.16) of \cite{Eden_and_Taylor}, which states that the submatrix $\widehat{\boldsymbol{U}}$ composed of the channels ${c \in \widehat{\mathfrak{C}}}$, verifies the generalized unitarity condition:
\begin{equation}
\widehat{\boldsymbol{U}}(\boldsymbol{k}) \Big[ \widehat{\boldsymbol{U}}(\boldsymbol{\widetilde{k}}) \Big]^\dagger = \Id{}
\label{eq:: Eden and Taylor Generalized unitarity}
\end{equation}
We now show that analytically continuing the R-matrix expression (\ref{eq:U expression}) ensures the scattering matrix respects Edan \& Taylor generalized unirarity condition.

\begin{theorem}\label{theo::satisfied generalized unitarity condition} \textsc{Analytic continuation of the R-matrix expression for the scattering matrix ensures generalized unirarity}. \\
By performing the analytic continuation of the R-matrix expression (\ref{eq:U expression}), the scattering matrix $\boldsymbol{U}$ satisfies Edan \& Taylor's generalized unitarity condition (\ref{eq:: Eden and Taylor Generalized unitarity}). 
\end{theorem}

\begin{proof}
The proof is based on the conjugacy relations of the outgoing and incoming wavefunctions eq. (2.12), VI.2.c. in \cite{Lane_and_Thomas_1958}, whereby, for any channel $c$:
\begin{equation}
\begin{IEEEeqnarraybox}[][c]{rclcrcl}
\big[O_c(k_c^*)\big]^* & \ = \ & I_c(k_c) \quad  &,& \quad \big[I_c(k_c^*)\big]^* & \ = \ & O_c(k_c) \\
O_c(-k_c) & \ = \ & I_c(k_c) \quad &,& \quad I_c(-k_c) & \ = \ & O_c(k_c) \\ -O_c^{(1)}(-k_c) & \ = \ & I_c^{(1)}(k_c) \quad &,& \quad -I_c^{(1)}(-k_c) & \ = \ & O_c^{(1)}(k_c) \\
\IEEEstrut\end{IEEEeqnarraybox}
\label{eq:Conjugacy for O and I}
\end{equation}
where the third line was obtained by taking the derivative of the second. 
Recalling the definition of the outgoing-wave reduced logarithmic derivative (\ref{eq:L expression}), these congugacy relations (\ref{eq:Conjugacy for O and I}) entail the following on $\boldsymbol{L}$:
\begin{equation}
\begin{IEEEeqnarraybox}[][c]{rclcrcl}
\Big[L_c(k_c^*)\Big]^* & \ = \ & L_c(-k_c) \quad  &,& \quad \Big[L_c(-k_c^*)\Big]^* & \ = \ & L_c(k_c)
\IEEEstrut\end{IEEEeqnarraybox}
\label{eq:L_c congugacy}
\end{equation}
We also notice that the Wronskian condition (\ref{eq:wronksian expression}) is equivalent to:
\begin{equation}
\begin{IEEEeqnarraybox}[][c]{rcl}
\frac{2\mathrm{i}\rho_c}{O_c I_c} & \ = \ & \rho_c\left[\frac{O_c^{(1)}}{O_c} - \frac{I_c^{(1)}}{I_c}\right]
\IEEEstrut\end{IEEEeqnarraybox}
\label{eq:wronksian expression for generalized unitarity}
\end{equation}
one recognizes here the definition (\ref{eq:L expression}) of $\boldsymbol{L}$, and, using the conjugacy relations (\ref{eq:L_c congugacy}), the  Wronskian condition (\ref{eq:wronksian expression for generalized unitarity}) can be expressed as a difference of the reduced logarithmic $L_c$ derivatives as:
\begin{equation}
\begin{IEEEeqnarraybox}[][c]{rcl}
 \Delta L_c (k_c)&  \ \coloneqq \ & L_c(k_c) - L_c(-k_c) \ = \ \frac{2\mathrm{i}\rho_c}{O_c I_c}(k_c)
\IEEEstrut\end{IEEEeqnarraybox}
\label{eq:Delta_L def}
\end{equation}
Defining the diagonal matrix $\boldsymbol{\Delta L} \coloneqq \boldsymbol{\mathrm{diag}}\Big(\Delta L_c (k_c)\Big)$, we can then re-write, similarly to (\ref{eq:U expression L&T}), the R-matrix expression (\ref{eq:U expression}) of the scattering matrix $\boldsymbol{U}$ as a function of $ \Delta L_c (k_c)$, so that:
\begin{equation}
\begin{IEEEeqnarraybox}[][c]{rcl}
 \boldsymbol{U} &  \ = \ & \boldsymbol{O}^{-1}\left[ \Id{} + \left[ \boldsymbol{\rho}^{1/2}  \boldsymbol{R}_{L} \boldsymbol{\rho}^{-1/2} \right] \boldsymbol{\Delta L}  \right] \boldsymbol{I}\\
 & \ = \ &  \boldsymbol{I} \left[ \Id{} + \boldsymbol{\Delta L} \left[ \boldsymbol{\rho}^{-1/2}  \boldsymbol{R}_{L} \boldsymbol{\rho}^{1/2} \right]   \right]  \boldsymbol{O}^{-1}
\IEEEstrut\end{IEEEeqnarraybox}
\label{eq: U as a function of Delta_L}
\end{equation}
Notice again how this expression is closely related to the analytic continuation of expression (\ref{eq:U expression L&T}).

Coming back to the Eden \& Taylor continuation (\ref{eq:: Eden and Taylor continuation}), let us now establish a relation between the Kapur-Peierls operator $\boldsymbol{R}_{L}$ and $\boldsymbol{\Delta L} $. 
From the definition (\ref{eq:Kapur-Peirels Operator and Channel-Level equivalence}) of the Kapur-Peirels operator $\boldsymbol{R}_{L}$, recalling that under Eden \& Taylor continuations (\ref{eq:: Eden and Taylor continuation}) the energy $E$ from mapping $\ref{eq:rho_c(E) mapping}$ remains unaltered, and given that the boundary-condition $B_c$ in the $\boldsymbol{L^0}$ matrix function is real and thus the R-matrix parameters (\ref{eq:R expression}) are too, it follows that:
\begin{equation}
\begin{IEEEeqnarraybox}[][c]{rcl}
  \left[\boldsymbol{R}_{L}^{-1}(\boldsymbol{\widetilde{k}} ) \right]^*&  -\boldsymbol{R}_{L}^{-1}(\boldsymbol{k})  \ = \ & \left( \begin{array}{cc}
\widehat{\boldsymbol{\Delta L}}(\boldsymbol{k}) & 0 \\
0 & 0 \\
\end{array}\right) 
\IEEEstrut\end{IEEEeqnarraybox}
\label{eq:Delta_L and R_L}
\end{equation}
where we have used the $\boldsymbol{L}$  congugacy relations (\ref{eq:L_c congugacy}) to establish that all channels $c \not\in \widehat{\mathfrak{C}}$ cancel out, and the rest yield $\Delta L_{c \in \widehat{\mathfrak{C}}}(k_c)$. The $\widehat{\boldsymbol{\Delta L}}$ thus designates the sub-matrix composed of all the channels $c \in \widehat{\mathfrak{C}}$.
Multiplying both left and right, and considering the sub-matrices on the channels $c \in \widehat{\mathfrak{C}}$ thus yields:
\begin{equation}
\begin{IEEEeqnarraybox}[][c]{rcl}
    \boldsymbol{\widehat{R}}_L(\boldsymbol{k}) - \left[\boldsymbol{\widehat{R}}_L(\boldsymbol{\widetilde{k}} ) \right]^*& \ = \ & \boldsymbol{\widehat{R}}_L(\boldsymbol{k}) \widehat{\boldsymbol{\Delta L}}(\boldsymbol{k}) \left[\boldsymbol{\widehat{R}}_L(\boldsymbol{\widetilde{k}} ) \right]^*
\IEEEstrut\end{IEEEeqnarraybox}
\label{eq:Delta_L and R_L submatrix}
\end{equation}
This relation is what guarantees the scattering matrix $\boldsymbol{U}$ satifsfies generalized unitarity condition (\ref{eq:: Eden and Taylor Generalized unitarity}).
Indeed, let us develop the left-hand side of (\ref{eq:: Eden and Taylor Generalized unitarity}), using expressions (\ref{eq: U as a function of Delta_L}) on the sub-matrices of the channels $c \in \widehat{\mathfrak{C}}$:
\begin{equation}
\begin{IEEEeqnarraybox}[][c]{l}
    \widehat{\boldsymbol{U}}(\boldsymbol{k}) \Big[ \widehat{\boldsymbol{U}}(\boldsymbol{\widetilde{k}}) \Big]^\dagger \ = \  \\ \boldsymbol{\widehat{O}}^{-1}(\boldsymbol{k})\left[ \Id{} + \widehat{\left[ \boldsymbol{\rho}^{1/2}  \boldsymbol{R}_{L} \boldsymbol{\rho}^{-1/2} \right]}(\boldsymbol{k}) \widehat{\boldsymbol{\Delta L}}(\boldsymbol{k})  \right] \boldsymbol{\widehat{I}}(\boldsymbol{k})  \\
     \  \   \times  \left[ \boldsymbol{\widehat{I}}(\boldsymbol{\widetilde{k}}) \left[ \Id{} + \widehat{\boldsymbol{\Delta L}}(\boldsymbol{\widetilde{k}}) \widehat{\left[ \boldsymbol{\rho}^{-1/2}  \boldsymbol{R}_{L} \boldsymbol{\rho}^{1/2} \right]}  (\widetilde{\boldsymbol{k}}) \right]  \boldsymbol{\widehat{O}}^{-1}(\boldsymbol{\widetilde{k}}) \right]^\dagger \\
     \ = \ \boldsymbol{\widehat{O}}^{-1}(\boldsymbol{k})\left[ \Id{} + \widehat{\left[ \boldsymbol{\rho}^{1/2}  \boldsymbol{R}_{L} \boldsymbol{\rho}^{-1/2} \right]}(\boldsymbol{k}) \widehat{\boldsymbol{\Delta L}}(\boldsymbol{k})  \right] \boldsymbol{\widehat{I}}(\boldsymbol{k}) \times \\
     \  \     \left[ \boldsymbol{\widehat{O}}^{-1}(\boldsymbol{\widehat{k}^*})\right]^*   \left[ \Id{} +  \widehat{\left[ \boldsymbol{\rho}^{-1/2}  \boldsymbol{R}_{L} \boldsymbol{\rho}^{1/2} \right]}  (\boldsymbol{\widehat{k}^*}) \left[\widehat{\boldsymbol{\Delta L}}(\boldsymbol{\widehat{k}^*})\right]^* \right] \left[\boldsymbol{\widehat{I}}(\boldsymbol{\widehat{k}^*}) \right]^* \\
\IEEEstrut\end{IEEEeqnarraybox}
\label{eq:U developments for generalized unitarity}
\end{equation}
Noticing that conjugacy relation (\ref{eq:L_c congugacy}) entail the following $\boldsymbol{\Delta L}$ symmetry from definition (\ref{eq:Delta_L def}), $\left[\widehat{\boldsymbol{\Delta L}}(\boldsymbol{\widehat{k}^*})\right]^* = -\widehat{\boldsymbol{\Delta L}}(\boldsymbol{k}) $, and making use of the conjugacy relations for the wave functions (\ref{eq:Conjugacy for O and I}), we can further simplify (\ref{eq:U developments for generalized unitarity}) to:
\begin{equation}
\begin{IEEEeqnarraybox}[][c]{l}
 \widehat{\boldsymbol{U}}(\boldsymbol{k}) \Big[ \widehat{\boldsymbol{U}}(\boldsymbol{\widetilde{k}}) \Big]^\dagger \ = \ \Id{}  \\
 \quad \quad \quad \quad + \;  \boldsymbol{\widehat{O}}^{-1}(\boldsymbol{k})   \widehat{\left[ \boldsymbol{\rho}^{1/2}  \boldsymbol{R}_{L} \boldsymbol{\rho}^{-1/2} \right]}(\boldsymbol{k}) \times \Bigg[  \\  \left[\left[ \widehat{ \boldsymbol{\rho}^{1/2}  \boldsymbol{R}_{L} \boldsymbol{\rho}^{-1/2} }\right]^{-1}(\boldsymbol{k^*})\right]^\dagger   - \left[ \widehat{ \boldsymbol{\rho}^{1/2}  \boldsymbol{R}_{L} \boldsymbol{\rho}^{-1/2} }\right]^{-1}(\boldsymbol{k})   - \widehat{\boldsymbol{\Delta L}}(\boldsymbol{k})  \\ \quad \quad \quad \quad \Bigg] 
   \times  \left[ \widehat{\left[ \boldsymbol{\rho}^{-1/2}  \boldsymbol{R}_{L} \boldsymbol{\rho}^{1/2} \right]}  (\boldsymbol{\widehat{k}^*}) \right]^\dagger \widehat{\boldsymbol{\Delta L}}(\boldsymbol{k}) \boldsymbol{\widehat{O}}(\boldsymbol{k})
\IEEEstrut\end{IEEEeqnarraybox}
\label{eq:U developed for generalized unitarity}
\end{equation}
In the middle we recognize identity (\ref{eq:Delta_L and R_L}), where the $\boldsymbol{\rho}^{\pm 1/2}$ cancel out by commuting in the diagonal matrix identity (\ref{eq:Delta_L and R_L}).
Property (\ref{eq:Delta_L and R_L}) thus annuls all non-identity terms, leaving Eden \& Taylor's generalized unitarity condition (\ref{eq:: Eden and Taylor Generalized unitarity}) satisfied.
\end{proof}

Let us also note that the proof required real boundary conditions $B_c \in \mathbb{R}$.
Thus, in R-matrix parametrization (\ref{eq:U expression}), real boundary conditions $B_c \in \mathbb{R}$ are necessary for the scattering matrix $\boldsymbol{U}$ be unitarity (and by extension generalized unitary).

Theorem \ref{theo::satisfied generalized unitarity condition} beholds a strong argument in favor of performing analytic continuation of the R-matrix operators as the physically correct way of prolonging the scattering matrix to complex wavenumbers $k_c\in\mathbb{C}$.

\subsection{\label{subsec::Closure of sub-threshold cross sections through analytic continuation} Closure of sub-threshold cross sections through analytic continuation}

We finish this article with the key question of how to close sub-threshold channels.
Analytically continuing the scattering matrix below thresholds entails it cannot be identically zero there, since this would entail it is the null function on the entire sheet of the maniforld (unicity of analytic continuation).
However, we here show that for massive particles subject to mappings (\ref{eq:rho_c massive}) or (\ref{eq:rho_c EDA}), adequate definitions and careful consideration will both make the transmission matrix evanescent sub-threshold (in a classical case of quantum tunnelling), and annul the sub-threshold cross-section --- the physically measurable quantity.

The equations linking the scattering matrix $\boldsymbol{U}$ to the cross section --- equations (1.9), (1.10) and (2.4) section VIII.1. of \cite{Lane_and_Thomas_1958} pp.291-293 --- were only derived for real positive wavenumbers. 
Yet, when performing analytic continuation of them to sub-threshold energies, the quantum tunneling effect will naturally make the transmission matrix infinitesimal on the physical sheet of mapping (\ref{eq:rho_c massive}).
Indeed, the \textit{transmission matrix}, $\boldsymbol{T}$, is defined in \cite{Lane_and_Thomas_1958} after eq. (2.3), VIII.2. p.292, as:
\begin{equation}
\begin{IEEEeqnarraybox}[][c]{rcl}
T_{cc'} & \ \coloneqq \ &  \delta_{cc'}\mathrm{e}^{2\mathrm{i}\omega_c} - U_{cc'} 
\IEEEstrut\end{IEEEeqnarraybox}
\label{eq:Lane and Thomas Transmission matrix definition}
\end{equation}
where $\omega_c$ is defined by Lane \& Thomas in eq.(2.13c) III.2.b. p.269, and used in eq.(4.5a) III.4.a. p.271 in \cite{Lane_and_Thomas_1958}, and is the difference $ \omega_c = \sigma_{\ell_c}(\eta_c) - \sigma_{0}(\eta_c)$,  where the \textit{Coulomb phase shift}, $\sigma_{\ell_c}(\eta_c) $ is defined by Ian Thompson in eq.(33.2.10) of \cite{NIST_DLMF}. 
Defining the diagonal matrix $\boldsymbol{\omega} \coloneqq \boldsymbol{\mathrm{diag}}\big( \omega_c \big)$, and using the R-matrix expression (\ref{eq:U expression}) for the scattering matrix, the Lane \& Thomas transmission matrix (\ref{eq:Lane and Thomas Transmission matrix definition}) can be expressed with R-matrix parameters as:
\begin{equation}
\begin{IEEEeqnarraybox}[][c]{rcl}
\boldsymbol{T}_{\text{L\&T}} \ \coloneqq \ -2\mathrm{i} \boldsymbol{O}^{-1}\left[\underbrace{\left(\frac{\boldsymbol{I} - \boldsymbol{O}\boldsymbol{\mathrm{e}}^{2\mathrm{i}\boldsymbol{\omega}}}{2\mathrm{i}}\right)}_{\boldsymbol{\Theta}} + \boldsymbol{\rho}^{1/2} \boldsymbol{R}_{L}  \boldsymbol{O}^{-1} \boldsymbol{\rho}^{1/2} \right]
\IEEEstrut\end{IEEEeqnarraybox}
\label{eq:Transmission matrix Lane and Thomas matrix expression}
\end{equation}
The angle-integrated partial cross sections $\sigma_{cc'}(E)$ can then be expressed as eq.(3.2d) VIII.3. p.293 of \cite{Lane_and_Thomas_1958}:
\begin{equation}
\begin{IEEEeqnarraybox}[][c]{rcl}
\sigma_{cc'}(E) & \ = \ & \pi g_{J^\pi_c} \left| \frac{ T_{\text{L\&T}}^{cc'}(E)}{k_c(E)}\right|^2
\IEEEstrut\end{IEEEeqnarraybox}
\label{eq:partial sigma_cc'Lane and Thomas}
\end{equation}
where $ g_{J^\pi_c} \ \coloneqq \ \frac{2 J + 1 }{\left(2 I_1 + 1 \right)\left(2 I_2 + 1 \right) }$ is the \textit{spin statistical factor} defined eq.(3.2c) VIII.3. p.293.
Plugging-in the transmission matrix R-matrix parametrization (\ref{eq:Transmission matrix Lane and Thomas matrix expression}) into  cross-section expression (\ref{eq:partial sigma_cc'Lane and Thomas}) then yields:
\cite{Lane_and_Thomas_1958}:
\begin{equation}
\begin{IEEEeqnarraybox}[][c]{rcl}
\sigma_{cc'} & \ = \ & 4\pi g_{J^\pi_c} \left|\frac{1}{O_c k_c} \right|^2  \left| \boldsymbol{\Theta} + \boldsymbol{\rho}^{1/2} \boldsymbol{R}_{L}  \boldsymbol{O}^{-1} \boldsymbol{\rho}^{1/2} \right|_{cc'}^2
\IEEEstrut\end{IEEEeqnarraybox}
\label{eq:partial sigma_cc'Lane and Thomas R-matrix}
\end{equation}

An alternative, more numerically stable, way of computing the cross section is used at Los Alamos National Laboratory, where one of us, G. Hale, introduced the following rotated transmission matrix, defined as:
\begin{equation}
\begin{IEEEeqnarraybox}[][c]{rcl}
\boldsymbol{T}_{\text{H}} & \ \coloneqq \ & - \frac{ \boldsymbol{\mathrm{e}}^{-\mathrm{i}\boldsymbol{\omega}} \boldsymbol{T}_{\text{L\&T}} \boldsymbol{\mathrm{e}}^{-\mathrm{i}\boldsymbol{\omega}}}{2\mathrm{i}}
\IEEEstrut\end{IEEEeqnarraybox}
\label{eq:Hale Transmission matrix}
\end{equation}
and whose R-matrix parametrization is thus
\begin{equation}
\begin{IEEEeqnarraybox}[][c]{rcl}
\boldsymbol{T}_{\text{H}}  & \ = \ & \boldsymbol{H_+}^{-1}\left[ \boldsymbol{\rho}^{1/2} \boldsymbol{R}_{L} \boldsymbol{\rho}^{1/2} \boldsymbol{H_+}^{-1}  - \underbrace{\left(\frac{\boldsymbol{H_+} - \boldsymbol{H_-}}{2\mathrm{i}}\right)}_{\boldsymbol{F}} \right]
\IEEEstrut\end{IEEEeqnarraybox}
\label{eq:Hale Transmission matrix expression}
\end{equation}
where $\boldsymbol{H_\pm}$ are defined as in eq.(2.13a)-(2.13b) III.2.b p.269 \cite{Lane_and_Thomas_1958}:
\begin{equation}
\begin{IEEEeqnarraybox}[][c]{rcl}
{H_{+}}_c & \ = \ &  O_c \mathrm{e}^{\mathrm{i} \omega_c} = G_c + \mathrm{i} F_c \\
{H_{-}}_c & \ = \ &  I_c \mathrm{e}^{-\mathrm{i} \omega_c} = G_c - \mathrm{i} F_c
\IEEEstrut\end{IEEEeqnarraybox}
\label{eq:def H_pm I and O}
\end{equation}
and for which we refer to Ian. J. Thompson's Chapter 33, eq.(33.2.11) in \cite{NIST_DLMF}, or Abramowitz \& Stegun chapter 14, p.537 \cite{Abramowitz_and_Stegun}.
The partial cross section is then directly related to the $\boldsymbol{T}_{\text{H}}$ rotated transmission matrix (\ref{eq:Hale Transmission matrix}) as:
\begin{equation}
\begin{IEEEeqnarraybox}[][c]{rcl}
\sigma_{cc'}(E) & \ = \ & 4\pi g_{J^\pi_c} \left|\frac{T_H^{cc'}(E)}{k_c(E)} \right|^2 
\IEEEstrut\end{IEEEeqnarraybox}
\label{eq:partial sigma_cc' Hale T-matrix}
\end{equation}

\begin{theorem}\label{theo::evanescence of sub-threshold transmission matrix} \textsc{Evanescence of sub-threshold transmission matrix}. \\
For massive particles, analytic continuation of R-matrix parametrization (\ref{eq:U expression}) makes the sub-threshold transmission matrix $\boldsymbol{T}$, defined as (\ref{eq:Transmission matrix Lane and Thomas matrix expression}), evanescent on the physical sheets of mappings (\ref{eq:rho_c massive}), or (\ref{eq:rho_c EDA}).
In turn, this quantum tunnelling entails the partial cross sections $\sigma_{cc'}(E)$ become infinitesimal below threshold.
\end{theorem}

\begin{proof}
The proof is based on noticing that both transmission matrix expressions (\ref{eq:Transmission matrix Lane and Thomas matrix expression}) and (\ref{eq:Hale Transmission matrix}) entail their modulus square is proportional to:
\begin{equation}
\begin{IEEEeqnarraybox}[][c]{rcl}
|\boldsymbol{T}_{cc'}|^2(E) & \ \propto \ &  \left|\frac{1}{H_+(E)} \right|^2 
\IEEEstrut\end{IEEEeqnarraybox}
\label{eq:T_cc' proportional to H}
\end{equation}
This is because $\boldsymbol{R}_{L}  \boldsymbol{O}^{-1} = \left[ \boldsymbol{O}\left[ \boldsymbol{R}^{-1} - \boldsymbol{B}\right] - \boldsymbol{\rho}\boldsymbol{O}^{(1)}\right]^{-1}$, which does not diverge below threshold. 
Asymptotic expressions for the behavior of $H_+(\rho)$ then yield, for small $\rho$ values:
\begin{equation}
\begin{IEEEeqnarraybox}[][c]{rcl}
H_{+}(\rho) & \ \underset{\rho \to 0}{\sim} \ &  \frac{\rho^{-\ell}}{(2\ell + 1)C_\ell(\eta)} - \mathrm{i} C_\ell(\eta)\rho^{\ell+1}
\IEEEstrut\end{IEEEeqnarraybox}
\label{eq:H+ near zero}
\end{equation}
and asymptotic large-$\rho$ behavior: 
\begin{equation}
\begin{IEEEeqnarraybox}[][c]{rcl}
H_{+}(\rho) & \ \underset{\rho \to \infty} {\sim}\ & \mathrm{e}^{\mathrm{i} (\rho - \eta \ln(2\rho) - \frac{1}{2}\ell\pi + \sigma_\ell(\eta) ) }
\IEEEstrut\end{IEEEeqnarraybox}
\label{eq:H+ near infinity}
\end{equation}
Above the threshold, $\rho \in \mathbb{R}$ is real and thus equation (\ref{eq:H+ near infinity}) shows how $| H_{+}(\rho) | \ \underset{\rho \to \infty} {\longrightarrow} 1  $. In other terms, the $| H_{+}(\rho) |$ term cancels out of the cross section expressions (\ref{eq:partial sigma_cc'Lane and Thomas R-matrix}) and (\ref{eq:partial sigma_cc' Hale T-matrix}) for open-channels above threshold. 
Yet, from mappings (\ref{eq:rho_c massive}) and (\ref{eq:rho_c EDA}), the sub-threshold dimensionless wavenumber is purely imaginary: $\rho \in \mathrm{i}\mathbb{R}$. 
Since asymptotic form (\ref{eq:H+ near infinity}) is dominated in modulus by: $\left|H_{+}(\rho)\right|  \underset{\rho \to \infty} {\sim} \left| \mathrm{e}^{\mathrm{i} \rho }\right|$.
Depending on which sheet $\rho$ is continued sub-threshold, we can have $\rho = \pm \mathrm{i} x$, with $x\in \mathbb{R}$. Thus, on the non-physical sheet $\big\{ E, \hdots, -_c , \hdots \big\}$ for the given channel $c$ of $\rho_c$, the transmission matrix (\ref{eq:T_cc' proportional to H}) experiences exponential decay of $1/\left|H_{+}(\rho)\right|$ leading to the evanescence of the cross section (\ref{eq:partial sigma_cc'Lane and Thomas}), or (\ref{eq:partial sigma_cc' Hale T-matrix}). 
In effect, this means that the $|O_c(\rho_c)|$ term in (\ref{eq:partial sigma_cc'Lane and Thomas R-matrix}) asymptotically acts like a Heaviside function, being unity for open channels, but closing the channels below threshold. 
Since $\rho_c = k_c r_c$ for the outgoing scattered wave $O_c(\rho_c)$, the exponential closure depends on two factors: the distance $r_c$ from the nucleus, and how far from the threshold one is $|E - E_{T_c}|$. This is a classical evanescence behavior of quantum tunneling. 

What happens when continuing on the physical sheet $\big\{ E, \hdots, +_c , \hdots \big\}$, as $\left|H_{+}(\rho)\right|$ will now tend to a $1/0$ diversion? The authors have no rigorous answer, but point to the fact that since $E$ is left unchanged by the choice of the $k_c$ sheet, evanescence result ought to also stand, despite the apparent divergence. 

Note that for photon channels, the semi-classic mapping (\ref{eq:rho_c photon}) does not yield this behavior, only the relativistic mapping (\ref{eq:rho_c EDA}) does. 

\end{proof}

We can estimate the orders of magnitude required to experimentally observe this evanescent quantum tunneling closure of the cross sections below threshold.
At distance $r_c$ from the center of mass of the nucleus, and at wavenumber $k_c$, distant from the threshold as $|E-E_{T_c}|$, the asymptotic behavior or the cross-section below threshold is:
\begin{equation}
\begin{IEEEeqnarraybox}[][c]{rcl}
\ln\Big( \sigma_{cc'}(k_c,r_c) \Big) & \ \underset{\begin{array}{c}
   E_c \leq E_{T_c} \\
   k_c \to -\infty
\end{array}  } {\sim}\ & - 2 r_c |k_c| 
\IEEEstrut\end{IEEEeqnarraybox}
\label{eq:sub-threshold cross-section evanescence}
\end{equation}

Assuming a detector is placed at a distance $r_c$ of the nucleus, the cross section would decay exponentially below threshold as the distance $\Delta E_c = | E - E_{T_c}|$ of $E$ to the threshold $E_{T_c}$ increases. For instance, for a threshold of $^{238}$U target reacting with neutron $n$ channel, evanescence (\ref{eq:sub-threshold cross-section evanescence}) would be of the rate of $\log_{10}\Big( \sigma_{cc'}(k_c,r_c) \Big)  \sim  -3\times 10^{16} {r_c}_{\text{m}} \sqrt{{\Delta E_c}_{\text{eV}}} $. 
For a detector placed at a millimeter $r_c \sim 10^{-3} \mathrm{m}$, this means one order of magnitude is lost for the cross section in $\Delta E_c \sim 10^{-27} \text{eV}$, evanescent indeed. Conversely, detecting this quantum tunneling with a detector sensitive to micro-electonvolts $\Delta E_c \sim 10^{-6}\text{eV} \sim 1 \mu \text{eV}$ (200 times more sensitive than the thermal energy of the cosmic microwave background) would see the cross section drop of one order of magnitude for a move of less than $10^{-13}\text{m}$, or a tenth of a pico-meter. We are at sub-atomic level of quantum tunneling: the outgoing wave evanesces into oblivion way before reaching the electron cloud...

Regardless of the evanescence of the transmission matrix, a more general argument on the cross section shows that analytic continuation of the above-threshold expressions will automatically close the channels bellow the threshold. 

\begin{theorem}\label{theo::Analytic continuation annuls sub-threshold cross sections} \textsc{Analytic continuation annuls sub-threshold cross sections}.\\
For massive particles, analytic continuation of above-threshold cross-section expressions to complex wavenumbers $k_c \in \mathbb{C}$ will automatically close sub-threshold channels.
\end{theorem}

\begin{proof}
The proof is based on the fact that massive particles are subject to mappings (\ref{eq:rho_c massive}), or (\ref{eq:rho_c EDA}) for relativistic correction, which entail the wavenumbers are real above threshold, and exactly imaginary sub-threshold: $\forall E < E_{T_c} , k_c \in \mathrm{i}\mathbb{R}$.
Let $\psi(\vec{r})$ be a general wave function, so that the probability density is $\left| \psi \right|^2(\vec{r})$.

For a massive particle subject to a real potential, the de Broglie non-relativistic Schr\"odinger equation applies, so that writing the conservation of probability on a control volume, and applying the Green-Ostrogradky theorem, will yield the following expression for the probability current vector: 
\begin{equation}
\begin{IEEEeqnarraybox}[][c]{rcl}
\vec{j}_\psi  & \; \coloneqq \, \frac{\hbar}{\mu}\Im\left[ \psi^* \vec{\nabla} \psi \right]
\IEEEstrut\end{IEEEeqnarraybox}
\label{eq: probability current vector}
\end{equation}
where $\mu$ is the reduced mass of the two-particle system (c.f. equations (2.10) and (2.12) section VIII.2.A, p.312 in \cite{Blatt_and_Weisskopf_Theoretical_Nuclear_Physics_1952}).
By definition, the differential cross section $\frac{\mathrm{d}\sigma_{cc'}}{\mathrm{d}\Omega}$ is the ratio of the outgoing current in channel $c'$ by the incoming current from channel $c$, by unit of solid angle $\mathrm{d}\Omega$.

Consider the incomming channel $c$, classically modeled as a plane wave, $\psi_c(\vec{r}_c) \propto \mathrm{e}^{\mathrm{i}\vec{k}_c\cdot\vec{r}_c}$; and the outgoing channel $c'$, classically modeled as radial wave, $\psi_{c'}(r_{c'}) \propto \frac{ \mathrm{e}^{\mathrm{i}k_{c'} r_{c'}}}{r_{c'}}$.
For arbitrary complex wavenumbers, $k_c , k_{c'}  \in \mathbb{C}$, definition (\ref{eq: probability current vector}) will yield the following probability currents respectively:
\begin{equation}
\begin{IEEEeqnarraybox}[][c]{rcl}
\vec{j}_{\psi_c}  & \propto & \frac{\hbar}{\mu}\Im\left[\mathrm{i}\vec{k}_c \mathrm{e}^{-2\Im\left[\vec{k}_c\right]\cdot\vec{r}_c} \right]   \\ \vec{j}_{\psi_{c'}} & \propto & \frac{\hbar}{\mu}\Im\left[\left(\mathrm{i}k_{c'} - \frac{1}{r_{c'}} \right)\frac{ \mathrm{e}^{-2\Im\left[k_{c'}\right]\cdot r_c}}{r_{c'}^2} \right]\vec{e}_{r}  
\IEEEstrut\end{IEEEeqnarraybox}
\label{eq: probability current vectors for complex k}
\end{equation}
One will note these expressions are not the imaginary part of an analytic function in the wavenumber, because of the imaginary part $\Im\left[k_c\right]$. 
If however we look at real wavenumbers $k_c , k_{c'}  \in \mathbb{R}$, that is at above-threshold energies $E \geq E_{T_c}$, the probability currents (\ref{eq: probability current vectors for complex k}) readily simplify to:
\begin{equation}
\begin{IEEEeqnarraybox}[][c]{rcl}
\vec{j}_{\psi_c}   \propto  \frac{\hbar}{\mu}\Re\left[\vec{k}_c \right]   & \quad , \quad  & \vec{j}_{\psi_{c'}}  \propto  \frac{\hbar}{\mu}\Re\left[k_{c'} \right]\vec{e}_{r}  
\IEEEstrut\end{IEEEeqnarraybox}
\label{eq: probability current vectors for real k}
\end{equation}
These expressions are the real part of analytic functions of the wavenumbers.
If we analytically continue them to complex wavenumbers, and consider the cases of subthreshold reactions $E < E_{T_c}$, for either the incoming or the outgoing channel, the wavenumbers are then exactly imaginary, $k_c , k_{c'}  \in \mathrm{i}\mathbb{R}$.
The real parts in (\ref{eq: probability current vectors for real k}) become zero, thereby annulling the cross section $\sigma_{c,c'}(E)$.
This means that for massive particles subject to real potentials, analytic continuation of the probability currents expressions above threshold (\ref{eq: probability current vectors for real k}) will automatically close the sub-threshold channels.
This is true regardless of whether the transmission matrix (\ref{eq:Lane and Thomas Transmission matrix definition}) is evanescent or nor below threshold. 
This constitutes another major argument in favor of analytic continuation of open-channels expressions to describe the closed channels. 
\end{proof}

Note that our proof does not stand for photon channels. 
For photon channels, the derivations for the probability current vector (\ref{eq: probability current vector}) do not stand, and the wavenumber $k_c$ is not imaginary below threshold using mapping (\ref{eq:rho_c photon}), nor using the relativistic-correction (\ref{eq:rho_c EDA}). 
The fundamental reason why photon treatment is not straightforward is that R-matrix theory was constructed on the semi-classical formalism of quantum physics, with wave functions instead of state vectors. Though not incorrect, this wave function approach of quantum mechanics does not translate directly for the photons, though some works have been done to describe photons through wave functions \cite{Bialynicki-Birula_Photon_wave_function_1994, Bialynicki-Birula_Photon_wave_function_1996}. 
This is another open area in the field of R-matrix theory, beyond the scope of this article.

\section{\label{sec:Conclusion}Conclusion}

In this article, we conduct a study and establish novel properties of three alternative parametrizations of the scattering matrix in R-matrix theory: the Wigner-Eisenbud parameters, the Brune parameters, and the Siegert-Humblet parameters. We link these parametrizations to the Humblet-Rosenfelt complex-pole expansion of the scattering matrix, and show that, in general, these parametrizations mark a trade-off between the complexity of the parameters and the complexity of the energy dependence of the scattering matrix.

The Wigner-Eisenbud parameters are the poles $\big\{E_\lambda\big\}$ and residue widths $\big\{\gamma_{\lambda,c}\big\}$ of the $\boldsymbol{R}$ matrix (\ref{eq:R expression}). They are $N_\lambda \in \mathbb{N}$, simple, real poles, which are independent from one another (meaning any choice of real parameters are physically acceptable), and de-entangle the energy dependence of the $\boldsymbol{R}$ matrix from the branch-points the thresholds $\left\{E_{T_c}\right\}$ introduce in the multi-sheeted Riemann surface of mapping (\ref{eq:rho_c(E) mapping}).
Both $\big\{E_\lambda\big\}$ and $\big\{\gamma_{\lambda,c}\big\}$ are dependent on both the channel radii $\big\{a_c \big\}$ and the boundary conditions $\big\{ B_c\big\}$.
The set of Wigner-Eisenbud parameters $\Big\{ E_{T_c}, a_c, B_c, E_{\lambda}, \gamma_{\lambda,c} \Big\}$ is sufficient to entirely determine the energy behavior of the scattering matrix $\boldsymbol{U}$ through (\ref{eq:U expression}).

The Brune parameters are the poles $\left\{\widetilde{E_i}\right\}$ of the $\boldsymbol{R}_S$ matrix (\ref{eq:R_S by Brune det search}) and the widths $\left\{\widetilde{\gamma_{i,c}}\right\}$, transformed by (\ref{eq:Brune parameters}) from the residue widths of the physical level matrix $\boldsymbol{\widetilde{A}}$ in (\ref{eq::Brune physical level matrix}) and (\ref{eq:Brune eigenproblem}). They are $N_S \geq N_\lambda$ simple, real poles, and are intimately interdependent in that not any set of real parameters is physically acceptable (they must be solutions to  (\ref{eq:Brune eigenproblem}) ).
If definition (\ref{eq:: Def S = Re[L], P = Im[L]}) is chosen for the shift function $\boldsymbol{S}$, the Brune parameters live on the multi-sheeted Riemann surface of mapping (\ref{eq:rho_c(E) mapping}): they have shadow poles $\left\{\widetilde{E_i}\right\}$ on the unphysical sheets $\left\{ E, -\right\}$ below threshold $ E < E_{T_c}$, though there are only $N_\lambda$ real poles on the physical sheet (theorem \ref{theo::shadow_Brune_poles}). If definition (\ref{eq:: Def S analytical}) is chosen, then the shift factor $\boldsymbol{S}$ is a function of $\rho_c^2$, which unfolds the sheets in mapping (\ref{eq:rho_c(E) mapping}): there are then $N_S\geq N_\lambda$ real poles $\left\{\widetilde{E_i}\right\}$, all living on the same sheet with no branch points (theorem \ref{theo::analytic_Brune_poles}).
Both $\left\{\widetilde{E_i}\right\}$ and $\left\{\widetilde{\gamma_{i,c}}\right\}$ are invariant to change in boundary conditions $\big\{B_c\big\}$, though both depend on the channel radii $\big\{a_c\big\}$.
Any set of $N_\lambda$ Brune parameters $\Big\{ E_{T_c}, a_c, \widetilde{E_{i}}, \widetilde{\gamma_{i,c}} \Big\}$ is sufficient to entirely determine the energy behavior of the scattering matrix $\boldsymbol{U}$ through (\ref{eq:R_L unchanged by Brune}) and (\ref{eq:U expression}) (theorem \ref{theo::Choice of Brune poles}).

The Siegert-Humblet parameters are the poles $\big\{\mathcal{E}_j\big\}$ and residue widths $\left\{r_{j,c}\right\}$ of the Kapur-Peierls $\boldsymbol{R}_{L}$ operator (\ref{eq:Kapur-Peirels Operator and Channel-Level equivalence}). They are $N_L \geq N_\lambda$ complex, (almost always) simple poles, that reside on the Riemann surface of mapping (\ref{eq:rho_c(E) mapping}), comprised of $2^{N_c}$ branches, and for which one must specify on which sheet they reside, as in (\ref{eq:: pole E_j sheet reporting}). They are intimately interwoven in that not any set of complex parameters is physically acceptable (they must be solution to (\ref{eq:R_L eigenproblem})).
Both $\big\{\mathcal{E}_j\big\}$ and $\big\{r_{j,c}\big\}$ are invariant to change in boundary conditions $\left\{B_c\right\}$. Furthermore, $\big\{\mathcal{E}_j\big\}$ is invariant to a change in channel radii $\left\{a_c\right\}$, and we established in theorem \ref{theo::r_j,c uncer change of a_c} a simple way of transforming the radioactive widths $\big\{r_{j,c}\big\}$ under a change of channel radius $a_c$.
Since the Siegert-Humblet parameters are the poles and residues of the local Mittag-Leffler expansion (\ref{eq::RL Mittag Leffler}) of the Kapur-Peierls matrix $\boldsymbol{R}_{L}$, the set of Siegert-Humblet parameters $\Big\{ E_{T_c}, a_c, \mathcal{E}_j, r_{i,c} \Big\}$ is insufficient to entirely determine the energy behavior of the scattering matrix $\boldsymbol{U}$ through (\ref{eq::u_j scattering residue width}) and (\ref{eq::U Mittag Leffler}). The latter expressions directly link the R-matrix parameters to the poles and residues of the Humblet-Rosenfeld expansion of the scattering matrix, and can be complemented by local coefficients $\left\{\boldsymbol{s}_n\right\}_{\mathcal{V}(E)}$ of the entire part (\ref{eq::Hol_U expansion}), to untangle the energy dependence of the scattering matrix into the simple sum of poles and residues (\ref{eq::U Mittag Leffler}), which is the full Humblet-Rosenfeld expansion of the scattering matrix. Theorem \ref{theo:: poles of U are Siegert-Humblet poles} establishes that under analytic continuation of the R-matrix operators, the poles of the Kapur-Peierls $\boldsymbol{R}_{L}$ operator (i.e. the Siegert-Humblet poles) are exactly the poles of the scattering matrix $\boldsymbol{U}$.
The latter result constitues one of the three arguments we here advance to argue that, contrary to the legacy force-closure of sub-threshold channels presented in Lane \& Thomas \cite{Lane_and_Thomas_1958}, the scattering matrix $\boldsymbol{U}$ ought to be analytically continued for complex momenta.
Such analytic continuation is necessary to cancel the spurious poles which would otherwise be introduced by the outgoing wavefunctions, as we establish in theorem \ref{theo::Analytic continuation of scattering matrix cancels spurious poles}.
Moreover, we show in theorem \ref{theo::satisfied generalized unitarity condition} that analytic continuation of the scattering matrix R-matrix parametrization (\ref{eq:U expression}) verifies Eden \& Taylor's generalized unitarity condition (\ref{eq:: Eden and Taylor Generalized unitarity}), in the wake showing that real boundary conditions $B_c\in\mathbb{R}$ are also necesary for unitarity.
Finally, we argue in theorems \ref{theo::evanescence of sub-threshold transmission matrix} and \ref{theo::Analytic continuation annuls sub-threshold cross sections} that analytic continuation will close cross sections for massive particle channels below threshold.

\begin{acknowledgments}
This work was partly funded by the Los Alamos National Laboratory (summer 2017 internship in T-2 division with G. Hale and M. Paris), as well as by the Consortium for Advanced Simulation of Light Water Reactors (CASL), an Energy Innovation Hub for Modeling and Simulation of Nuclear Reactors under U.S. Department of Energy Contract No. DE-AC05-00OR22725.

The authors are profoundly grateful to Prof. Semyon Dyatlov, from MIT and U.C. Berkeley, for his critical contribution in pointing to us towards the Gohberg-Sigal theory and providing important insights within it.
Great thanks and long-lasting friendship to Yoann Desmouceaux, author of the proof of diagonal divisibility and capped multiplicities lemma \ref{lem::diagonal divisibility and capped multiplicities}, who's help and inputs were key on technical algebraic points. 
Our genuine gratitude to Dr. Andrew Holcomb, from Oak-Ridge National Laboratory, who helped us numerically test the veracity of Mittag-Leffler expansion (\ref{eq::Explicit Mittag-Leffler expansion of L_c}).
They would also like to thank Prof. Javier Sesma, from Universidad de Zaragoza, for his invaluable guidance on the properties of the Hankel functions; as well Dr. H. Owusu for his time in discussing Hamiltonian degeneracy.
Finally, the first author was invaluably supported by the second author, Dr. Vladimir Sobes, in whom he found a lifelong friend, and is profoundly grateful to Dr. Mark Paris, of Los Alamos National Laboratory: the R-matrix 2016 summer workshop he organized in Santa-Fe was genuinely catalytic to these findings.
\end{acknowledgments}

\bibliography{citations_PhysRevC}
\end{document}